\documentclass[11pt]{article}

\RequirePackage{amsthm,amsmath,amsfonts,amssymb,bbm} 
\RequirePackage[numbers]{natbib}
\RequirePackage[colorlinks,citecolor=blue,urlcolor=blue]{hyperref}
\RequirePackage{graphicx}

\theoremstyle{plain}

\newtheorem{theorem}{Theorem}[section]

\numberwithin{equation}{section}
\newtheorem{lemma}{Lemma}[section]

\theoremstyle{definition}
\newtheorem{definition}[theorem]{Definition} 
\newtheorem{assumption}{Assumption} 
\newtheorem{example}{Example} 
\newtheorem{remark}{Remark} 
\newcommand{\calE}[0]{\mathcal{E}}
\def\P{\mathbb{P}}
\def\B{\mathbb{B}} 
\def\R{\mathbb{R}}
\def\X{\mathcal{X}}
\def\E{\mathbb{E}}
  
\def\L{\mathcal{L}}  
\def\A{\mathcal{I}} 
\def\id{\mathbbm{1}}
\def\I{{\bf I}}

\newcommand{\indep}{\perp \!\!\!\! \perp}

\DeclareMathOperator{\argmin}{argmin}

\title{Approximate co-sufficient sampling with regularization}
\author{Wanrong Zhu\thanks{Department of Statistics, University of California, Irvine} \and Rina Foygel Barber\thanks{Department of Statistics, University of Chicago}}
\date{}

\begin{document} 
	\maketitle

\begin{abstract}
 In this work, we consider the problem of goodness-of-fit (GoF) testing for parametric models. 
This testing problem involves a composite null hypothesis, due to the unknown values of the model parameters.
In some special cases, \emph{co-sufficient sampling} (CSS) can remove the influence of these unknown parameters
via conditioning on a sufficient statistic.
However, many common parametric settings do not permit this approach, since 
conditioning on a sufficient statistic leads to a powerless test.
The recent \emph{approximate co-sufficient sampling} (aCSS) framework of \cite{barber2020testing}
offers an alternative, replacing sufficiency with an approximately sufficient statistic (namely, a noisy version of the maximum likelihood estimator (MLE)).
This approach recovers power in a range of settings where CSS cannot be applied, but can only be applied
in settings where the unconstrained MLE is well-defined and well-behaved, which implicitly assumes
a low-dimensional regime. In this work, we extend aCSS to the setting of constrained and penalized
MLE, so that more complex estimation problems can now be handled within the aCSS framework,  including examples such as mixtures-of-Gaussians (where the unconstrained MLE is not well-defined due to degeneracy) and high-dimensional Gaussian linear models (where the MLE can perform well under regularization, such as an $\ell_1$ penalty or a shape constraint).
\end{abstract}


	\section{Introduction}
Goodness-of-fit (GoF) testing is an essential statistical method, widely used in various fields such as biology, economics, engineering, and finance, to assess whether the observed data follows a certain pattern or distribution that is expected based on theoretical assumptions.  
Given data $X$ belonging to some sample space $\X$, the fundamental problem addressed
by  GoF is the question of testing the  null hypothesis
\begin{equation}\label{eqn:GoF_H0}
	H_{0}: X\sim P_{\theta} \textnormal{ for some } \theta\in\Theta,
\end{equation}
where $\{P_{\theta}: \theta\in\Theta\subseteq\R^{d}\}$ is a parametric family, versus a more complex (usually higher-dimensional) model.
For example, we may be interested in testing whether a logistic regression model is appropriate for our binary data $X=(X_1,\dots,X_n)$ 
(in the presence of some covariates), or whether a more complex---perhaps even nonparametric---model is needed.

As for any standard hypothesis testing problem, our approach to GoF testing involves two core ingredients: finding a test statistic
that captures the important trends in the data (with the convention that
large values of $T=T(X)$ indicate evidence against $H_0$), and deriving the null distribution of this test statistic $T(X)$
so that we can appropriately calibrate our test to make sure we do not exceed the allowable Type I error level.
In many settings, this second component often poses the larger challenge; 
it is often the case that the null distribution of $T(X)$ cannot be computed exactly or even estimated accurately.
An alternative approach, common in many statistical problems, is to mimic this null distribution with some form of resampling---e.g.,
methods based on permutations, on bootstrapping, or on knockoffs \citep{barber2015controlling, barber2020robust,beran1988prepivoting, berrett2020conditional, candes2018panning,  davidson2007improving, efron1979bootstrap, ernst2004permutation, lehmann1986testing, welch1990construction,  wu1986jackknife}  all have this flavor.
 For more literature on testing goodness-of-fit, especially in high-dimensional settings, we refer to Section \ref{sec:related_work} for an overview. 

At a high level, we can consider sampling \emph{copies} of the observed data, $\tilde{X}^{(1)}, ..., \tilde{X}^{(M)}$,
and using the empirical distribution of the statistic, given by the corresponding values $T(\tilde{X}^{(1)}), ..., T(\tilde{X}^{(M)})$,
as a null distribution against which we compare the evidence $T(X)$. More concretely, 
given these sampled copies, we can define a p-value corresponding to the observed evidence $T(X)$ as
\begin{equation}\label{eq:pval}
	\textnormal{pval} = \textnormal{pval}_{T}(X, \tilde{X}^{(1)}, ..., \tilde{X}^{(M)}) = \frac{1}{M+1}\left(1 + \sum_{m=1}^{M}\mathbbm{1}\left\{T(\tilde{X}^{(m)})\ge T(X)\right\}\right).
\end{equation}
If it holds that the real data and its copies $X, \tilde{X}^{(1)}, ..., \tilde{X}^{(M)}$ are  exchangeable under the null, then 
it follows immediately that this p-value is valid under the null, $\P_{H_0}(\textnormal{pval} \le\alpha)\le\alpha$ (for any rejection threshold $\alpha$).
The core challenge for this type of approach is therefore reduced to the following question:
\begin{quote}
	\normalsize
	How can we generate copies $\tilde{X}^{(1)}, ..., \tilde{X}^{(M)}$ of the observed data $X$ such that, if $H_0$ is true, then $X, \tilde{X}^{(1)}, ..., \tilde{X}^{(M)}$ are (approximately) exchangeable?
\end{quote}
Now we consider this question specifically for the GoF testing problem. Of course, in the case that $\Theta = \{\theta_0\}$ is
a singleton set, the problem is trivial---we can simply draw the $\tilde{X}^{(m)}$'s from the known null distribution $P_{\theta_0}$,
so that $X, \tilde{X}^{(1)}, ..., \tilde{X}^{(M)}$ are i.i.d.\ (and thus, exchangeable). Beyond this trivial case, however,
this simple strategy can no longer be used. For example, drawing $\tilde{X}^{(m)}$'s from $P_{\hat\theta}$ for a plug-in estimate
$\hat\theta$, which is often called the \emph{parametric bootstrap} \citep{ efron2012bayesian,  efron1994introduction, hall2006parametric,singh1981asymptotic}, may work well in some settings but has the potential to substantially inflate the Type I error rate \citep[Section 1]{barber2020testing}. The co-sufficient sampling (CSS) and approximate co-sufficient sampling (aCSS) approaches, which we will describe in detail below, avoid this issue by conditioning on a sufficient (or approximately sufficient) statistic for the unknown $\theta$. 
aCSS in particular can be applied to a range of models, but is not suited for addressing challenges such as high dimensionality.

\paragraph*{Our contributions}
In this paper, our aim is to extend the aCSS approach to the setting where $\theta$ cannot be estimated via unconstrained
maximum likelihood estimation---for example, a high-dimensional sparse linear regression problem, where unconstrained estimation
is not consistent but adding $\ell_1$ regularization restores consistency.
We develop a form of aCSS that is able to handle constrained maximum likelihood estimation (and will also extend
to the penalized case). Consequently, this new approach allows for aCSS to accommodate more robust and accurate parameter estimation in complex problems, particularly in high-dimensional settings.

\subsection{Notation and organization} 
For an integer $n\geq 1$, $[n]$ denotes the set $\{1,\dots,n\}$. We will write $\|\cdot\|$ to denote the usual Euclidean norm on vectors, and the operator norm on matrices. Furthermore,  for a vector $v$, $\|v\|_{0}$ denotes the $\ell_{0}$ norm (the number of nonzero entries), and $\|v\|_q$ denotes the usual $\ell_q$ norm for $1\leq q \leq \infty$.
For a matrix $M$, $\lambda_{\max}(M)$ and $\lambda_{\min}(M)$ denotes its largest and smallest eigenvalues.
We write $\E_{\theta}$ and $\P_{\theta}$ to denote expectation or probability taken with respect to 
the distribution $P_\theta$.   Moreover, we denote an open  ball centered at $\theta$ with radius $r$ as $\mathbb{B}(\theta, r)$, and use $(x)_{+}$ to denote $max\{x, 0\}$.

The remainder of this paper is organized as follows. We begin by providing an overview
of CSS and aCSS in Section \ref{sec:background}. In Section \ref{sec:constrained-aCSS}, we present our proposed method,
the  constrained aCSS procedure. In Section \ref{sec:theory}, we discuss the theoretical guarantees for constrained aCSS
in a range of different settings. In Section \ref{sec:penalty}, we extend our method
and theoretical results to the case of penalized, rather than constrained, maximum likelihood 
estimation, for the special case of an $\ell_1$ penalty.  Finally, we show empirical results in Section \ref{sec:numerical} to demonstrate the performance of our method, and conclude with a brief discussion in Section~\ref{sec:discussion}.
All proofs are deferred to the Appendix.

 \section{Background: goodness-of-fit testing via CSS and aCSS}\label{sec:background}
 
 
 We now focus on addressing the sampling problem introduced above. Specifically, given the null hypothesis $H_0$~\eqref{eqn:GoF_H0} that the data $X$ is drawn from $P_\theta$, for some (unknown) $\theta\in\Theta$,  our goal is to generate sample copies $\tilde{X}^{(1)},\dots,\tilde{X}^{(M)}$ that are approximately exchangeable with the observed data $X$ under the null $H_0$, so that we can then assess $T(X)$ via the p-value defined in~\eqref{eq:pval} above.  
 Of course, we can trivially achieve exchangeability 
 by simply taking $\tilde{X}^{(m)}=X$ for each copy $m$---but this would lead to zero power for testing any alternative,
 since the p-value defined in~\eqref{eq:pval} would be equal to 1 regardless of the choice of test statistic.
 
 In the remainder of this section, we will give background on the CSS and aCSS methods for producing these copies, the $\tilde{X}^{(m)}$'s,
 along with some examples to illustrate the types of settings where these methods may be applied.
 From this point on, we will write $\theta_0\in\Theta$ to denote the unknown true value of the parameter.
 
 \subsection{Co-sufficient sampling (CSS)}\label{sec:CSS}
 We cannot sample the copies $\tilde{X}^{(m)}$ from the distribution $P_{\theta_0}$ of the data
 $X$,  because of its dependence on the unknown $\theta_0$. To remove this dependence we can condition on a \emph{sufficient statistic} $S(X)$.
 To be precise, $S(X)$ is a sufficient statistic if the conditional distribution of $X$ no longer depends on $\theta$---that is,
 we can construct a conditional distribution $P(X\mid S)$ such that, for any $\theta\in\Theta$, 
 \[\textnormal{If $X\sim P_\theta$, then $X\mid S(X)$ has distribution $P(\cdot \mid S(X))$.}\]
 Co-sufficient sampling (see, e.g., \cite{agresti1992survey,engen1997stochastic,stephens2012goodness})  leverages this property to sample the copies:
 \[\textnormal{CSS method: after observing $X$, sample $\tilde{X}^{(1)},\dots,\tilde{X}^{(M)}$ i.i.d.\ from $P(\cdot \mid S(X))$.}\]
 By construction, $X,\tilde{X}^{(1)},\dots,\tilde{X}^{(M)}$ are exchangeable when $X\sim P_{\theta}$, for \emph{any} $\theta$---and thus,
 the p-value constructed in~\eqref{eq:pval} is valid under the null $H_0$~\eqref{eqn:GoF_H0}.
 
 As a concrete example, suppose that $X=(X_1,\dots,X_n)$ follows a Gaussian linear model,
 \[X \sim \mathcal{N}(Z\theta , \nu^2\I_n),\]
 for known covariates $Z \in\R^{n\times d}$ (assumed to have full column rank), known variance $\nu^2>0$, and unknown coefficients $\theta\in\Theta=\R^d$.
 Then $S(X) = Z^\top  X$ is a sufficient statistic for this parametric family,  and we can calculate the conditional distribution
 \[X\mid S(X) \sim \mathcal{N}(Z(Z^\top  Z)^{-1} S(X), \nu^2\mathcal{P}_Z^\perp),\]
 where $\mathcal{P}_Z^\perp\in\R^{d\times d}$ is the projection matrix for the subspace orthogonal to the column span of $Z$. As long as $d<n$,
 then, the copies $\tilde{X}^{(m)}$ are distinct from $X$ (and from each other), and we may be able to achieve high power
 under a suitable alternative hypothesis. Additional background and discussion of CSS can be found in \citep[Section 1]{barber2020testing}.

 \subsection{Approximate co-sufficient sampling (aCSS)}
 While the CSS method performs well for certain goodness-of-fit problems, there are many settings
 where CSS leads to a degenerate method and consequently zero power. \cite{barber2020testing} consider the example of logistic regression:
 suppose $X=(X_1,\dots,X_n)$ follows a logistic regression model, where
 \[X_i \sim\textnormal{Bernoulli}(1 / (1 + e^{-Z_i^\top \theta}))\]
 independently for each $i\in[n]$, where again $Z_1,\dots,Z_n\in\R^d$ are known covariate vectors, while $\theta\in\Theta=\R^d$ is unknown.
 In this case, for  generic values of the $Z_i$'s (for instance, if these covariates are drawn from some continuous distribution), the minimal
 sufficient statistic $S(X) = Z^\top  X$ uniquely determines $X$ ($Z\in\R^{n\times d}$ is the matrix with rows $Z_{i}$)---that is, the conditional distribution of $X\mid S(X)$ is simply a point mass.
 Consequently, applying CSS to this problem would lead to zero power since we would have $X =\tilde{X}^{(1)}=\dots=\tilde{X}^{(M)}$.
 
 To address this type of degenerate scenario, \cite{barber2020testing} propose \emph{approximate co-sufficient sampling} (aCSS). 
 The idea of aCSS is to condition on less information (to restore power), while
 ensuring that the sampled copies are approximately exchangeable (to retain Type I error control).   
 (We refer the reader to \cite[Section 1]{barber2020testing} for a more comprehensive discussion on the comparison between bootstrap, CSS, and aCSS methods.)
 
 Concretely, consider an approximate maximum likelihood estimator,
 \[\hat\theta = \hat\theta(X,W) = \argmin_{\theta\in\Theta}\left\{ - \log f(X;\theta) + R(\theta) + \sigma W^\top  \theta\right\},\]
 where $f(\cdot;\theta)$ is the density for distribution $P_\theta$ (with respect to some base measure), $R(\theta)$ is an optional 
 twice-differentiable regularizer (e.g., a ridge penalty), $W\sim\mathcal{N}(0,\frac{1}{d}\I_d)$ is Gaussian noise that 
 adds a perturbation to the maximum likelihood estimation problem, and $\sigma>0$ is a  parameter that controls the magnitude of this perturbation.
 For each $\theta\in\Theta$, define $P_\theta(\cdot \mid\hat\theta)$ as the conditional distribution of $X\mid \hat\theta$,
 when $X\sim P_\theta$ and $\hat\theta = \hat\theta(X,W)$ is defined as above. 
 
 Now we return to the GoF problem, where $X\sim P_{\theta_0}$ for an unknown $\theta_0$.
 Note that, even if the unperturbed MLE were a sufficient statistic (as would be the case for a Gaussian linear model, for example),
 the perturbed MLE $\hat\theta$ is no longer a sufficient statistic in the exact sense, and so the conditional distribution 
 $P_{\theta_0}(\cdot \mid\hat\theta)$ does depend on the unknown parameter $\theta_0$. However, it turns out that $\hat\theta$
 is approximately sufficient, meaning that $P_{\theta_0}(\cdot \mid\hat\theta)$ depends only weakly on $\theta_0$. 
 In particular, \cite{barber2020testing}'s method proposes replacing $\theta_0$ with $\hat\theta$ as a plug-in estimate:
 \begin{equation*}
 	\begin{tabular}{c} 
 		\textnormal{aCSS method:  after observing $X$, draw $W\sim\mathcal{N}(0,\frac{1}{d}\I_d)$, compute $\hat\theta = \hat\theta(X, W)$, then }  \\
 		\textnormal{sample $\tilde{X}^{(1)},\dots,\tilde{X}^{(M)}$ i.i.d.\ from $P_{\hat\theta}(\cdot\mid \hat\theta)$.} 
 	\end{tabular}
 \end{equation*}
 Of course, these copies are no longer exactly exchangeable with $X$ under the null, since in general we will have
 $P_{\hat\theta}(\cdot\mid \hat\theta)\neq P_{\theta_0}(\cdot\mid \hat\theta)$. To quantify this issue, 
 \cite{barber2020testing} define the ``distance to exchangeability'',
 \begin{equation*}\label{eq:exchangeability}
 	d_{\textnormal{exch}}(A_{1}, \dots, A_{k}) = \inf\left\{d_{\textnormal{TV}}((A_{1}, \dots, A_{k}), (B_{1}, \dots, B_{k})): B_{1}, \ldots, B_{k}\ \textnormal{are exchangeable}\right\},
 \end{equation*}
 where $d_{\textnormal{TV}}$ denotes the total variation distance.
 The p-value defined in \eqref{eq:pval} is then approximately valid with 
 \[\P(\textnormal{pval}_{T}(X, \tilde{X}^{(1)}, \dots, \tilde{X}^{(M)})  \le\alpha)\le\alpha + 	d_{\textnormal{exch}}(X, \tilde{X}^{(1)}, \dots, \tilde{X}^{(M)}),\]
 where $d_{\textnormal{exch}}(X, \tilde{X}^{(1)}, \dots, \tilde{X}^{(M)})$ can be bounded under certain conditions on the parametric family $\{P_\theta : \theta\in\Theta\}$.

 While aCSS is able to handle a far broader range of models and problems than the CSS framework,
 there are nonetheless limitations to this method that motivate our present work.
 In particular,
 \cite{barber2020testing}'s work assumes a bound on $\|\hat\theta - \theta_0\|$, i.e., consistency of the
 perturbed MLE $\hat\theta$, which may not be possible to achieve in high dimensional settings unless we regularize using constraints or
 non-smooth penalization. Moreover,
 computing
 $P_\theta(\cdot \mid\hat\theta)$, which is a key step in the aCSS procedure, relies heavily on the fact that $\hat\theta$ is the solution to an \emph{unconstrained, differentiable} optimization problem
 over a \emph{convex, open} parameter space $\Theta\subseteq\R^d$ (as these assumptions allow for using first-order optimality conditions on $\hat\theta$
 to derive this conditional distribution), and consequently, aCSS is not able to handle optimization under constraints or under
 a non-differentiable penalty.
 
 \subsubsection{The role of $\sigma$}\label{sec:aCSS_choose_sigma} 
 Here we pause to discuss the role of the noise parameter $\sigma$ in the aCSS method, and the tradeoffs
 inherent in choosing the value of $\sigma$. 
 The aCSS method requires choosing a parameter $\sigma>0$ that controls the amount by which the MLE is perturbed.
 As discussed by  \cite{barber2020testing}, the choice of $\sigma$ represents a tradeoff between
 Type I error control, and the statistical and computational efficiency of the method. A smaller $\sigma$ leads
 to a lower inflation of the Type I error (that is,  \cite{barber2020testing}'s bound on $d_{\textnormal{exch}}(X, \tilde{X}^{(1)}, \dots, \tilde{X}^{(M)})$ 
 increases with $\sigma$). On the other hand, choosing $\sigma$ to be too small can lead to low power---if the perturbed MLE $\hat\theta$
 reveals too much information about $X$, the copies $\tilde{X}^{(m)}$ may be extremely similar to $X$ and therefore, our power
 to reject the null is low. Moreover, a small value of $\sigma$ makes it more challenging to sample the $\tilde{X}^{(m)}$'s from
 the conditional distribution of $X\mid\hat\theta$, since this distribution becomes more concentrated as $\sigma$ tends to zero.
 
 As we will see later on, these considerations will play an important role in our constrained
 version of aCSS, as well. We will return to a discussion of this parameter in Section~\ref{sec:revisit_sigma} below, after defining our new methods and presenting
 theoretical results.

 \subsection{Additional related work} \label{sec:related_work}
 The literature on GoF testing is extensive, particularly in low-dimensional settings, and giving an overview
 of this broad field is beyond the scope of the present work.
 Here we discuss some challenges faced in the high-dimensional regime.
 
 	For high dimensional two-sample test,  to correct for the inconsistency of Hotelling's $T^{2}$ in high dimensions, \cite{srivastava2016raptt} focus on projecting the high-dimensional data onto a lower-dimensional subspace and \cite{li2020adaptable} propose a test based on a ridge-regularized Hotelling's $T^{2}$.  
 	For simple null testing in high-dimensional linear and generalized linear models, pointwise statistical inference for the components of the parameter vector, such as the construction of confidence intervals and p-values, is studied via the distribution of estimation error when considering lasso and sparse models \cite{van2014asymptotically, zhang2014confidence, dezeure2015high}.
 	When applied to the composite null case, which is more related to the problems we considered, \cite{shah2018goodness} propose the Residual Prediction (RP) tests for high-dimensional Gaussian linear models. The family of test statistics is related to signals left in scaled residuals, and the null distribution is mimicked via parametric bootstrap with a lasso-type estimate. The sampled scaled residuals are shown to depend only weakly on the unknown regression coefficients as long as the sign of the estimation is correct.
 	\cite{jankova2020goodness} generalize RP tests to generalized linear models. They propose a specific test statistic based on the Pearson-type residuals and 
 	an auxiliary dataset. The test statistic is asymptotically normal under the null when the estimation is in the local neighborhood of the true parameter.
 	Note that the aforementioned works are all restricted to specific test statistics. In contrast, our approach offers greater flexibility, allowing users to choose test statistics tailored to their particular problem or prior knowledge, which may yield higher power under specific alternatives. Moreover, our framework explicitly characterizes the relationship between Type I error control and estimation error.
 
 \section{The aCSS method with linear constraints}\label{sec:constrained-aCSS}
 Our constrained aCSS method will address the problem of goodness-of-fit testing for the hypothesis
 \[ H_{0}: X\sim P_{\theta} \textnormal{ for some } \theta\in\Theta,\]
 where as before, $\{P_{\theta}: \theta\in\Theta\}$ is a parametric family, indexed by a convex and open subset $\Theta\subseteq\R^d$.
 For \cite{barber2020testing}'s aCSS method to provide approximate Type I error control,
 we need consistency of the (perturbed) MLE, i.e., a bound on $\|\hat\theta - \theta_0\|$. 
 Many important problems are therefore excluded from this framework. In particular, 
 consistency of the MLE cannot be assumed for problems where the unconstrained MLE is not well-defined---for example,
 a mixture of two Gaussians with unknown means and variances, due to the degenerate behavior of the likelihood as 
 we take one component's variance to zero. In addition, consistency of the MLE will not hold for high-dimensional
 problems, such as Gaussian linear regression with dimension $d$ larger than the sample size $n$---even if we add a ridge regularizer $R(\theta)$
 so that the solution $\hat\theta$ is unique, in general $\hat\theta$ will not be a consistent estimator of $\theta$. 
 In high-dimensional settings, achieving consistent parameter estimation is impossible without additional structural assumptions.  Constraints serve as an effective tool for incorporating prior knowledge about the structure into the estimation process. The most common illustration of this is the application of LASSO \cite{tibshirani1996regression} and the Dantzig selector \cite{candes2007dantzig} under specific sparsity assumptions. These techniques, linked with $\ell_{1}$-regularization, have been demonstrated to be consistent \cite{ bickel2009simultaneous, zhang2008sparsity,zhao2006model}.
 In contrast to aCSS, however, where we need to be able to estimate the true parameter $\theta_0$ accurately
 with the \emph{unconstrained} MLE solution $\hat\theta$, here we are interested in settings where $\theta_0$
 can only be accurately estimated with a \emph{constrained} optimization problem.

 To this end, we now introduce constraints, \[A\theta\leq b,\] for a fixed and known matrix $A\in\R^{r\times d}$ and 
 vector $b\in\R^r$. The inequality should be interpreted elementwise, i.e., we are requiring $(A\theta)_i\leq b_i$ for each $i=1,\dots,r$.
 (Of course, in the special case $r=0$, this reduces to the earlier, unconstrained setting.)
 At a high level, to run aCSS in this setting, we first need to compute a constrained MLE (with a random perturbation),
 \begin{equation}\label{eqn:def_thetahat}\hat\theta = \hat\theta(X,W) = \argmin_{\theta\in\Theta}\left\{ \L(\theta;X,W) \, : \, A\theta\leq b\right\},\end{equation}
 where
 \[\L(\theta;X,W)= \L(\theta;X) + \sigma W^\top  \theta, \quad \L(\theta;X) = - \log f(X;\theta) + R(\theta). \]
 As before, $f(\cdot;\theta)$ is the density for distribution $P_\theta$, $R(\theta)$ is an optional 
 twice-differentiable regularizer, $W\sim\mathcal{N}(0,\frac{1}{d}\I_d)$ is independent Gaussian noise, and $\sigma>0$ is a  parameter that controls the magnitude of this perturbation. 
 We then compute the conditional distribution of $X$ given $\hat\theta$, and sample the copies $\tilde{X}^{(1)},\dots,\tilde{X}^{(M)}$
 from this conditional distribution (or rather, sample from an approximation, since $\theta_0$ is unknown). 
 Defining
 \begin{equation}\label{eqn:def_ghat}\hat{g} = \hat{g}(X,W) = \nabla_\theta \L(\hat\theta(X,W);X,W),\end{equation}
 we can see that we would trivially have $\hat{g}\equiv 0$ in the unconstrained setting but may in general have $\hat{g}\neq 0$ now that constraints
 have been introduced. We will see that, in the constrained optimization setting, while $\hat\theta$ on its own does not carry enough information to serve
 as an approximately sufficient statistic, instead the pair $(\hat\theta,\hat{g})$ now plays this role.
 
 For each $\theta\in\Theta$, we will define $P_\theta(\cdot\mid \hat\theta,\hat{g})$ as the conditional distribution of $X\mid (\hat\theta,\hat{g})$ if we assume
 that $X$ was drawn as $X\sim P_\theta$. Using $\hat\theta$ as a plug-in for the true parameter 
 $\theta_0$,  we will use $P_{\hat\theta}(\cdot\mid \hat\theta,\hat{g})$ as the distribution from which the copies $\tilde{X}^{(m)}$ are drawn.
 The constrained aCSS algorithm is then defined via the following steps:
 \begin{quote}	\normalsize
 	\textbf{Constrained aCSS algorithm (informal version):}
 	\begin{enumerate}
 		\item Observe data $X\sim P_{\theta_0}$.
 		\item Draw noise $W\sim\mathcal{N}(0,\frac{1}{d}\I_d)$.
 		\item Solve for a constrained perturbed MLE $\hat\theta = \hat\theta(X,W)$ as in~\eqref{eqn:def_thetahat},
 		and compute the corresponding gradient $\hat{g}=\hat{g}(X,W)$ as in~\eqref{eqn:def_ghat}.
 		\item Sample the copies $\tilde{X}^{(1)},\dots,\tilde{X}^{(M)}$ from the approximate conditional distribution $P_{\hat\theta}(\cdot\mid\hat\theta,\hat{g})$.
 		\item Compute the p-value defined in~\eqref{eq:pval}
 		for our choice of test statistic $T$.
 	\end{enumerate}
 \end{quote}
 As compared to (unconstrained) aCSS, the difference lies 
 in the fact that $\hat\theta$ is computed via a constrained optimization problem, and as a result, the conditional distribution $P_\theta(\cdot\mid\hat\theta,\hat{g})$
 is now  more challenging to compute; we will return to this question shortly.
 
 When running constrained aCSS, 
 we note that we are not assuming explicitly that the true parameter $\theta_0$ itself satisfies the constraints---that is,
 we do not assume $A\theta_0\leq b$ must hold. However, 
 in order for the method to retain approximate Type I error control, $\hat\theta = \hat\theta(X,W)$ will need to be an accurate
 estimator of $\theta_0$; this implicitly requires that $A\theta_0\leq b$ must at least approximately hold.
 
 The choice of $\sigma$ controls the amount of perturbation in the constrained MLE $\hat\theta$.
 This choice represents a tradeoff between Type I error, which is better for small $\sigma$, versus statistical power
 and computational efficiency, which tend to improve with larger $\sigma$---this tradeoff occurs
 for unconstrained aCSS as well (see Section~\ref{sec:aCSS_choose_sigma}). For constrained aCSS,
 additional challenges can arise since we may now be working in a high-dimensional setting---we will discuss
 these questions more in Section~\ref{sec:theory} below, when presenting our theoretical results, and will
 explore the role of $\sigma$ empirically in our simulations in Section~\ref{sec:numerical}.

 \subsection{Examples of constraints}\label{sec:examples_of_constraints} 
 Before defining the method more formally, we present several key examples of constraints $A\theta\leq b$ to motivate this method.
 \begin{itemize}
 	\item Nonnegativity constraint: if we believe $\theta_0$ has only nonnegative entries, we can choose 
 	\[A = - \I_d, \quad b = \mathbf{0}_d\]
 	to enforce $\theta_i \geq 0$ for all $i$.
 	\item Bounding away from zero: if we believe the entries of $\theta_0$ cannot be too close to zero, 
 	we can choose
 	\[A = - \I_d, \quad b = -c\cdot \mathbf{1}_d,\]
 	for a small constant $c>0$ (or we can take a submatrix of the identity, if we want to place a lower bound on only certain entries of $\theta$),
 	to enforce $\theta_i\geq c$ for all $i$ (or for certain entries).
 	For example, for a Gaussian mixture model, we need to place a positive lower bound on the variance of each component in order
 	for the MLE to be well-defined.
 	\item Monotonicity constraint: if we believe $\theta_0$ has entries that appear in nondecreasing order, i.e., $(\theta_0)_1\leq \dots \leq (\theta_0)_d$,
 	we can choose
 	\[A = \left(\begin{array}{cccccc}1 & -1 & 0 & \dots & 0 & 0 \\
 		0 & 1 & - 1 & \dots & 0 & 0\\
 		\vdots & \vdots& \vdots& &\vdots& \vdots\\
 		0 & 0 & 0 & \dots & 1 & -1\end{array}\right), \quad b = \mathbf{0}_d,\]
 	to enforce the monoticity constraint $\theta_1\leq \dots\leq \theta_d$.
 	\item $\ell_\infty$ constraint: if we believe $\theta_0$ has bounded entries, we can choose
 	\[A = \left(\begin{array}{c} \I_d \\ - \I_d \end{array}\right) , \quad b = C \cdot \mathbf{1}_{2d},\]
 	to enforce the constraint $\|\theta\|_\infty\leq C$.
 	\item $\ell_1$ constraint: if we believe that $\theta_0$ is sparse or approximately sparse, such as in a high-dimensional regression problem, 
 	we can choose
 	\[ A \in\{\pm 1\}^{2^d\times d} \textnormal{ (with rows given by the set of sign vectors of length $d$)}, \quad b = C\cdot\mathbf{1}_{2^d}\]
 	in order to enforce the constraint $\|\theta\|_1\leq C$.
 	(Note that, in high-dimensional statistics, it is more common to use an $\ell_1$ penalty---i.e., the lasso---rather than an $\ell_1$ constraint,
 	when defining the regularized MLE. We will define a penalized version of our method later on, in Section~\ref{sec:penalty}.)
 	\item Fused $\ell_1$ norm constraint: if we believe $\theta_0$ is locally constant (or is smooth and therefore can be well approximated
 	by a locally constant vector), we can choose to constrain $\|D\theta\|_1 \leq C$, 
 	where $D\in\{-1,0,+1\}^{(d-1)\times d}$ is defined with first row $(+1,-1,0,\dots,0)$, second row $(0,+1,-1,0,\dots,0)$, etc, so that $\|D\theta\|_1 = \sum_{i=1}^{d-1}|\theta_i - \theta_{i+1}|$.
 	This corresponds to choosing $A\in\R^{2^{d-1}\times d}$ given by $A = A'\cdot D$, where $A'\in\{\pm 1\}^{2^{d-1}\times (d-1)}$ has
 	rows given by all possible sign vectors of length $d-1$, and $b = C\cdot \mathbf{1}_{2^{d-1}}$. 
 \end{itemize}

 \subsection{Formally defining the method}
 We now turn to the details of the method and its implementation, including questions of optimization and sampling,  
 then combine all these ingredients to formally define the constrained aCSS method.
 \subsubsection{The second-order stationary condition}\label{sec:SSOSP}
 First we consider the question of optimization. 
 In certain settings, it may be the case that we cannot reliably solve
 for the global minimizer of $\L(\theta;X,W)$, or, that this global minimizer may not be well-defined or may not be unique---for 
 example, the negative log-likelihood might be nonconvex.  Note that, in general, $\L(\theta;X)$ can also represent objective functions other than negative log-likelihood for the good estimation of $\theta$. Formally, we define 
 \[\hat\theta:\X\times\R^d\rightarrow\Theta\]
 to be \emph{any} measurable function, which represents the output of our solver
 when we input the constrained optimization problem~\eqref{eqn:def_thetahat}.
 For each subset $\A\subseteq[r]$ of constraints, define a matrix $U_{\A}$ that forms an orthonormal basis for subspace orthogonal to $\textnormal{span}\{A_i : i\in \A\}$ (where $A_i\in\R^d$ is the vector given by the $i$th row of $A$), that is,
 \begin{equation}\label{eqn:define_U} \textnormal{$U_{\A}\in\R^{d \times (d-\textnormal{rank}(\textnormal{span}\{A_i \}_{i\in \A}))}$ satisfies $U_{\A}U_{\A}^\top  = \mathcal{P}^\perp_{\textnormal{span}\{A_i\}_{ i\in \A}}$,}\end{equation}  
 so that $U_{\A}U_{\A}^\top $ projects to the subspace orthogonal to the span of constraints indexed by $\A$.
 \begin{definition}[SSOSP]\label{def:ssosp} A parameter $\theta\in\Theta$ is a strict second-order stationary point (SSOSP) of the optimization
 	problem~\eqref{eqn:def_thetahat} if it satisfies all of the following:		
 	\begin{enumerate}
 		\item Feasibility: 
 		\[A \theta \leq b.\]
 		\item  First-order necessary conditions, i.e., Karush--Kuhn--Tucker (KKT) conditions: 
 		\[\nabla_\theta \L(\theta;X,W) + \sum_{i=1}^r \lambda_i A_i = 0,\]
 		where $\lambda_i\geq 0$ for all $i$, and $\lambda_i=0$ for all $i\in[r]\backslash \A(\theta)$, 
 		where $\A(\theta) = \{i\in[r]: A_i^\top  \theta = b_i\}$ is the set of active constraints.
 		\item   Second-order sufficient condition: \[U_{\A(\theta)}^\top  \nabla^2_{\theta}\L(\theta; X,W) U_{\A(\theta)}\succ 0,\]
 		that is, the Hessian $ \nabla^2_{\theta}\L(\theta; X,W) $ is strictly positive definite when restricted to the subspace orthogonal 
 		to the active constraints.
 	\end{enumerate}  
 \end{definition}
 As in the unconstrained aCSS algorithm \citep{barber2020testing}, to allow for the possibility that our solver might fail to find a valid solution,
 if $\hat\theta(X,W)$ fails the SSOSP condition then we will set $\tilde{X}^{(1)}=\dots=\tilde{X}^{(M)}=X$ to trivially obtain a p-value of 1 (i.e., to avoid
 the possibility of a rejection in this scenario where our estimate $\hat\theta$ of $\theta_0$ is unreliable).

 \subsubsection{The conditional distribution}
 With the SSOSP condition in place, we are now ready to define the conditional distribution $P_\theta(\cdot\mid \hat\theta,\hat{g})$. 
 We first need some regularity conditions. 
 \begin{assumption}
 	\label{asm:reg}
 	Assume the family $\{P_{\theta}: \theta\in\Theta\}$ and regularization function $R(\theta)$ satisfy:
 	\begin{enumerate}
 		\item[$\bullet$]  $\Theta\subseteq \R^{d}$ is a convex and open set;
 		\item[$\bullet$] For each $\theta\in \Theta$, $P_{\theta}$ has density $f(x; \theta)>0$ with respect to a common base measure $\nu_{\X}$;
 		\item[$\bullet$] for each $x\in\X$, the function  $\theta\rightarrow \L(\theta; x) = -\log f(x;\theta) + R(\theta)$ 	is continuously twice differentiable.
 	\end{enumerate}
 \end{assumption} 
 This first assumption is the same as Assumption 1 of \cite{barber2020testing}, for the unconstrained aCSS setting. The following result, however,
 is a strict generalization of \cite[Lemma 1]{barber2020testing}, computing the conditional density of $X$ after solving for $\hat\theta$ under linear constraints (with the unconstrained setting
 as a special case).
 \begin{lemma}[Conditional density]\label{lemma:conditional_density}
 	Suppose Assumption \ref{asm:reg} holds. For $A \in \R^{r\times d}$, $b \in \R^{r}$, fix any $\theta_{0}\in \Theta$ and let $(X, W, \hat{\theta}, \hat{g})$ be drawn from the
 	joint model
 	\begin{equation}
 		\left\{ \begin{array}{l}
 			X\sim P_{\theta_{0}},\\
 			W\sim \mathcal{N}(0, \frac{1}{d}\I_{d}),\\
 			\hat\theta = \hat\theta(X, W), \\
 			\hat{g} = \hat{g}(X, W) = \nabla_{\theta}\L(\hat\theta; X, W).
 		\end{array}\right.
 	\end{equation}  
 	
 	Fix any $\A\subseteq[r]$, and assume that the event that
 	$\hat{\theta}(X, W)$ is a SSOSP of \eqref{eqn:def_thetahat} with active set $\A(\hat{\theta}(X, W)) = \A$ has positive probability. Then, conditional on this event,
 	the conditional distribution of $X| \hat{\theta}, \hat g$ has density  
 	\begin{equation}
 		\label{eqn:conditional_density}
 		p_{\theta_{0}}(\cdot\mid\hat{\theta}, \hat g )\propto f(x; \theta_{0})\cdot \exp\left\{-\frac{\|\hat{g} - \nabla_{\theta}\L(\hat\theta; x)\|^{2} }{2\sigma^{2}/d}\right\}\cdot\det\left(U_{\A}^\top \nabla^{2}_{\theta}\L(\hat\theta; x)U_{\A}\right) \cdot
 		\id_{x\in\X_{\hat\theta,\hat{g}}}
 	\end{equation} 	
 	with respect to the base measure $\nu_{\X}$, where $U_{\A}$ is defined in \eqref{eqn:define_U} and
 	\[\X_{\theta,g} = \left\{x\in\X: \textnormal{ for some $w\in\R^d$, $\theta=\hat\theta(x,w)$ is a SSOSP of~\eqref{eqn:def_thetahat}, and $g=\nabla\L(\theta;x,w)$}\right\}.\]
 \end{lemma}
 The four terms of the conditional density reflect, respectively, the original distribution of $X$ in the first term; the
 Gaussian distribution of the noise $W$ in the second term; the determinant term, which captures a change-of-variables type calculation relating $(X,W)$ with $(X,\hat\theta,\hat{g})$;
 and the final indicator term, which accounts for possible failure to find a SSOSP.
 In the case where $\A = \emptyset$, i.e., no active constraints, we have $\hat{g} \equiv 0$ (by first-order optimality) and the conditional density then coincides
 with the calculations in \cite{barber2020testing} for the unconstrained case.
 
 With this calculation in place, we can now specify the estimated conditional distribution $P_{\hat\theta}(\cdot\mid\hat\theta,\hat{g})$,
 from which we would like to sample the copies $\tilde{X}^{(1)},\dots,\tilde{X}^{(M)}$ for the constrained aCSS algorithm: it is the distribution obtained
 by plugging in $\hat\theta$ in place of the unknown $\theta_0$, in the conditional distribution computed in Lemma~\ref{lemma:conditional_density}, namely,\footnote{For this to result in a well defined density, we need to verify that the right-hand side integrates to a positive and finite value; 
 	in fact, this holds almost surely on the event that $\hat\theta=\hat\theta(X,W)$ is a SSOSP, as we will verify in Appendix B.}
 \begin{equation}
 	\label{eqn:conditional_density_thetahat}
 	p_{\hat\theta}(\cdot\mid\hat{\theta}, \hat g )\propto f(x; \hat\theta)\cdot \exp\left\{-\frac{\|\hat{g} - \nabla_{\theta}\L(\hat\theta; x)\|^{2} }{2\sigma^{2}/d}\right\}\cdot\det\left(U_{\A(\hat\theta)}^\top \nabla^{2}_{\theta}\L(\hat\theta; x)U_{\A(\hat\theta)}\right) \cdot
 	\id_{x\in\X_{\hat\theta,\hat{g}}}.
 \end{equation} 
 As mentioned at the beginning of Section \ref{sec:SSOSP}, in practice $\mathcal{L}(\theta; X)$ can represent objective functions beyond the negative log-likelihood for the good estimation of $\theta$. From the proof in Appendix A.3, we can see that any $\mathcal{L}(\theta; X)$ is applicable for deriving the conditional density, as long as it is continuously twice differentiable.

 \subsubsection{Sampling strategies}
 
 In the informal version of the algorithm defined above, we require that the copies $\tilde{X}^{(m)}$ are drawn i.i.d.\ from the conditional density
 $p_{\hat\theta}(\cdot\mid \hat\theta,\hat{g})$, as calculated in~\eqref{eqn:conditional_density_thetahat}. In other words,
 conditional on $X,\hat\theta,\hat{g}$, the collection of copies is drawn from a product distribution,
 \begin{equation}\label{eqn:copies_iid}(\tilde{X}^{(1)},\dots,\tilde{X}^{(M)})\mid (X,\hat\theta,\hat g) \sim p_{\hat\theta}(\cdot\mid \hat\theta,\hat{g})\times \dots \times p_{\hat\theta}(\cdot\mid \hat\theta,\hat{g}).\end{equation}
 In some settings, this may be computationally very easy---we will see some examples of this type below when the parametric family $\{P_\theta\}$
 is Gaussian. In more complex settings, however, sampling directly from $p_{\hat\theta}(\cdot\mid \hat\theta,\hat{g})$ may be infeasible,
 and we will instead turn to approximations, such as MCMC-based strategies.
 Of course, without analyzing complex conditions such as the mixing time of the Markov chain, we cannot ensure that theoretical guarantees enjoyed by
 the algorithm would be preserved when sampling directly from $p_{\hat\theta}(\cdot\mid \hat\theta,\hat{g})$ is replaced with an approximation---particularly
 as this approximation might induce additional dependence among the copies. 
 
 In the unconstrained aCSS setting, \cite{barber2020testing} describe several exchangeable MCMC strategies, based on the work of \cite{besag1989generalized}, that 
 avoid these difficulties. For completeness, we will describe these schemes in more detail in  Appendix D.  
 In general, following \cite{barber2020testing}, we can generalize the sampling strategy~\eqref{eqn:copies_iid}, drawing the copies as
 \[(\tilde{X}^{(1)},\dots,\tilde{X}^{(M)}) \mid (X,\hat\theta,\hat{g}) \sim \tilde{P}_M(\cdot;X,\hat\theta,\hat{g})\]
 where the family of conditional distributions $\{\tilde{P}_M(\cdot;x,\theta,g)\}$ is required to satisfy the following condition:
 \begin{equation}\label{eqn:define_tilde_P}
 	\begin{tabular}{c}
 		\textnormal{If $X\sim p_\theta(\cdot\mid \theta,g)$ and $(\tilde{X}^{(1)},\dots,\tilde{X}^{(M)})\mid X \sim \tilde{P}_M(\cdot;X,\theta,g)$, then}\\
 		\textnormal{the random vector $(X,\tilde{X}^{(1)},\dots,\tilde{X}^{(M)})$ is exchangeable.}\end{tabular}
 \end{equation}
 In particular, we note that choosing
 \[\tilde{P}_M(\cdot;x,\theta,g) = p_\theta(\cdot\mid \theta,g)\times\dots\times p_\theta(\cdot\mid \theta,g),\]
 i.e., sampling the copies i.i.d.\ from $p_\theta(\cdot\mid \theta,g)$,
 will trivially always satisfy the exchangeability condition~\eqref{eqn:define_tilde_P}. More generally, however, if sampling the copies directly from $p_\theta(\cdot\mid \theta,g)$
 is computationally infeasible, the MCMC based strategy described in   Appendix D
 will also satisfy~\eqref{eqn:define_tilde_P} while allowing for more complex problems where direct sampling is not achievable.
 
 \subsubsection{Combining everything}
 With all our formal calculations and definitions in place, we can now state the full version of the constrained aCSS algorithm.
 \begin{quote}
 	\normalsize
 	\textbf{Constrained aCSS algorithm:}\begin{enumerate}
 		\item Observe data $X\sim P_{\theta_0}$.
 		\item Draw noise $W\sim\mathcal{N}(0,\frac{1}{d}\I_d)$.
 		\item Solve for a constrained perturbed MLE $\hat\theta = \hat\theta(X,W)$ as in~\eqref{eqn:def_thetahat},
 		and compute the corresponding gradient $\hat{g}=\hat{g}(X,W)$ as in~\eqref{eqn:def_ghat}.
 		\item If $\hat\theta$ is not a SSOSP of~\eqref{eqn:def_thetahat}, then set $\tilde{X}^{(1)}=\dots=\tilde{X}^{(M)}=X$.
 		Otherwise, sample copies $(\tilde{X}^{(1)},\dots,\tilde{X}^{(M)}) \mid (X,\hat\theta,\hat{g}) \sim \tilde{P}_M(\cdot;X,\hat\theta,\hat{g})$, where $\tilde{P}_M$
 		is chosen to satisfy property~\eqref{eqn:define_tilde_P} relative to the conditional density $p_{\hat\theta}(\cdot\mid \hat\theta,\hat g)$ as computed
 		in~\eqref{eqn:conditional_density_thetahat}.
 		\item Compute the p-value defined in~\eqref{eq:pval}
 		for our choice of test statistic $T$.
 	\end{enumerate}
 \end{quote}
 This more general form of the constrained aCSS algorithm is more flexible than our original informal definition: it allows us to handle
 settings where solving for the (perturbed, constrained) MLE is more challenging (e.g., convergence may not be guaranteed),
 as well as settings where sampling directly from the estimated conditional density~\eqref{eqn:conditional_density_thetahat} may be computationally infeasible.
 
 \section{Theoretical results}\label{sec:theory}
 
 In this section, we provide theoretical guarantees for the constrained aCSS procedures, establishing
 an upper bound on the Type I error level of the test. 
 First, in Section~\ref{sec:theory_general}, we give a general result that holds for any problem where constrained aCSS can be applied.
 We will then refine the result to provide a stronger bound for two special cases: Section~\ref{sec:theory_sparse}
 addresses the setting where $\hat\theta$ is sparse
 in some basis, and Section~\ref{sec:theory_gaussian} considers the setting of (potentially high-dimensional) Gaussian data.

 \subsection{General results: Type I error control}\label{sec:theory_general} 
 In order to establish a bound on the Type I error level of the constrained aCSS procedure,
 we first need several assumptions (in addition to the regularity conditions of Assumption~\ref{asm:reg}). 
 The following assumption ensures that, with high probability,
 we successfully find a strict second-order stationary point (SSOSP) $\hat\theta$ of the optimization problem~\eqref{eqn:def_thetahat},
 and this solution $\hat\theta$ is a good approximation to the true parameter $\theta_0$.
 \begin{assumption}
 	\label{asm:SSOSP_general} 
 	For any $\theta_{0}\in  \Theta$ in Assumption \ref{asm:reg}, the estimator $\hat{\theta}: \X\times\R^{d}\rightarrow \Theta$ satisfies
 	\begin{equation*}
 		\left\{ \begin{array}{l}
 			\hat{\theta}(X, W)\ \textnormal{is a SSOSP of the constrained  optimization problem \eqref{eqn:def_thetahat}},\\
 			\|\hat{\theta}(X, W) - \theta_{0}\|\le r(\theta_{0}), 
 		\end{array}\right.
 	\end{equation*}
 	with probability at least $1 - \delta(\theta_{0})$, where  the probability is taken with respect to the distribution $(X, W)\sim P_{\theta_{0}}\times N(0, \frac{1}{d}\I_{d})$.
 \end{assumption}  
 Next, we need an assumption on the Hessian of the log-likelihood.	Define $H(\theta; x) = -\nabla_{\theta}^{2}\log f(x; \theta)$,  and let $H(\theta) = \E_{\theta_{0}}\left[H(\theta; x)\right]$. 
 \begin{assumption}\label{asm:L_H}
 	For any $\theta_{0}\in \Theta$, the expectation $H(\theta)$ exists for all $\theta\in\B(\theta_{0}, r(\theta_{0}))\cap\Theta$, and furthermore 
 	\begin{equation}\label{eqn:Hessian_asm_1}
 		\E_{\theta_0} \left[\sup_{\theta\in \B(\theta_{0}, r(\theta_{0}))\cap \Theta} r(\theta_{0})^{2}\left(\lambda_{\max}\left(H(\theta) - H(\theta; X) \right)\right)_{+} \right]\le \epsilon(\theta_{0}),
 	\end{equation}
 	\begin{equation}\label{eqn:Hessian_asm_2}
 		\log\E_{\theta_{0}}\left[\exp\left\{\sup_{\theta\in\B(\theta_{0}, r(\theta_{0}))\cap\Theta}r(\theta_{0})^{2}\cdot \left(\lambda_{\max}(H(\theta; X)-H(\theta))\right)_{+}\right\}\right]\le\epsilon(\theta_{0}).
 	\end{equation}
 	Here $r(\theta_{0})$ is the same constant as that appears in Assumption \ref{asm:SSOSP_general}.
 \end{assumption}
 These two assumptions are analogous to Assumptions 2 and 3 in \cite{barber2020testing}'s theoretical results for unconstrained aCSS.
 However,  in the present work $\hat\theta$ is defined as the solution to the constrained, rather than unconstrained, perturbed maximum likelihood estimation problem. Since constraints allow for more accurate estimation in many settings, 
 we can expect that the error $\|\hat\theta - \theta_0\|$ might be substantially smaller in this constrained setting, making these assumptions
 more realistic for a broader range of problems.

 \begin{theorem}\label{theorem_general}
 	Suppose Assumptions~\ref{asm:reg}, \ref{asm:SSOSP_general},  \ref{asm:L_H} hold,
 	and the data is generated as $X\sim P_{\theta_0}$. Then the copies $\tilde{X}^{(1)}, \dots, \tilde{X}^{(M)}$
 	generated by the constrained aCSS procedure are approximately exchangeable with $X$, satisfying	
 	\begin{equation*}
 		d_{\textnormal{exch}}(X, \tilde{X}^{(1)}, \dots, \tilde{X}^{(M)})\le 3\sigma r(\theta_{0}) + \epsilon(\theta_{0}) + \delta(\theta_{0}),
 	\end{equation*}
 	where  $r(\theta_{0}), \epsilon(\theta_{0}) , \delta(\theta_{0})$ are defined in Assumptions \ref{asm:SSOSP_general} and \ref{asm:L_H}. In particular, this implies that for any predefined test statistic $T: \X \rightarrow \R$ and rejection threshold $\alpha\in[0,1]$, the p-value defined in~\eqref{eq:pval} satisfies
 	\begin{equation*}
 		\P\left(\textnormal{pval}_{T}(X, \tilde{X}^{(1)}, \dots, \tilde{X}^{(M)})\le \alpha \right)\le \alpha + 3\sigma r(\theta_{0}) + \epsilon(\theta_{0}) + \delta(\theta_{0}).
 	\end{equation*}
 \end{theorem}
 
 The above upper bound on the Type I error appears identical to the result of \cite[Theorem 1]{barber2020testing}, but in fact this new
 result offers important contributions. Firstly, this new result holds for the more complex setting of a constrained optimization problem,
 which requires a more technical analysis. Moreover, as mentioned above, the estimation
 error $\|\hat\theta - \theta_0\|$   may be much smaller for the constrained optimization problem, since constraints can reduce
 the effective dimensionality of the statistical problem; consequently, the value of $r(\theta_0)$  can be much smaller in the constrained
 setting, leading to a tighter bound on Type I error control. (We will see that our empirical results, shown in Section~\ref{sec:numerical},  
 support this intuition.) 
 
 \subsubsection{Revisiting the role of $\sigma$}\label{sec:revisit_sigma}
 As discussed earlier in Section~\ref{sec:aCSS_choose_sigma}, the choice of $\sigma$ plays an important role in the performance of the method,
 typically with better Type I error control when $\sigma$ is smaller versus better power when $\sigma$ is larger. Now we return
 to this question in the context of constrained aCSS.
 The upper bound on Type I error shown in Theorem~\ref{theorem_general} suggests that $\sigma$ should not be too large---in particular,
 for most statistical settings with sample size $n$, we can expect $r(\theta_0)\asymp n^{-1/2}$ at best, suggesting that we need
 to choose $\sigma\ll n^{1/2}$ to ensure a  meaningful bound on Type I error. On the other hand,
 recalling that the noise $W$ in the perturbed maximum likelihood
 estimation problem~\eqref{eqn:def_thetahat} is generated as $W\sim \mathcal{N}(0,\frac{1}{d}\I_d)$, 
 in a high-dimensional setting where $d\gg n$ the perturbation term $\sigma W^\top\theta$ in~\eqref{eqn:def_thetahat}
 may therefore be negligible. This might lead to extremely low power and/or to computational challenges in sampling the copies
 $\tilde{X}^{(m)}$. This issue leads us to our next question: 
 are there any settings where we can improve the result of Theorem~\ref{theorem_general}, 
 and allow for a larger value of $\sigma$?

 \subsection{Special case: sparse structure}\label{sec:theory_sparse}
 We next turn to the special case where, due to the constraints imposed on the estimate $\hat\theta$,
 we can assume that the error $\hat\theta - \theta_0$ is likely to be sparse, relative to some basis.
 We will see that, in this setting, the upper bound on Type I error given in Theorem~\ref{theorem_general} can be improved
 to account for the lower effective dimension of $\hat\theta$, and that we are therefore free to use a substantially larger value of $\sigma$
 in the constrained aCSS procedure---leading downstream to higher power and easier computation.
 
 To formalize this idea, consider a fixed set of vectors $v_1,\dots,v_p\in\R^d$. We are interested in settings
 where the solution $\hat\theta$ to the perturbed constrained maximum likelihood estimation problem~\eqref{eqn:def_thetahat}
 is likely to lie in the span of a small subset of $v_i$'s. To motivate this setting, we can revisit several examples
 that we considered in Section~\ref{sec:examples_of_constraints}:
 \begin{itemize}
 	\item Sparsity: in a setting where we believe $\theta_0$ is sparse, we might use an $\ell_1$ constraint
 	for the optimization problem, requiring $\|\theta\|_1\leq C$, which is likely to lead to a solution $\hat\theta$ that is sparse as well.
 	In this setting, we can take $p=d$ and choose the set of vectors to be the canonical basis, i.e.,  $v_i = \mathbf{e}_i$ for $i\in[d]$,
 	reflecting our belief that the error $\hat\theta-\theta_0$ will itself be sparse.
 	\item Locally constant signal: if we believe $\theta_0$ is locally constant, 
 	we might choose the constraint $\sum_{i=1}^{d-1}|\theta_i - \theta_{i+1}|\leq C$.
 	This constraint often leads to solutions $\hat\theta$ that are piecewise constant, with $\hat\theta_i = \hat\theta_{i+1}$ 
 	for many indices $i\in[d-1]$, and therefore the error $\hat\theta-\theta_0$ will also 
 	be piecewise constant. Consequently, we can take $p = d$, and choose $v_i =  \mathbf{e}_{1} + ... + \mathbf{e}_{i}$ for $i\in[d]$.
 	(This choice of vectors $\{v_i\}$ means that, for any $w\in\R^d$, if $w$ has $\ell$ many changepoints---that is, $w_i \neq w_{i+1}$ 
 	for $\ell$ many indices $i$---then $w$ can be written as a linear combination of at most $\ell+1$ many $v_i$'s.)
 	\item Monotonicity: in a setting where we believe $\theta_0$ is monotone nondecreasing, 
 	we might use the isotonic constraint, choosing $A$ and $b$ to constrain $\theta_1 \leq \dots \leq \theta_d$.
 	This constraint often leads to solutions $\hat\theta$ that are piecewise constant, with $\hat\theta_i = \hat\theta_{i+1}$ 
 	for many indices $i\in[d-1]$.
 	If the true parameter $\theta_0$ is also piecewise constant, we therefore again have an error $\hat\theta-\theta_0$
 	that is likely to be piecewise constant, and we can then choose the same $v_i$'s as for the preceding example.
 \end{itemize}
 
 \subsubsection{Effective dimension definitions}\label{sec:effectivedim}
 For a given choice of vectors $\{v_i\}_{i\in[p]}$, we define
 \[\|w\|_{v,0} = \begin{cases}
 	\min\left\{ |S| : S\subseteq[p], w \in\textnormal{span}(\{v_i\}_{i\in S})\right\}, & w \in\textnormal{span}(\{v_i\}_{i\in[p]}),\\
 	+\infty,&\textnormal{otherwise}.\end{cases}\]
 for any $w\in\R^d$. In other words, $\|w\|_{v,0}$ is the minimum number of vectors $v_i$ needed 
 so that $w$ lies in their span. Note that, despite the notation, the function $w\mapsto \|w\|_{v,0}$ is not a norm. We choose this notation
 to agree with the commonly used ``$\ell_0$ norm'', $\|w\|_0$, the number of nonzero elements of the vector $w$; in particular,
 in the first example where $v_i = \mathbf{e}_i$, $i\in[d]$, we have $\|w\|_{v,0} = \|w\|_0$.
 
 Next, for each $k=0,\dots,d$, we define
 \[h_v(k) = \E_{Z\sim \mathcal{N}(0,\I_d)}\left[\max_{S\subseteq[p],|S|\leq k} \|\mathcal{P}_{v_S}(Z)\|^2\right],\]
 where $\mathcal{P}_{v_S}$ denotes projection to $\textnormal{span}(\{v_i\}_{i\in S})$.
 This quantity will play an important role in our theory below.
 We can think of $h_v(k)$ as describing the ``effective dimension'' of vectors that can be written as a $k$-sparse combination
 of the vectors $v_1,\dots,v_p$.
 In particular, we can see that for any $k$, we have
 $h_v(k) \leq \E_{Z\sim \mathcal{N}(0,\I_d)}[\|Z\|^2] = d$. On the other hand, if $k\ll d$, the following result shows that $h_v(k)$ can be substantially smaller:
 \begin{lemma}\label{lemma:proj_chisq}
 	For each $k$ it holds that
 	$h_v(k) \leq  \min\{4k\log(4p/k), d\}$.
 \end{lemma} 

\subsubsection{Improved Type I error bound under low effective dimension}

For this setting, our main result given in Theorem~\ref{theorem_general} can be strengthened to the following
tighter bound.

\begin{theorem}\label{theorem_sparse}
	Under the notation and assumptions of Theorem~\ref{theorem_general}, 
	suppose it also holds that
	\[\P\{\|\hat\theta(X,W) - \theta_0\|_{v,0}\leq k(\theta_0)\} \geq 1 - \tilde\delta(\theta_0),\]
	for a fixed set of vectors $v_1,\dots,v_p\in\R^d$.
	Then  the copies $\tilde{X}^{(1)}, ..., \tilde{X}^{(M)}$
	generated by the constrained aCSS procedure are approximately exchangeable with $X$, satisfying	
	\begin{equation*}
		d_{\textnormal{exch}}(X, \tilde{X}^{(1)}, ..., \tilde{X}^{(M)})\le 3\sigma r(\theta_{0})  \cdot \sqrt{\frac{h_v(k(\theta_0))}{d}}+ \epsilon(\theta_{0}) + \delta(\theta_{0})+\tilde\delta(\theta_{0}).
	\end{equation*}
	In particular, this implies that for any predefined test statistic $T: \X \rightarrow \R$ and rejection threshold $\alpha\in[0,1]$, the p-value defined in~\eqref{eq:pval} satisfies
	\begin{equation*}
		\P\left(\textnormal{pval}_{T}(X, \tilde{X}^{(1)}, ..., \tilde{X}^{(M)})\le \alpha \right)\le \alpha + 3\sigma r(\theta_{0}) \cdot \sqrt{\frac{h_v(k(\theta_0))}{d}}+ \epsilon(\theta_{0}) + \delta(\theta_{0})+\tilde\delta(\theta_{0}).
	\end{equation*}
\end{theorem} 

As discussed above, a small value of $k(\theta_0)$ indicates that the error vector, $\hat\theta-\theta_0$, typically lies
in a region of $\R^d$ that is characterized by a lower effective dimension. As another interpretation, we can think of $k(\theta_0)$ as
capturing the effective degrees of freedom in our estimation problem. 

The result of Theorem~\ref{theorem_sparse} is strictly stronger than that of Theorem~\ref{theorem_general}.
In particular, Theorem~\ref{theorem_general} can be derived as a special case, by taking $v_1=\mathbf{e}_1,\dots,v_d=\mathbf{e}_d$
and $k(\theta_0)\equiv d$---then the additional condition of Theorem~\ref{theorem_sparse} holds trivially with $\tilde{\delta}(\theta_0)=0$,
and so the
two theorems give the same bound (since $h_v(d) = d$). On the other hand, if the constrained estimation problem
exhibits sparsity relative to the chosen set of vectors $\{v_i\}$, we may be able to choose a value $k(\theta_0)\ll d$ that allows for a low value of $\tilde{\delta}(\theta_0)$;
in this setting, $h_v(k(\theta_0))\ll d$ by Lemma~\ref{lemma:proj_chisq}, and consequently,
we see that we can afford to choose a much larger value of the perturbation noise parameter $\sigma$
while still retaining approximate Type I error control. Of course, to have $k(\theta_0)\ll d$ (or equivalently, $h_v(k(\theta_0))\ll d$), we need to choose a suitable set $\{v_i\}$  that corresponds well to the structure induced by the constraints $A \theta \le b$, as in the examples given above.

\begin{remark}
	As we will see in the proof,	the result of Theorem~\ref{theorem_sparse} holds even if we replace Assumption~\ref{asm:L_H} with a weaker condition: defining \[\Theta_0 = \{\theta\in\Theta : \|\theta-\theta_0\|\leq r(\theta_0), \|\theta-\theta_0\|_{v,0}\leq k(\theta_0)\},\]
	and writing $\theta_t = (1-t)\theta_0 + t\theta$ for any $\theta$, it suffices to assume
	\[
	\E_{\theta_0} \left[\sup_{\theta\in \Theta_0, t\in[0,1]} \left((\theta - \theta_0)^\top\left(H(\theta_t) - H(\theta_t; X) \right)(\theta - \theta_0)\right)_{+} \right]\le \epsilon(\theta_{0}),
	\]
	and
	\[
	\log\E_{\theta_{0}}\left[\exp\left\{\sup_{\theta\in \Theta_0, t\in[0,1]} \left((\theta - \theta_0)^\top\left(H(\theta_t;X) - H(\theta_t) \right)(\theta - \theta_0)\right)_{+}\right\}\right]\le\epsilon(\theta_{0}).
	\]
	in place of conditions~\eqref{eqn:Hessian_asm_1} and~\eqref{eqn:Hessian_asm_2}, respectively.
	That is, we only need to establish concentration of the error 
	in the Hessian along directions $\theta-\theta_0$ that have sparse structure with respect to the chosen vectors $\{v_i\}$, which may
	be a much more feasible condition in high-dimensional settings.	
\end{remark}

\subsection{Special case: Gaussian linear model}\label{sec:theory_gaussian}
In this section, we turn to another setting where the scaling of our result has a much more favorable dependence on dimension $d$,
for the special case of a Gaussian linear model. Unlike the result in Theorem~\ref{theorem_sparse} above, here we do not need to assume an underlying sparse structure.

For this special case, we assume that the parametric family $\{P_\theta\}$ is given by
\begin{equation}\label{eq:gaussian_model}
	P_\theta: \ X \sim \mathcal{N}(Z\theta, \nu^2\I_n)\end{equation}
where both the covariate matrix $Z\in\R^{n\times d}$ and the variance $\nu^2>0$ are fixed and known.
This model is parametrized by the coefficient vector, $\theta\in\Theta = \R^d$. 
In this setting, as described earlier in Section~\ref{sec:CSS}, co-sufficient sampling (CSS) can be directly applied to sample copies $\tilde{X}^{(m)}$ that are
\emph{exactly} exchangeable with $X$. Concretely, we can consider the sufficient statistic $\mathcal{P}_Z X$,
where $\mathcal{P}_Z\in\R^{n\times n}$ denotes the  projection matrix to the column span of $Z$, and sample the copies as 
\[\tilde{X}^{(m)}\mid \mathcal{P}_ZX \stackrel{\textnormal{iid}}{\sim} \mathcal{N}(\mathcal{P}_ZX, \nu^2\mathcal{P}_Z^\perp).\]
Then, under the null, $(X,\tilde{X}^{(1)},\dots,\tilde{X}^{(M)})$ is exchangeable, and so the p-value defined in~\eqref{eq:pval} is \emph{exactly} valid
for any test statistic $T$.

In a low-dimensional regime where $n > d$, the copies $\tilde{X}^{(m)}$ are distinct from $X$, and the resulting
test can have high power against the alternative for a suitably chosen statistic $T$. However, in the high-dimensional setting
with $d\geq n$, we will have $\mathcal{P}_Z = \I_n$, leading to copies $\tilde{X}^{(m)}$ that are identical to $X$
and, therefore, a powerless test. In the high-dimensional setting, therefore, we turn to aCSS as a practical alternative that can 
offer nontrivial power, while sacrificing some Type I error control.

The challenge for applying aCSS is that, as we are in a high-dimensional setting, the estimator $\hat\theta$ may have low accuracy---but 
we need a tight bound $r(\theta_0)$ on its error in order to achieve approximate Type I error control. In many settings, 
the accuracy of the estimator $\hat\theta$ will be greatly improved by adding constraints that reflect structure in the problem (e.g., 
an $\ell_1$ constraint if we believe $\theta_0$ is sparse), and so we would expect that constrained aCSS can offer a strong advantage
in this setting. 

However, the power of the method will rely on being able to choose a sufficiently large value of $\sigma$ in the implementation.
We 
are therefore motivated to develop a theoretical guarantee that is stronger than the general result of Theorem~\ref{theorem_general},
so that we can choose a higher value of $\sigma$ and, consequently, achieve higher power.
We will now see that the Gaussian case offers both computational and theoretical advantages.

First, we will assume that $R$ is chosen to ensure that the loss has strongly positive definite Hessian,
i.e., \begin{equation}\label{eqn:assume_R}
	\frac{1}{\nu^2}Z^\top Z + \nabla^2_\theta R(\theta)\succ c\I_d\textnormal{ for all $\theta\in\R^d$, for some $c>0$.}\end{equation}
For example, if $n\geq d$ and $Z$ has full rank $d$, then this holds with $R(\theta)\equiv 0$. More generally, for any $d,n$ and any $Z$, a ridge penalty $R(\theta) = \frac{\lambda_{\textnormal{ridge}}}{2}\|\theta\|^2$ (for some positive penalty parameter
$\lambda_{\textnormal{ridge}}>0$) will ensure that this condition holds. 

Then $\hat\theta$ is defined by the optimization problem
\[\hat\theta=\hat\theta(X,W) = \argmin_{\theta\in\R^d}\left\{ \frac{1}{2\nu^2}\|X-Z\theta\|^2 +  R(\theta)+ \sigma W^\top  \theta \  : \ A\theta\leq b \right\},\] and we compute the gradient as
\[\hat g = \frac{1}{\nu^2}Z^\top (Z\hat\theta - X)+ \nabla_\theta R(\hat\theta)+ \sigma  W.\]
Note that, by our assumptions on $R$, this optimization problem is guaranteed to have a unique minimizer, and moreover, this minimizer is guaranteed to be a SSOSP. In other words, we can assume that the event $X\in\X_{\hat\theta,\hat{g}}$ holds almost surely,  meaning that the indicator function in the sampling density is always equal to 1.
Then, applying Lemma~\ref{lemma:conditional_density}, we can compute the distribution $p_{\hat\theta}(\cdot\mid \hat\theta, \hat g)$ as
\begin{equation}\label{eqn:conditional_distribution_gaussian}
	\mathcal{N}\left(Z\hat\theta + \frac{d}{\sigma^2}\left(\I_{n}+\frac{d}{\sigma^2\nu^2}ZZ^\top \right)^{-1}Z(\nabla_\theta R(\hat\theta) - \hat{g}), \nu^2\left(\I_{n} + \frac{d}{\sigma^2\nu^2}ZZ^\top \right)^{-1}\right).
\end{equation} 
This means that it is possible to draw the copies $\tilde{X}^{(1)},\dots,\tilde{X}^{(M)}$ directly as i.i.d.\ draws from $p_{\hat\theta}(\cdot\mid \hat\theta, \hat g)$.

Next we turn to our theoretical guarantee, which shows an $O(\sqrt{d})$ improvement in the excess Type I error for the Gaussian case.	
\begin{theorem} 
	\label{theorem_gaussian}
	Consider the Gaussian linear model~\eqref{eq:gaussian_model}, and assume
	that  $R(\theta)$ is chosen so that condition~\eqref{eqn:assume_R} is satisfied. Assume also that $\P\{\|\hat\theta(X,W)-\theta_0\|\leq r(\theta_0)\}\geq 1-\delta(\theta_0)$.
	Then the copies $\tilde{X}^{(1)}, ..., \tilde{X}^{(M)}$
	generated by the constrained aCSS procedure are approximately exchangeable with $X$, satisfying	 
	\begin{equation*}
		d_{\textnormal{exch}}(X, \tilde{X}^{(1)}, ..., \tilde{X}^{(M)})\le \frac{\sigma}{2\sqrt{d}} r(\theta_0)+ \delta(\theta_{0}).
	\end{equation*}
	In particular, this implies that for any predefined test statistic $T: \X \rightarrow \R$ and rejection threshold $\alpha\in[0,1]$, the p-value defined in~\eqref{eq:pval} satisfies
	\begin{equation*}
		\P\left(\textnormal{pval}_{T}(X, \tilde{X}^{(1)}, ..., \tilde{X}^{(M)})\le \alpha \right)\le \alpha + \frac{\sigma}{2\sqrt{d}} r(\theta_0)+ \delta(\theta_{0}).
	\end{equation*}
\end{theorem}

The Type I error inflation described above offers an improvement by a factor of $O(\sqrt{d})$ in terms of dependence on $\sigma$, when compared to Theorem~\ref{theorem_general}.  In other words,
we see that we are free to choose a substantially larger $\sigma$ in this Gaussian setting to increase power without losing the guarantee of approximate Type I error control. 

 We also note from the proof that Theorem \ref{theorem_gaussian} does not depend on the specific form of $\mathcal{L}(\theta, X)$, due to the explicit total variation bound between two Gaussian distributions. That is, the theorem holds for any $\mathcal{L}(\theta, X)$ beyond just the negative log-likelihood. However, in the general case—whether or not a sparse structure is present—the Type I error control results in Theorems \ref{theorem_general} and \ref{theorem_sparse} do rely on $\mathcal{L}(\theta, X)$ being the negative log-likelihood.  In practice, though, when applying aCSS, one can still choose any suitable $\mathcal{L}(\theta, X)$, since Lemma \ref{lemma:conditional_density} for sampling from the conditional distribution remains applicable, as previously discussed.

 \begin{remark}[Practical considerations]\label{rmk:unknown_var}
		If the error variance $\nu$ is unknown, it can be treated as part of the unknown parameters $\theta$, and the general constrained (or regularized) aCSS procedure can still be applied.  That is, we optimize the objective function and compute the gradient with respect to both the regression coefficients and the error variance. Lemma \ref{lemma:conditional_density} (on conditional density) and Theorems \ref{theorem_general} and \ref{theorem_sparse} (on Type I error control) still hold when the parameter includes both the coefficients and the variance.
		
		For simpler closed-form sampling in practice,  we may consider perturbations only in the coefficients---similar to the known-variance case---and make a slight modification to the estimation (optimization) step. To align with the known-variance case, we continue to use $\theta$ to denote the coefficients and $g$ to denote the gradient with respect to the coefficients.
		We first solve for $\hat\theta$ by optimizing the objective without including the error variance, which is indeed more common in coefficient estimation,
		\[\hat\theta=\hat\theta(X,W) = \argmin_{\theta\in\R^d}\left\{\L(\theta; x, w) = \frac{1}{2}\|X-Z\theta\|^2 +  R(\theta)+ \sigma W^\top  \theta \  : \ A\theta\leq b \right\},\] 
		and compute the gradient as
		\[\hat g =  Z^\top (Z\hat\theta - X)+ \nabla_\theta R(\hat\theta)+ \sigma  W.\] 
		We further  estimate the variance via:  \[\hat\nu = \sqrt{\frac{1}{n}\|X - Z\hat\theta\|^2}.\]  
		Since  $\hat\nu$  is determined by $X$ and $\hat\theta$, the conditional density satisfies:
		\[p_{\theta, \nu}(\cdot\mid \hat\theta, \hat{g},\hat{\nu})\propto p_{\theta, \nu}(\cdot\mid \hat\theta, \hat{g}) \id_{\{\|X - Z\hat\theta\|^2 = n\hat{\nu}^2\}}.\]
		Combined with Lemma~\ref{lemma:conditional_density}, the  conditional density $p_{\hat\theta, \hat\nu}(\cdot\mid \hat\theta, \hat g, \hat\nu)$ is  proportional to 
		\begin{equation}\label{eqn:conditional_distribution_gaussian_unknownvar}
			\exp\{-\frac{1}{2}(x - \mu)^{\top}\Sigma^{-1}(x-\mu)\}\cdot \id_{\{\|x - Z\hat\theta\|^2 = n\hat{\nu}^2\}},
		\end{equation}
		where
		\[\mu = Z\hat\theta + \frac{d}{\sigma^2}\left(\frac{1}{\hat\nu^2}\I_{n}+\frac{d}{\sigma^2}ZZ^\top \right)^{-1}Z(\nabla_\theta R(\hat\theta) - \hat{g}), \ \Sigma =\left(\frac{1}{\hat\nu^2}\I_{n} + \frac{d}{\sigma^2}ZZ^\top \right)^{-1}. \]
		That is a Gaussian distribution constrained to the sphere $\{x\in\R^{n}: \|x - Z\hat\theta\|^2 = n\hat{\nu}^2\}$,   also known as the  Fisher--Bingham distribution.   We can efficiently sample from this distribution and use the generated samples to compute the p-value for our chosen test statistic  $T$. 
	\end{remark}

\section{Generalization of linear constraint: $\ell_{1}$ penalty}\label{sec:penalty}
Thus far, we have considered settings where the estimator $\hat\theta$ 
is obtained via a constrained optimization problem.  Section \ref{sec:theory} shows that the constraints introduced  can improve the estimation of unknown parameters, thereby leading to a tighter bound on Type I error control. One important example is placing a bound on $\|\theta\|_1$ to  encourage sparsity, a technique that is popular in high-dimensional settings. However, in many statistical applications, it is more common---and more effective---to use a $\ell_{1}$ penalty rather than a constraint.
Therefore, in this section, we will consider a $\ell_{1}$-penalized, rather than constrained, form of aCSS. 	 

We consider replacing the constrained optimization problem
\[
\hat\theta_C = \argmin_{\theta\in\Theta} \{\L(\theta; X, W)  :\|\theta\|_1\leq C\}
\]
with its penalized version, 
\begin{equation}\label{eqn:def_thetahat_l1pen} 
	\hat\theta_\lambda = \argmin_{\theta\in\Theta} \{\L(\theta; X, W)  + \lambda\|\theta\|_{1}\},
\end{equation} 
(i.e., the lasso \citep{tibshirani1996regression}, but with an added perturbation term due to $W$).
The penalized and constrained forms of the optimization problem have a natural correspondence---for $\ell_1$ regularization, each constrained solution $\hat\theta_C$ corresponds to some penalized solution $\hat\theta_\lambda$
for some data-dependent $\lambda$, and vice versa.
However, in a statistical analysis, these two versions of the problem often behave very differently: for $\ell_1$ regularization, the fact that the correspondence
between $C$ and $\lambda$ is data-dependent means that theoretical results obtained for $\hat\theta_\lambda$ at a fixed $\lambda$ do not transfer
over to a theoretical guarantee for $\hat\theta_C$ for a fixed $C$, and vice versa. Therefore, proper modification is needed for the $\ell_{1}$-penalized aCSS.

Before state the modified method, we first define SSOSP for the penalized problem. For $\theta\in\R^d$, we will write $S(\theta) = \{j\in[d]: \theta_j\neq 0\}$ to denote the support of $\theta$. 
\begin{definition}[SSOSP for the $\ell_{1}$-penalized problem]\label{def:ssosp-pen} A parameter $\theta\in\Theta$ is a strict second-order stationary point (SSOSP) of the optimization
	problem \eqref{eqn:def_thetahat_l1pen}	if it satisfies all of the following:		
	\begin{enumerate}
		\item  First-order necessary conditions, i.e., Karush--Kuhn--Tucker (KKT) conditions: 
		\[\nabla\L(\theta;X,W)+ \lambda s = 0, \textnormal{ \  where } 
		\begin{cases}
			s_j=			\textnormal{sign}(\theta_j),& j\in S(\theta),\\
			s_j\in		 [-1,1],&  j\not\in S(\theta).
		\end{cases}
		\]
		\item   Second-order sufficient condition: 
		\[ \nabla^2_{\theta}\L(\theta; X,W)_{S(\theta)}\succ 0,\]
		where for a matrix $M\in\R^{d\times d}$ and a nonempty subset $J\subseteq[d]$, $M_{J}\in\R^{|J|\times|J|}$ denotes the submatrix of $M$ restricted to row and column subsets $J$.
		That is, the Hessian $ \nabla^2_{\theta}\L(\theta; X,W) $ is strictly positive definite when restricted to the support of $\theta$.
	\end{enumerate}  
\end{definition}

\subsection{The conditional density in the penalized case}

Next we compute the conditional density  of $X$ given $(\hat\theta, \hat{g})$. We will see that this calculation looks quite similar to the constrained case (which was addressed in Lemma~\ref{lemma:conditional_density}).
\begin{lemma}[Conditional density for the $\ell_1$-penalized case]
	\label{lemma:conditional_density_l1pen}
	Suppose Assumption \ref{asm:reg} holds. Fix any $\theta_{0}\in\Theta$ and let $(X, W, \hat{\theta}, \hat{g})$ be drawn from the joint model
	\begin{equation*}
		\left\{ \begin{array}{l}
			X\sim P_{\theta_{0}},\\
			W\sim \mathcal{N}(0, \frac{1}{d}\I_{d}),\\
			\hat\theta = \hat\theta(X, W) \\
			\hat{g} = \hat{g}(X, W) = \nabla_{\theta}\L(\hat\theta; X, W).
		\end{array}\right.
	\end{equation*}
	Let $S\subseteq[d]$.	Assume that the event that $\hat{\theta}(X, W)$ is a SSOSP of \eqref{eqn:def_thetahat_l1pen} with support $S(\hat\theta(X, W))  = S$ has positive probability.  Then, conditional on this event, the conditional distribution of $X| \hat{\theta}, \hat g$ has density  
	\begin{equation} 
		\label{eqn:conditional_density-pen}
		p_{\theta_{0}}(\cdot| \hat{\theta}, \hat g )\propto
		f(x; \theta_{0}) \exp\left\{-\frac{\|\hat{g} - \nabla_{\theta}\L(\hat\theta; x)\|^{2} }{2\sigma^{2}/d}\right\}\det\left(\nabla_{\theta}^2\L(\hat\theta; x)_S\right) 
		\id_{x\in\tilde{\X}_{\hat\theta,\hat{g}}}
	\end{equation} 	
	with respect to the base measure $\nu_{\X}\times$Leb, and  
	\[\tilde{\X}_{\theta,g} = \left\{x\in\X: \textnormal{ for some $w\in\R^d$, $\theta=\hat\theta(x,w)$ is a SSOSP of~\eqref{eqn:def_thetahat_l1pen}, and $g=\nabla\L(\theta;x,w)$}\right\}.\]
\end{lemma} 
Comparing to the analogous result given in Lemma~\ref{lemma:conditional_density} for the constrained case, we see that the only difference 
is in the $\det(\cdot)$ term: the density involves the determinant of a different matrix (namely, $U_{\A}^\top \nabla_{\theta}^2\L(\hat\theta; x)U_{\A}$ in the constrained case, and $\nabla_{\theta}^2\L(\hat\theta; x)_S$ in the penalized case). This is not merely a difference in notation: the matrices will actually have different dimension in the $\ell_1$-constrained and $\ell_1$-penalized settings, because under the constrained setting, if we know the support is $S$, the solution $\hat\theta$ effectively has $|S|-1$ degrees of freedom (due to the $\ell_1$ constraint which specifies the sum of the terms), in contrast to $|S|$ for the $\ell_1$-penalized setting.

\subsection{The aCSS method in the penalized case}
To implement an $\ell_1$-penalized version of aCSS, we can modify the constrained aCSS method in a straightforward way: we 
simply replace the constrained optimization problem~\eqref{eqn:def_thetahat} with the $\ell_1$-penalized optimization problem~\eqref{eqn:def_thetahat_l1pen}, and then proceed as before, using our new calculation for the conditional density as given
in Lemma~\ref{lemma:conditional_density_l1pen}. In particular, the copies $\tilde{X}^{(m)}$ will be sampled as
\[(\tilde{X}^{(1)},\dots,\tilde{X}^{(M)}) \mid (X,\hat\theta,\hat{g}) \sim \tilde{P}_M(\cdot;X,\hat\theta,\hat{g})\]
where $\{\tilde{P}_M(\cdot;x,\theta,g)\}$ is required to satisfy \eqref{eqn:define_tilde_P}, the same property as before, but now
relative to the conditional density $p_{\hat\theta}(\cdot\mid \hat\theta,\hat{g})$ calculated as
\begin{equation}
	\label{eqn:conditional_density_thetahat_l1pen}
	p_{\hat\theta}(\cdot\mid\hat{\theta}, \hat g )\propto f(x; \hat\theta)\cdot \exp\left\{-\frac{\|\hat{g} - \nabla_{\theta}\L(\hat\theta; x)\|^{2} }{2\sigma^{2}/d}\right\}\cdot\det\left(\nabla^{2}_{\theta}\L(\hat\theta; x)_{S(\hat\theta)}\right) \cdot
	\id_{x\in\X_{\hat\theta,\hat{g}}}.
\end{equation} 
As a special case, if computationally feasible, we can choose
\[\tilde{P}_M(\cdot;x,\hat\theta,\hat{g}) = p_{\hat\theta}(\cdot\mid \hat\theta,\hat{g}) \times\dots\times p_{\hat\theta}(\cdot\mid \hat\theta,\hat{g}),\]
i.e., sampling the copies i.i.d.\ from the conditional density $p_{\hat\theta}(\cdot\mid \hat\theta,\hat{g})$ defined in~\eqref{eqn:conditional_density_thetahat_l1pen}.

Formally, the algorithm is defined as follows.  The bold text highlights the only modifications in the algorithm, relative to constrained aCSS.

\begin{quote}
	\normalsize
	\textbf{$\ell_1$-penalized aCSS algorithm:}
	\begin{enumerate}
		\item Observe data $X\sim P_{\theta_0}$.
		\item Draw noise $W\sim\mathcal{N}(0,\frac{1}{d}\I_d)$.
		\item \textbf{\normalsize Solve for an $\ell_1$-penalized perturbed MLE $\hat\theta = \hat\theta(X,W)$ as in~\eqref{eqn:def_thetahat_l1pen}.}\\
		Compute the corresponding gradient $\hat{g}=\hat{g}(X,W)$ as in~\eqref{eqn:def_ghat}.
		\item If $\hat\theta$ is not a SSOSP of~\eqref{eqn:def_thetahat}, then set $\tilde{X}^{(1)}=\dots=\tilde{X}^{(M)}=X$.
		Otherwise, sample copies $(\tilde{X}^{(1)},\dots,\tilde{X}^{(M)}) \mid (X,\hat\theta,\hat{g}) \sim \tilde{P}_M(\cdot;X,\hat\theta,\hat{g})$, where $\tilde{P}_M$
		is chosen to satisfy property~\eqref{eqn:define_tilde_P} \textbf{\normalsize relative to the conditional density $p_{\hat\theta}(\cdot\mid \hat\theta,\hat g)$ as computed
			in~\eqref{eqn:conditional_density_thetahat_l1pen}.}
		\item Compute the p-value defined in~\eqref{eq:pval}
		for our choice of test statistic $T$.
	\end{enumerate}
\end{quote}

In contrast to the typical challenges for translating results between the constrained and penalized form of a regularized estimation problem, in the context of aCSS, both the conditional density in Lemma~\ref{lemma:conditional_density_l1pen} and our next result establish that the exact same results can be obtained for the $\ell_1$-penalized case. This unusually favorable behavior is due to the fact that aCSS operates conditionally on the solution $\hat\theta$---effectively, once we condition on $\hat\theta$, we no longer face the challenge of the data-dependent correspondence between the penalty parameter $\lambda$ versus the constraint parameter $C$, since both values are revealed by $\hat\theta$ itself.

\begin{theorem}\label{theorem_penalized}
 	The results of Theorems~\ref{theorem_general}, \ref{theorem_sparse}, and \ref{theorem_gaussian}
		all hold for the $\ell_1$-penalized form of aCSS in place of constrained aCSS, under the same assumptions—except that in Assumption \ref{asm:SSOSP_general}, the estimator $\hat{\theta}(X, W)$ is assumed to be a SSOSP of the $\ell_1$-penalized problem \eqref{def:ssosp-pen}.

\end{theorem}

In the context of utilizing the  $\ell_1$ penalty, it is commonly the case that the parameter is high-dimensional and sparse. This naturally directs our attention towards Theorem~\ref{theorem_sparse}, which offers the most relevant insights for this scenario.  Specifically, we can select the set of vectors $\{v_{i}\}$ as the canonical basis $\{\mathbf{e}_{i}\}_{i=1, ..., d}$. Then we have $\|w\|_{v,0} = \|w\|_{0}$ (i.e., the cardinality of the support of $w$). 
The result of Theorem~\ref{theorem_sparse} then gives a much stronger bound on the excess Type I error rate, as long as 
we can assume that
\[\|\hat\theta - \theta_0\|_0\leq k(\theta_0)\]
holds with high probability. This is very favorable for the $\ell_1$ penalized setting: if $\theta_0$ itself is sparse,
then the sparsity of $\hat\theta$ (which is ensured by the $\ell_1$ penalty) means that the difference $\hat\theta - \theta_0$ will also be sparse.

\section{Numerical experiments}\label{sec:numerical}
In this section, we will study the performance of aCSS with regularization on three simulated examples.\footnote{Code for reproducing all experiments
	is available at \url{http://rinafb.github.io/code/reg_acss.zip}.}
The first, Example \ref{example:mixgaussian}, is a Gaussian mixture model, which showcases a scenario where constraints on the parameters being estimated are essential to ensure the existence of a well-defined MLE. In the remaining examples, Example \ref{example:isotonic} (isotonic regression) and Example \ref{example:sparse} (sparse regression), 
we shift our focus to a high-dimensional Gaussian linear model, where the imposition of suitable constraints or penalties can allow for accurate estimation despite high dimensionality.

\subsection{Necessary constraints: the Gaussian mixture model}\label{sec:gmm}
In this section, we will examine the Gaussian mixture model example, where constraints are needed for ensuring the existence of a well-defined MLE.
\begin{example}[Gaussian mixture model]\label{example:mixgaussian}
	Suppose  we observe data   from  the Gaussian mixture model with a known number of components $J$,
	\[
	X_{1}, ..., X_{n} \overset{\textnormal{i.i.d.}}{\sim}\sum_{j=1}^{J}\pi_{j}\mathcal{N}(\mu_{j}, \eta_{j}^{2}),
	\]
	where $\{\pi_{j}\}_{j\in[J]}$ are the weights on the components, with $\pi_j > 0$ and $\sum_j \pi_j=1$.
	The family of distributions $\{P_{\theta}\}_{\theta\in\Theta}$ is parameterized by $\theta = (\pi_{1}, ..., \pi_{J-1}, \mu_{1}, \eta_{1}, ..., \mu_{J}, \eta_{J})\in\Theta$ where
	\[\Theta = \{t \in \R_+^{J-1} : \sum_i t_i < 1\} \times (\R \times \R_+)^J.\] Consequently we have $\Theta\subseteq\R^{d}$ with $d = 3J-1$. The density of  $P_{\theta}$, the distribution on the data $X=(X_1,\dots,X_n)$, is thus given by
	\begin{equation}\label{eqn:mixgaussian_density}
		f(x;\theta) = \prod_{i=1}^{n}\sum_{j=1}^{J}\pi_{j}\phi(x_{i};\mu_{j},\eta_{j}^2),\end{equation}
	where $\phi(\cdot;\mu,\eta^2)$ is the density of the normal distribution with mean $\mu$ and variance $\eta^2$.
\end{example}

Why is constrained aCSS useful for this example?
The Gaussian mixture model does not possess straightforward, compact sufficient statistics due to the presence of unobserved latent variables (i.e., 
identifying which of the $J$ components corresponds to the draw of each data point $X_i$).  Any sufficient statistic would reveal essentially all the information about the data $X$.
However, if we attempt to apply aCSS (without constraints), we are faced with a fundamental challenge:
the MLE does not exist for this model, because the likelihood approaches infinity if, for any component $j$, we take $\mu_j= X_i$ for some observation $i\in[n]$ and take $\eta_j\rightarrow 0$.
To prevent this divergence of the likelihood, one can impose a lower bound on the component variances, requiring $\eta_{j}\ge  c$ for each $j\in[J]$, where $c>0$ is some small constant. Under this restriction,  it can be shown that MLE is strongly consistent if the true parameter lies within the restricted parameter space  \cite{Tanaka2006strong}.
Then the constrained  aCSS framework is indeed suitable when generating sampling copies in the context of this example. 
   Existing methods for testing mixture models are primarily based on the likelihood ratio test (LRT), but they have certain limitations. The profiled LRT \citep{chen2014likelihood} and the EM test \citep{10.1214/08-AOS651} have tractable limiting distributions; however, both are restricted to two-component mixture models and assume that one of the components corresponds to the null distribution, which does not align with the setup considered in our example. The bootstrap LRT \citep{mclachlan1987bootstrapping} is commonly used but lacks finite-sample guarantees. Universal inference \citep{wasserman2020universal}, based on the split (or crossfit) LRT, is applicable to mixture models with multiple components; however, data splitting in these approaches often leads to a loss of power. 
	We later compare our method to this approach.  
As we will show in Appendix C,
for an appropriately-chosen initial estimator this example satisfies Assumptions~\ref{asm:reg},~\ref{asm:SSOSP_general}, and~\ref{asm:L_H} with $r(\theta_{0}) = O(\sqrt{\log{n}/n})$, $\delta(\theta_{0}) = O(n^{-1})$, and $\epsilon(\theta_{0}) = O( \sqrt{\frac{\log^3 n}{n}})$, as long as we assume $(\mu_1)_0\neq (\mu_2)_0$, i.e., the two components have distinct means under the true parameter $\theta_0$.
Therefore, Theorem~\ref{theorem_general} implies that constrained aCSS will have approximate Type I error control for this example.

\subsubsection{Simulation: setting}  We next examine the empirical performance of constrained aCSS for the Gaussian mixture model (Example~\ref{example:mixgaussian}).
For this setting, we will compare the null hypothesis of a Gaussian mixture model with $J=2$ components, against an alternative where there are more (specifically, $3$) 
components.
The setup of the simulation is summarized as follows:
\begin{itemize} 
	\item  To generate data, we take $n = 200$,  and draw the data points $X_1,\dots,X_n$ from a mixture of Gaussians  \[\pi_{0}\mathcal{N}(0, 0.01) + \frac{1 - \pi_{0}}{2}\mathcal{N}(0.4, 0.01) + \frac{1 - \pi_{0}}{2}\mathcal{N}(-0.4,0.01).\]
	\item Our null hypothesis is a mixture of \emph{two} Gaussians (i.e., a density of the form~\eqref{eqn:mixgaussian_density} with $J=2$). The data generating distribution above therefore corresponds to the null hypothesis~\eqref{eqn:mixgaussian_density}
	with parameter 
	\[\theta_{0} =  (\pi_1,\mu_1,\eta_1,\mu_2,\eta_2) = (0.5, \ 0.4, \ 0.1, \ -0.4, \ 0.1)\]
	in the case that $\pi_0 = 0$, while if $\pi_0>0$ then the null hypothesis is not true.
	\item  We enforce $r=2$ constraints, given by $\eta_{j}\ge 0.098$, $j=1,2$. (We choose the bound slightly below the true value $\eta_j = 0.1$, so that a reasonable proportion of constraints are active---this way, running our constrained aCSS procedure is meaningfully different than running unconstrained aCSS.) Constrained aCSS is then run with  noise level $\sigma = 8$, and  $M=300$ copies $\tilde{X}^{(m)}$, sampled via MCMC (see Appendix D of the Supplement for details).
	\item We compare constrained aCSS to the oracle method, which uses the same test statistic $T$ but is
	given full knowledge of the distribution of $X$ under null hypothesis, i.e., $P_{\theta_0} = 0.5\mathcal{N}(0.4, 0.01) + 0.5\mathcal{N}(-0.4,0.01)$, and can therefore sample the copies $\tilde{X}^{(m)}$ i.i.d.\ from the known null distribution. 
	\item The test statistic $T$ (used both for aCSS and for the oracle) is chosen as the decrease in total within-cluster sum of squares of the k-means algorithm,
	when the number of estimated clusters  is increased from $2$ to $3$.

\end{itemize}

\subsubsection{Simulation: results}
The results of the simulation are shown in Figure \ref{fig:necessary}. We see that the constrained aCSS method is empirically valid as a test of $H_0$, since the rejection probability  when $\pi_0=0$ (i.e., when $H_0$ is true) closely matches the nominal level $\alpha=0.05$. Of course, the power of constrained aCSS is lower than that of the oracle method,
as is expected since the oracle is given knowledge of the true null parameter $\theta_0$; nonetheless, constrained aCSS shows a good increase in power as the signal strength $\pi_0$ grows. 
 
We also compare our method with the universal inference approach \citep{wasserman2020universal}. In this method, the data is split into two halves: one half is used to fit a three-cluster mixture model via the EM algorithm, while the other half is used to fit a two-cluster model. The cross-fitted likelihood ratio is computed as the test statistic and compared against $1/\alpha$ to perform the hypothesis test.  Our method outperforms universal inference, demonstrating higher power while maintaining valid Type I error.
\begin{figure}
\centering   
\includegraphics[width=0.45\linewidth]{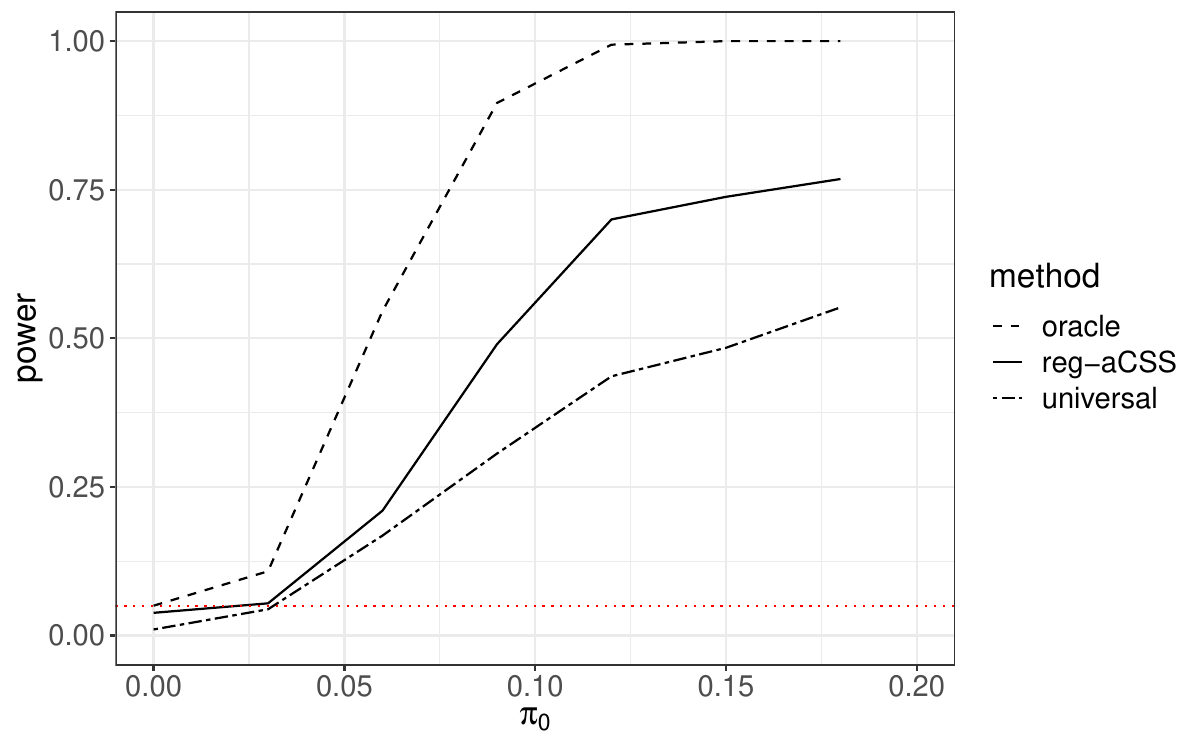}
\caption{ Power of the regularized  aCSS method, denoted as reg-aCSS, versus the oracle method and universal inference with crossfit LRT. The dotted red line denotes the nominal 5\% level. Results are based on 500 trials. $\pi_0=0$ corresponds to the null hypothesis being true.}
\label{fig:necessary}
\end{figure} 

\subsection{High dimensional setting: structured Gaussian linear model}\label{sec:example_hd}
We will now turn to the high-dimensional setting, where the data is distributed according to a Gaussian linear model with dimension $d\geq n$, 
\begin{equation*}  
	X\sim  \mathcal{N}(Z\theta, \nu^2\I_{n}), \ \textnormal{ with $Z\in\R^{n\times d}$, $\nu^2>0$ known},
\end{equation*} 
as in~\eqref{eq:gaussian_model}.
The family of distributions $\{P_{\theta}\}_{\theta\in\Theta}$ is parameterized  by  $\theta\in\Theta=\R^d$ and has density
\[
f(x;\theta) =   \frac{1}{(2\pi\nu^2)^{n/2}}e^{-\frac{\|x - Z\theta\|^2}{2\nu^2}}.
\]
In Section \ref{sec:theory_gaussian}, we examined the limitations of CSS testing, which will be powerless for this problem when $d\geq n$, 
as the copies $\tilde{X}^{(m)}$ will be identically equal to $X$.
We can instead run the aCSS method; however, the results of \cite{barber2020testing} indicate that the inflation in Type I error will scale
with our estimation error $\|\hat\theta - \theta_0\|$, which will in general be large when $d\geq n$, since the estimator $\hat\theta$ is computed with an unregularized maximum
likelihood estimation problem. (More precisely, aCSS does allow for a \emph{smooth} regularizer $R(\theta)$, such as a ridge penalty; however, it is challenging to achieve
accurate estimation in a high-dimensional setting unless we use \emph{nonsmooth} regularization, e.g., the $\ell_1$ norm).

In contrast, our proposed version of aCSS allows for constraints (or penalties) that allow us to achieve an accurate estimator $\hat\theta$, and consequently low Type I error,
in the high-dimensional setting. We now consider two specific examples where the application of appropriate regularization assists in the estimation process.

\begin{example}[Isotonic regression]\label{example:isotonic} 
	In the isotonic regression model, we are given a noisy observation $X\in\R^{n}$ of some monotone increasing signal $\theta_0\in \R^{n}$ with  
	\[(\theta_0)_{1}\le\cdots\le (\theta_0)_{n}.\]
	If the noise is Gaussian, with $X\sim\mathcal{N}(\theta_0,\nu^2\I_n)$, then this model  is a special case of the Gaussian linear model with $d = n$ and $Z = \I_{n}$.
\end{example}

To run constrained aCSS,  the perturbed  isotonic (least squares) regression is given by
\[
\hat\theta_{\textnormal{iso}} = \arg \min_{\theta\in\R^{n}} \{\L(\theta; X, W) \ : \ \theta_{1}\le\cdots\le \theta_{n}\}
\] 
to estimate the underlying signal.   
\cite{zhang2002risk} demonstrated that the isotonic least squares estimator (LSE),
which is given by minimizing $\|\theta-X\|$ subject to the constraints $\theta_1\leq\dots\leq\theta_n$, has an error rate scaling as  $\|\hat\theta-\theta_0\| = O(n^{1/6})$ (and choosing a sufficiently small $\sigma$ means that the perturbation will not substantially inflate this rate). This rate matches the minimax rate over the class of monotone and Lipschitz signals  \citep{chatterjee2015risk}. Thus, adding the monotonicity constraint will substantially reduce the error $\|\hat\theta-\theta_0\|$, which can help control the excess Type I error for our setting. In Appendix C, we will see that 
this example satisfies Assumptions~\ref{asm:reg},~\ref{asm:SSOSP_general}, and~\ref{asm:L_H} with $r(\theta_0) = O\left(n^{1/6}(\log n)^{1/3}\right)$, $\delta(\theta_0) = 1/n$, and $\epsilon(\theta_{0}) = 0$, if we choose $\sigma = O(1)$.			 Therefore, Theorem~\ref{theorem_gaussian} implies that constrained aCSS will have approximate Type I error control for this example.

Next, we examine a high-dimensional setting with a sparse parameter. 
\begin{example}[Sparse regression]\label{example:sparse}
	Let $d>n$, and let $Z\in\R^{d\times n}$ be a fixed covariate matrix. We assume the model	
	\[				X\sim  \mathcal{N}(Z\theta, \nu^2\I_{n}),\]
	for a known noise level $\nu^2$. This model is unidentifiable without further assumptions, but 
	becomes identifiable once we assume $\theta_0$ is sparse---specifically, as long as $Z$ satisfies some standard conditions (e.g., a restricted eigenvalue
	assumption). We will assume  that the underlying parameter $\theta_0$ is sparse, with
	\[\|\theta_0\|_{0}\le k\]
	for some sparsity bound $k$.
\end{example}

To address the problem of estimating a sparse $\theta_0$ in a linear model, the Lasso estimator \citep{tibshirani1996regression}, which combines the least squares loss with an $\ell_{1}$ penalty, is frequently employed. Under certain conditions, the error rate of the Lasso estimator can be on the order of $O( \sqrt{k\log(d)/ n})$ \citep{bickel2009simultaneous, hastie2015statistical}. Thus the perturbed Lasso is a suitable candidate for the estimator  in this context: for a given penalty level $\lambda>0$, we define
\[
\hat\theta_{\textnormal{lasso}} = \arg \min_{\theta\in\R^{d}} \{\L(\theta; X, W) + \lambda \|\theta\|_{1}\}.
\] 
In Appendix C, we will see that 
this example satisfies Assumptions~\ref{asm:reg},~\ref{asm:SSOSP_general}, and~\ref{asm:L_H} with $r(\theta_0) =  O(\sqrt{k\log d/n})$, $\delta(\theta_0) = 1/n$, and $\epsilon(\theta_{0}) = 0$, under suitable conditions.		 Therefore, Theorem~\ref{theorem_penalized} implies that constrained aCSS will have approximate Type I error control for this example.

\subsubsection{Simulation: setting} \label{sec:simu_setting_gaussian}
In this section, we demonstrate the advantage of regularized aCSS in high-dimensional settings.
Specifically, we will compare against the (unconstrained) aCSS method of \cite{barber2020testing}, to see how adding regularization allows for better estimation---consequently,
we can allow a high value of $\sigma$ without losing (approximate) Type I error control, which in turn leads to higher power.

For the isotonic regression setting  (Example~\ref{example:isotonic}), we will compare the null hypothesis that $X$ is given by an isotonic signal $\theta_0$ plus Gaussian noise, against
the alternative where $X$ also has dependence on an additional random variable $Y$. (Equivalently, we can take our covariate matrix $Z$ to be the identity, $Z=\I_d$, with $d=n$.) 
The setup of the simulation for isotonic regression is as follows:
\begin{itemize} 
	\item To generate data,
	we take $n = 100$, $\nu = 1$, and set the signal $\theta_{0}$ as
	\[\theta_0 = (0.1,\dots,0.1,0.2,\dots,0.2,\dots,1,\dots,1),\]
	with each value $0.1,0.2,0.3,\dots,1$ appearing 10 times. We then generate $X\sim\mathcal{N}(\theta_0,\nu^2\I_n)$. The additional random vector $Y$ is then drawn as 
	: 	\[Y\mid X \sim \mathcal{N}(\beta_0 X + (1-\beta_0)\theta_0,\I_n),\]
	where $\beta_0 \in \{0,0.05,0.1,\dots,0.5\}$, with $\beta_0=0$ corresponding to the null hypothesis.
	Formally, our null hypothesis
	is given by assuming that $X\mid Y\sim\mathcal{N}(\theta,\nu^2\I_n)$ for some $\theta\in\Theta=\R^n$, i.e., that the Gaussian model
	for $X$ is true even after conditioning on $Y$. If $\beta_0\neq 0$, then this null hypothesis does not hold.

	\item 
	For \cite{barber2020testing}'s aCSS method, $\hat\theta$ is computed via perturbed and unconstrained maximum likelihood estimation,
	\[\hat\theta = \hat\theta_{\textnormal{OLS}} = \argmin_{\theta\in\R^n} \left\{\frac{1}{2}\|X - \theta\|^2  + \sigma W^\top \theta\right\}. \]
	For our proposed constrained aCSS method, $\hat\theta$ is computed with the isotonic constraint,
	\[\hat\theta = \hat\theta_{\textnormal{iso}}  = \argmin_{\theta\in\R^n} \left\{\frac{1}{2}\|X - \theta\|^2  + \sigma W^\top \theta \ : \ \theta_1\leq \dots \leq \theta_n\right\}. \]
	For both methods, we sample the copies $\tilde{X}^{(m)}$ directly from the conditional distribution \eqref{eqn:conditional_distribution_gaussian}. 
 (When $\nu$ is unknown,  further details are provided in Appendix D of the Supplement.)
	
	\item For the oracle method, we assume oracle knowledge of the parameter $\theta_0$ that defines the null distribution, and 
	sample the copies $\tilde{X}^{(m)}$ i.i.d.\ from $P_{\theta_0} = \mathcal{N}(\theta_0,\I_n)$.
	
	\item For all   methods, the test statistic $T$ is
	given by the absolute value of the sample correlation between $X$ and $Y$. 
	
\end{itemize}

For the sparse regression setting (Example~\ref{example:sparse}), we will compare the null hypothesis that $X\mid Z$ follows a (sparse) Gaussian linear model, against
the alternative where $X$ also has dependence on an additional random variable $Y$.
The setup of the simulation for sparse regression is as follows:
\begin{itemize} 
	\item To generate data, we set $n = 50$, $d = 100$, $\nu = 1$, and $\theta_0 = (5, 5, 5, 5, 5, 0, ..., 0)$. 
	The covariate matrix $Z\in\R^{n\times d}$ is generated with i.i.d.\ $\mathcal{N}(0,1/d)$ entries, and we draw $X\mid Z\sim \mathcal{N}(Z\theta_0,\nu^2\I_n)$. 
	The random vector $Y\in\R^n$ is then generated with each entry $Y_i$ drawn as
	\[Y_i \mid X_i,Z_i \sim\mathcal{N}(\beta_0 X_i + \sum_{j=1}^5 Z_{i,j} , 1).\]
	We consider $\beta_{0}\in \{0, 0.1, 0.2, ..., 1\}$ with $\beta_{0} = 0$ corresponding to the setting where $Y\indep X\mid Z$. Formally, our null hypothesis
	is given by assuming that $X\mid Y,Z\sim\mathcal{N}(Z\theta,\nu^2\I_n)$ for some $\theta\in\Theta=\R^d$. If $\beta_0\neq 0$, then this null does not hold.
	\item For \cite{barber2020testing}'s aCSS method, we will use a ridge regularizer, $R(\theta) = \frac{\lambda_{\textnormal{ridge}}}{2}\|\theta\|^2$, for parameter estimation. We define
	\[\hat\theta = \hat\theta_{\textnormal{ridge}} = \argmin_{\theta\in\R^d} \left\{\frac{1}{2}\|X - Z\theta\|^2 + \frac{\lambda_{\textnormal{ridge}}}{2}\|\theta\|^{2} + \sigma W^\top \theta\right\}.\] 
	Adding ridge regularization  allows for a unique solution $\hat\theta$,
	achieving strict second-order stationarity conditions, to avoid a trivial result where the method achieves zero power (as would be the case if the SSOSP conditions are never satisfied).  
	For our proposed $\ell_{1}$-penalized aCSS method, in order to be more comparable to aCSS, we also add the regularizer $R(\theta)$. This means that our estimator
	is given by the elastic net \citep{zou2005regularization}, incorporating both $\ell_1$ and $\ell_2$ penalization:
	\[\hat\theta=\hat\theta_{\textnormal{elastic-net}} =  \argmin_{\theta\in\R^d} \left\{\frac{1}{2}\|X - Z\theta\|^2 + \frac{\lambda_{\textnormal{ridge}}}{2}\|\theta\|^{2} + \lambda\|\theta\|_1+ \sigma W^\top \theta\right\}.\] 
	For both methods, we sample the copies $\tilde{X}^{(m)}$ directly from  the conditional distribution \eqref{eqn:conditional_distribution_gaussian}. 
 (When $\nu$ is unknown,  further details are provided in Appendix D of the Supplement.)
	\item For the oracle method, we assume oracle knowledge of the parameter $\theta_0$ that defines the null distribution, and 
	sample the copies $\tilde{X}^{(m)}$ i.i.d.\ from $P_{\theta_0} = \mathcal{N}(Z\theta_0,\I_n)$.
	\item For all methods, the test statistic $T$ is
	given by the absolute value of the  estimate of the coefficient on $X$, 
	when $Y$ is regressed on $X,Z$ with elastic net for penalization on the coefficients on $Z$---specifically, the fitted coefficient $\hat\beta_X$ in the optimization problem
	\[(\hat\beta_X,\hat\beta)=\argmin_{\beta_X, \beta}\left\{\frac{1}{2}\|Y-X\beta_X - Z\beta \|^2+\frac{3}{2}\|\beta\|_{2}^{2} + 7\|\beta\|_{1}\right\}.\]
\end{itemize}

\subsubsection{Simulation: results}

Next, we turn to the results of this simulation.
In Figure~\ref{fig:regression1}, we show the power of the methods for isotonic regression (left) and sparse regression (right).
We see that aCSS (in its original unconstrained form as proposed by \cite{barber2020testing}) quickly loses Type I error control as $\sigma$ increases---this is exactly as expected from the theory,
since the excess Type I error rate is characterized by a term scaling as $\sigma r(\theta_0)$, where $r(\theta_0)$ bounds the estimation error $\|\hat\theta-\theta_0\|$ and therefore
is high in the unconstrained setting. This means that, to maintain (approximate) Type I error control
with  aCSS, we would need to use a small value of $\sigma$, which in turn leads to low power under the alternative. On the other hand,
for our proposed methods---constrained aCSS in the isotonic example,
and $\ell_{1}$-penalized aCSS in the sparse example---we see that approximate Type I error control is well maintained even for larger values of $\sigma$, which allows 
for fairly high power without losing validity. Of course, in each case, the power of the oracle method is higher, as the oracle is given access to the true parameter $\theta_0$ for the null
distribution.     We  also compare our method to alternative approaches that might be considered but are  invalid for testing conditional independence between   $X$ and $Y$ when the  conditional distribution $Y\mid X$ is  unknown. 
	In Example \ref{example:isotonic}, we compare against a  $t$-test that regresses $Y$ on $X$ and tests the significance of the coefficient for $X$. In Example \ref{example:sparse}, we compare against  high-dimensional inference using  the de-sparsified Lasso  \citep{dezeure2015high}, which  tests the significance of  $X$’s coefficient on $Y$. While this approach accounts for high-dimensional settings, it does not accommodate correlation or conditional independence. As shown in Figure~\ref{fig:regression1}, these alternative methods fail to control the Type I error, yielding values significantly larger than the nominal level of 0.05.

\begin{figure} 	
	\centering  
	\includegraphics[width=0.45\linewidth]{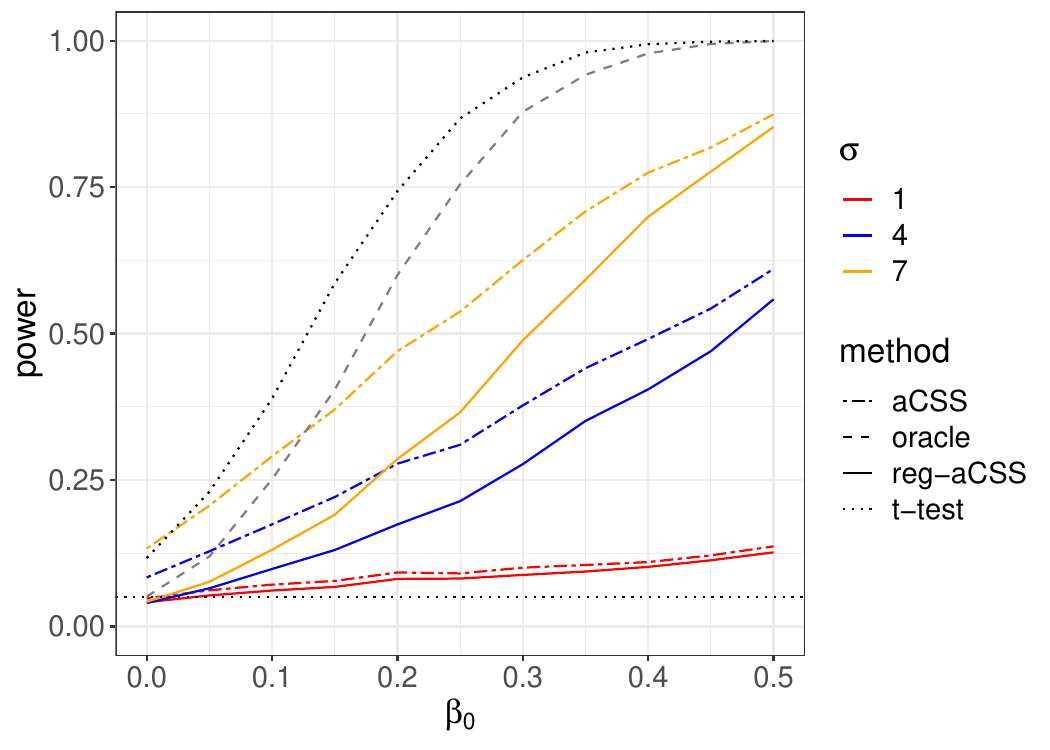}  	
	\includegraphics[width=0.45\linewidth]{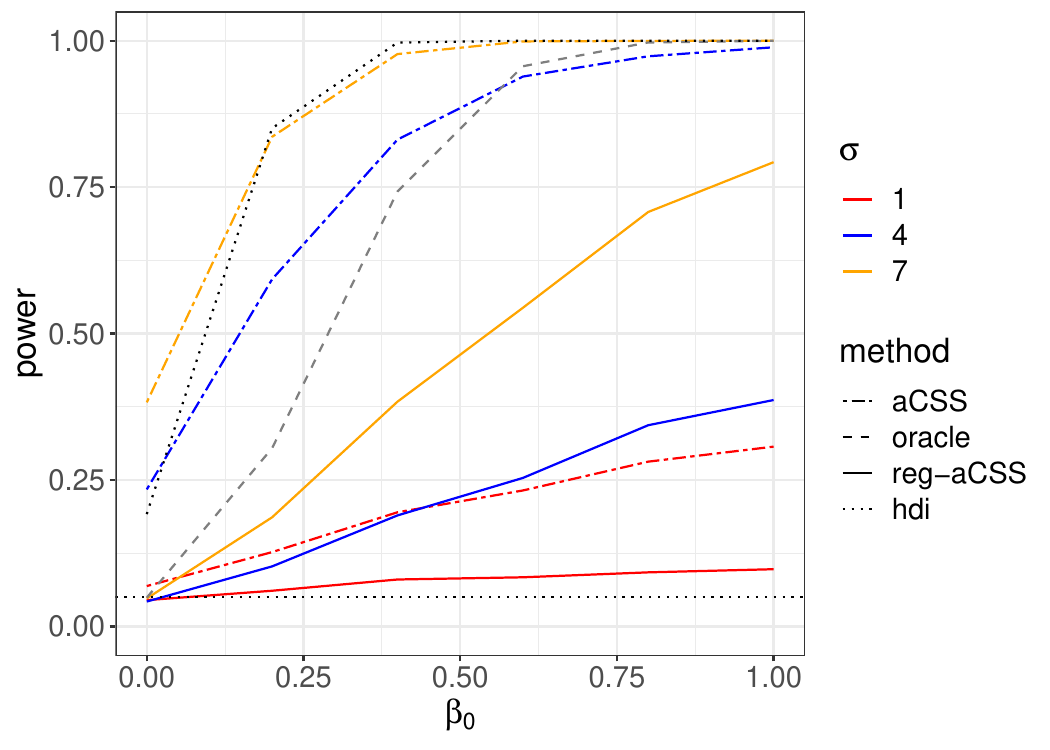}   
 	\caption{Power comparison of regularized aCSS (denoted as reg-aCSS) with aCSS (across different values of $\sigma$), the oracle method, and the t-test (for isotonic regression, left) / de-sparsified Lasso (denoted as hdi, for sparse regression, right), over 5000 independent trials.  The dotted red line denotes the nominal 5\% level. For both settings, $\beta_0=0$ corresponds to the null hypothesis being true.}
	\label{fig:regression1}
\end{figure}

\section{Discussion}\label{sec:discussion}
In this paper, we discuss how to extend the aCSS algorithm to cases where linear constraints, such as an $\ell_1$ constraint or an isotonicity constraint, are applied to enable better
accuracy in the estimator $\hat\theta$. We also extend to the case of an $\ell_{1}$ penalty (e.g., the lasso). This methodology addresses one of the primary open questions proposed in \cite{barber2020testing}, who pose the problem of ``Relaxing regularity conditions and extending to high dimensions". We demonstrate that this extension of the aCSS algorithm can accommodate complex estimators $\hat\theta$, which may be more stable and accurate in high-dimensional settings. Moreover, we show that the regularized aCSS testing has theoretical guarantees for high dimensions when the estimator exhibits a low-dimensional structure. 

A remaining challenge is the problem of efficient sampling for aCSS: as for \cite{barber2020testing}'s earlier work in the unconstrained setting,
aside from special cases such as a Gaussian linear model, overcoming computational challenges for sampling the copies $\tilde{X}^{(m)}$ will greatly increase  the practical utility of this methodology,
and remains an important issue to address in future work.

\section*{Funding}
	W.Z.\ and R.F.B.\ were supported by the Office of Naval Research via grant N00014-20-1-2337.
	R.F.B.\ was additionally supported by the National Science Foundation via grant DMS-2023109.

 
\appendix

	In this supplement,  appendix \ref{sec:proof} presents proofs of the main theoretical results, while Appendix \ref{sec:appdex_otherproof} provides supplementary proofs supporting these results. Appendix \ref{sec:check_asm} contains detailed proofs for the examples discussed in the paper.  We provide details for experiment in Section \ref{sec:exp_detail}.

	\section{Proofs of main results}\label{sec:proof}
	In this section, we provide proofs for our main results: Theorems~\ref{theorem_general}, \ref{theorem_sparse}, \ref{theorem_gaussian}, \ref{theorem_penalized} for establishing Type I error control, and Lemmas~\ref{lemma:conditional_density}, \ref{lemma:conditional_density_l1pen} for computing the conditional density. 
	\subsection{Proof of Theorems~\ref{theorem_general}, \ref{theorem_sparse}: error control for constrained aCSS}\label{sec:proof_sparse}
	\begin{proof}
		As mentioned in Section \ref{sec:theory},  Theorem \ref{theorem_general} is a special case of Theorem \ref{theorem_sparse}, achieved by taking $k(\theta_{0}) = d$
		and taking $v_i = \mathbf{e}_i$ for $i=1,\dots,d$.   Therefore, it is  sufficient to prove Theorem \ref{theorem_sparse}.
		Moreover, it is sufficient to bound the distance to exchangeability, since as argued in \cite{barber2020testing} we have
		\[\P\left(\textnormal{pval}_{T}(X, \tilde{X}^{(1)}, ..., \tilde{X}^{(M)})\le \alpha \right)\le \alpha +  d_{\textnormal{exch}}(X, \tilde{X}^{(1)}, ..., \tilde{X}^{(M)}).\]
		From this point on, then, we only need to establish the bound on $d_{\textnormal{exch}}(X, \tilde{X}^{(1)}, ..., \tilde{X}^{(M)})$.
		
		\subsubsection{Step 1: reduce to total variation distance}\label{sec:step1_proof_mainthm}			We first show that we can obtain the upper bound of the distance to exchangeability through the total variation distance between $P_{\theta_{0}}(\cdot\mid\hat{\theta},\hat{g})$ and its plug-in version. This part of the proof follows the same arguments as the analogous part of the proof of \cite[Theorem 1]{barber2020testing} for unconstrained aCSS. Let 	
		$$ 
		\Omega_{\textnormal{SSOSP}} = \left\{(x, w) \in X\times \R^{d}  : \hat\theta(x, w)  \textnormal{ is a SSOSP of \eqref{eqn:def_thetahat}} \right\},
		$$
		and $P_{\theta_{0}}^{*}$ be the distribution of $(X, W) \sim P_{\theta_{0}}\times \mathcal{N}(0, \frac{1}{d}\I_{d})$ conditional on the event $(X, W)\in \Omega_{\textnormal{SSOSP}}$.
		Consider the joint distribution (a)
		$$
		\textnormal{Distrib.\,(a)}\left\{ \begin{array}{l}
			(X, W) \sim P_{\theta_{0}}^{*},\\ 
			\hat\theta = \hat\theta(X, W), \hat{g} = \nabla\L(\hat\theta; X, W) = \nabla\L(\hat\theta; X) + \sigma W\\
			\tilde{X}^{(1)}, \dots, \tilde{X}^{(M)}\mid X,\hat{g} , \hat{\theta} \sim  \tilde{P}_{M}(\cdot; X, \hat\theta, \hat{g}),
		\end{array}\right.
		$$
		which is equivalent to the aCSS procedure conditional on the event $(X, W)\in \Omega_{\textnormal{SSOSP}}$. On the other hand, if $(X, W)\notin \Omega_{\textnormal{SSOSP}}$, then $\tilde{X}^{(1)}= \dots= \tilde{X}^{(M)}= X$ according to definition and therefore $(X, \tilde{X}^{(1)}, \dots, \tilde{X}^{(M)})$ is exchangeable. Thus, the exchangeability is violated only on the event $(X, W)\in \Omega_{\textnormal{SSOSP}}$. Combined with convex property of distance-to-exchangeability, we have
		$$
		d_{\textnormal{exch}}(X, \tilde{X}^{(1)},\dots, \tilde{X}^{(M)})\le d_{\textnormal{exch}}(\textnormal{Distribution of}\ X, \tilde{X}^{(1)}, \dots, \tilde{X}^{(M)} \textnormal{ under Distrib.\,(a)}),
		$$
		Let $Q_{\theta_{0}}^{*}$ be the joint distribution of $(\hat\theta(X, W), \hat{g}(X, W))$ under $(X, W) \sim P_{\theta_{0}}^{*}$ . Define distribution (b)
		$$
		\textnormal{Distrib.\,(b)}\left\{ \begin{array}{l}
			(\hat\theta,\hat{g})  \sim Q_{\theta_{0}}^{*},\\
			X\mid \hat\theta, \hat{g} \sim p_{\theta_{0}}(\cdot\mid\hat\theta, \hat{g}),\\
			\tilde{X}^{(1)}, \dots, \tilde{X}^{(M)}\mid X, \hat{\theta}, \hat{g} \sim  \tilde{P}_{M}(\cdot; X, \hat\theta, \hat{g}),
		\end{array}\right.
		$$
		where $p_{\theta_{0}}(\cdot\mid \hat\theta, \hat{g})$ is defined in Lemma \ref{lemma:conditional_density}. By definition of $p_{\theta_{0}}(\cdot\mid \hat\theta, \hat{g})$,
		it is clear that Distrib.\,(b) is equivalent to Distrib.\,(a), and then 
		$$
		d_{\textnormal{exch}}(X, \tilde{X}^{(1)},\dots, \tilde{X}^{(M)})\le d_{\textnormal{exch}}(\textnormal{Distribution of}\ X, \tilde{X}^{(1)}, \dots, \tilde{X}^{(M)} \textnormal{ under Distrib.\,(b)}),
		$$
		Further let $p_{\hat\theta}(\cdot\mid \hat\theta, \hat{g})$ be the plug-in version of $p_{\theta_{0}}(\cdot\mid \hat\theta, \hat{g})$ and define 
		$$
		\textnormal{Distrib.\,(c)}\left\{ \begin{array}{l}
			(\hat\theta,\hat{g})  \sim Q_{\theta_{0}}^{*},\\
			X\mid  \hat\theta, \hat{g} \sim p_{\hat\theta}(\cdot| \hat\theta, \hat{g}),\\
			\tilde{X}^{(1)}, \dots, \tilde{X}^{(M)}\mid X, \hat{\theta}, \hat{g} \sim  \tilde{P}_{M}(\cdot; X, \hat\theta, \hat{g}).
		\end{array}\right.
		$$
		From the definition of  $\tilde{P}_{M}$, $(X, \tilde{X}^{(1)}, \dots, \tilde{X}^{(M)})$ is exchangeable under Distrib.\,(c). Then, 
		$$
		d_{\textnormal{exch}}(\textnormal{Distribution of } X, \tilde{X}^{(1)}, \dots, \tilde{X}^{(M)} \textnormal{ under Distrib.\,(b)})\le d_{\textnormal{TV}}(\textnormal{Distrib.\,(b)}, \textnormal{Distrib.\,(c)}).
		$$
		Since the only difference between Distrib.\,(b) and Distrib.\,(c) lies in the conditional distribution $X| \hat\theta, \hat g$,  
		$$
		d_{\textnormal{TV}}(\textnormal{Distrib.\,(b)}, \textnormal{Distrib.\,(c)}) = \E_{Q_{\theta_{0}}^{*}}\left[d_{\textnormal{TV}}(p_{\theta_{0}}(\cdot\mid \hat\theta, \hat{g}), p_{\hat\theta}(\cdot\mid  \hat\theta, \hat{g}))\right].
		$$
		Therefore we can bound the distance to exchangeability  as 
		\begin{equation}\label{eq: exchangebound}
			d_{\textnormal{exch}}(X, \tilde{X}^{(1)}, \dots, \tilde{X}^{(M)})\le \E_{Q_{\theta_{0}}^{*}}\left[d_{\textnormal{TV}}(p_{\theta_{0}}(\cdot\mid  \hat\theta, \hat{g}), p_{\hat\theta}(\cdot\mid \hat\theta, \hat{g}))\right],
		\end{equation}
		i.e., the distance to exchangeability of $X, \tilde{X}^{(1)}, \dots, \tilde{X}^{(M)}$ from the constrained aCSS procedure is bounded by the expected total variation distance between the true conditional distribution and the plug-in conditional distribution. 
		
		\subsubsection{Step 2: bound the total variation distance}	
		Our next step is to bound this expected total variation distance. Here our arguments will need to address a more challenging setting than 
		the corresponding part of the proof of \cite[Theorem 1]{barber2020testing}, as we need to handle constrained rather than unconstrained optimization,
		as well as the issue of the sparse structure reflected by $k(\theta_0)$.
		
		To begin, we  calculate
		\begin{equation}		
			\begin{split}
				d_{\textnormal{TV}}(p_{\theta_{0}}(\cdot\mid \hat\theta, \hat{g}), p_{\hat\theta}(\cdot\mid \hat\theta, \hat{g})) & =   \E_{p_{\theta_{0}}(\cdot\mid  \hat\theta, \hat{g})}\left[\left(1 - \frac{p_{\hat\theta}(X\mid \hat\theta, \hat{g})}{p_{\theta_{0}}(X\mid\hat\theta, \hat{g})}\right)_{+}\right]\\
				& = \E_{p_{\theta_{0}}(\cdot\mid  \hat\theta, \hat{g})}\left[\left(1 -  \frac{\frac{f(X;\hat\theta)}{f(X; \theta_{0})}}{ \E_{p_{\theta_{0}}(\cdot\mid \hat\theta, \hat{g})}\frac{f(X';\hat\theta)}{f(X'; \theta_{0})}}\right)_{+}\right],
			\end{split}
		\end{equation}
		where $(x)_+ = \max\{x,0\}$. Here the first step holds by properties of the total variation distance, while the second step holds 
		by the density calculation in~\eqref{eqn:conditional_density}. 
		To bound this quantity, we first want to show that $\frac{f(X;\hat\theta)}{f(X; \theta_{0})}$ is almost a constant over $p_{\theta_{0}}(\cdot\mid \hat\theta, \hat{g})$. For any $x, \theta$, we take a Taylor series for the function $\theta\rightarrow \log f(x; \theta)$:
		\[
		\log f(x; \theta_{0})- \log f(x; \theta)\\ = (\theta_{0} - \theta)^\top \nabla_{\theta}\log f(x, \theta) + \int_{t=0}^{1}t( \theta - \theta_{0})^\top \nabla_{\theta}^{2}\log f(x; \theta_{t})( \theta - \theta_{0})\;\mathsf{d}t, 
		\] 	where we write $\theta_{t} = (1-t) \theta_{0} + t\theta $. Therefore, we have
		\begin{align*}		
			&\frac{f(x; \theta)}{f(x; \theta_{0})}  = \exp\left\{\log f(x; \theta)- \log f(x; \theta_{0}) \right\}\\
			&= \exp\left\{-( \theta_{0} - \theta)^\top \nabla_{\theta}\log f(x; \theta) - \int_{t=0}^{1}t( \theta - \theta_{0})^\top \nabla_{\theta}^{2}\log f(x; \theta_{t})( \theta - \theta_{0})\;\mathsf{d}t\right\}\\
			&= \exp\bigg\{( \theta_{0} - \theta)^\top (\nabla_{\theta}\L(x; \theta) -g) + \int_{t=0}^{1}t( \theta - \theta_{0})^\top \left( H(\theta_{t}; x) - H(\theta_{t})\right)( \theta - \theta_{0})\;\mathsf{d}t\\
			&\quad\quad\quad  +( \theta_{0} - \theta)^\top (g - \nabla_\theta R(\theta))+ \int_{t=0}^{1}t( \theta - \theta_{0})^\top  H(\theta_{t}) ( \theta - \theta_{0})\;\mathsf{d}t\bigg\},
		\end{align*}	
		where the last step holds for any fixed value $g\in\R^d$ (which will be chosen later), using the fact that $-\nabla_{\theta}\log f(x; \theta)=\nabla_\theta\L(x;\theta) - \nabla_\theta R(\theta)$ by definition of $\L$.

		Next let $\Theta_{0} = \B(\theta_{0}, r(\theta_{0}))\cap \Theta\cap \{\theta:  \|\theta-\theta_{0}\|_{v, 0} \le k(\theta_{0}) \} $.
		If $\theta \in\Theta_0$, then by definition of 		$ \|\theta-\theta_{0}\|_{v, 0}$,  there exists  a subset $S(\theta, \theta_{0})\subseteq[p]$ with $|S(\theta,\theta_0)|\leq k(\theta_0)$, such that $(\theta - \theta_{0})\in \textnormal{span}(\{v_{i}\}_{i\in S(\theta, \theta_{0})})$.	Recall that for any set $S\subseteq[p]$, $\mathcal{P}_{v_{S}}$ denotes the projection to   $\textnormal{span}(\{v_{i}\}_{i\in S})$. Then we have
		\begin{multline*}
			\left|( \theta_{0} - \theta)^\top (\nabla_{\theta}\L(x; \theta) - g) \right|
			= \left|( \theta_{0} - \theta)^\top  \mathcal{P}_{v_{S(\theta, \theta_{0})}}(\nabla_{\theta}\L(x; \theta) - g)\right|\\
			\le \| \theta_{0} - \theta\|\max_{S: |S|\le  k(\theta_0)}\| \mathcal{P}_{v_{S}}(\nabla_{\theta}\L(\theta; x)-g)\|
			\le r(\theta_0)\max_{S: |S|\le  k(\theta_0)}\| \mathcal{P}_{v_{S}}(\nabla_{\theta}\L(\theta; x)-g)\|.
		\end{multline*}
		We also calculate, for $\theta\in\Theta_0$,
		\begin{align*}&\int_{t=0}^{1}t( \theta - \theta_{0})^\top \left( H(\theta_{t}; x) - H(\theta_{t})\right)( \theta - \theta_{0})\;\mathsf{d}t
			\\&\leq \int_{t=0}^{1}t \| \theta - \theta_{0}\|^2 \cdot\lambda_{\max}(H(\theta_{t}; x) - H(\theta_{t}))\;\mathsf{d}t\\
			&\leq  \frac{1}{2}\sup_{\theta'\in\Theta_0}\big(\lambda_{\max}(H(\theta'; x) - H(\theta'))\big)_+ \cdot \|\theta-\theta_0\|^2\\&\leq  \frac{r(\theta_0)^2}{2}\sup_{\theta'\in\Theta_0}\big(\lambda_{\max}(H(\theta'; x) - H(\theta'))\big)_+,\end{align*}
		and similarly,
		\[\int_{t=0}^{1}t( \theta - \theta_{0})^\top \left( H(\theta_{t}; x) - H(\theta_{t})\right)( \theta - \theta_{0})\;\mathsf{d}t \geq  - \frac{r(\theta_0)^2}{2} \sup_{\theta'\in\Theta_0}\big(\lambda_{\max}(H(\theta') - H(\theta';x))\big)_+.\]
		Combining all these calculations, for any $\theta\in\Theta_0$ we have
		\begin{multline*}		
			\frac{f(x; \theta)}{f(x; \theta_{0})}  
			\leq  \exp\bigg\{r(\theta_0)\max_{S: |S|\le  k(\theta_0)}\| \mathcal{P}_{v_{S}}(\nabla_{\theta}\L(\theta; x)-g)\|\\
			+\frac{r(\theta_0)^2}{2} \sup_{\theta'\in\Theta_0}\big(\lambda_{\max}(H(\theta';x) - H(\theta'))\big)_+\\
			+( \theta_{0} - \theta)^\top (g - \nabla_\theta R(\theta))+ \int_{t=0}^{1}t( \theta - \theta_{0})^\top  H(\theta_{t}) ( \theta - \theta_{0})\;\mathsf{d}t\bigg\},
		\end{multline*}	
		and similarly,
		\begin{multline*}		
			\frac{f(x; \theta)}{f(x; \theta_{0})}  
			\geq  \exp\bigg\{ - r(\theta_0)\max_{S: |S|\le  k(\theta_0)}\| \mathcal{P}_{v_{S}}(\nabla_{\theta}\L(\theta; x)-g)\| \\
			- \frac{r(\theta_0)^2}{2} \sup_{\theta'\in\Theta_0}\big(\lambda_{\max}(H(\theta') - H(\theta';x))\big)_+\\
			+( \theta_{0} - \theta)^\top (g - \nabla_\theta R(\theta))+ \int_{t=0}^{1}t( \theta - \theta_{0})^\top  H(\theta_{t}) ( \theta - \theta_{0})\;\mathsf{d}t\bigg\},
		\end{multline*}	
		Now let
		$$
		\Delta_1(\theta,  g; x) = r(\theta_{0})\max_{S: |S|\le k(\theta_{0})}\|\mathcal{P}_{v_{S}}(\nabla_{\theta}\L(\theta; x)-g)\|+ \frac{r(\theta_{0})^{2}}{2}\sup_{\theta'\in \Theta_{0}}\left(\lambda_{\max}\left(H(\theta'; x) - H(\theta')\right)\right)_{+},
		$$
		and                                                       
		$$
		\Delta_1'(\theta,  g; x) = r(\theta_{0})\max_{S: |S|\le k(\theta_{0})}\|\mathcal{P}_{v_{S}}(\nabla_{\theta}\L(\theta; x)-g)\|+ \frac{r(\theta_{0})^{2}}{2}\sup_{\theta'\in \Theta_{0}}\left(\lambda_{\max}\left(H(\theta') - H(\theta'; x) \right)\right)_{+}.
		$$
		Then in our work above, we have shown that
		\[e^{-\Delta_1'(\theta,  g; x)} \leq \frac{f(x; \theta)}{f(x; \theta_{0})}\cdot e^{-( \theta_{0} - \theta)^\top (g - \nabla_\theta R(\theta))- \int_{t=0}^{1}t( \theta - \theta_{0})^\top  H(\theta_{t}) ( \theta - \theta_{0})\;\mathsf{d}t}\leq e^{\Delta_1(\theta,  g; x)}\]
		holds for all $x$, all $g$, and all $\theta\in\Theta_0$. This means that, for all $x,x'\in\X$, all $g$, and all $\theta\in\Theta_0$,
		\[\frac{ \ \frac{f(x'; \theta)}{f(x'; \theta_{0})} \ }{\frac{f(x; \theta)}{f(x; \theta_{0})}} \leq \frac{e^{\Delta_1(\theta,g;x')}}{e^{-\Delta_1'(\theta,  g; x)}}.\]
		In particular, on the event that $\hat\theta\in\Theta_0$,
		plugging in $g = \hat{g}$, we have
		\[\frac{ \ \frac{f(x'; \hat{\theta})}{f(x'; \theta_{0})} \ }{\frac{f(x; \hat{\theta})}{f(x; \theta_{0})}} \leq \frac{e^{\Delta_1(\hat{\theta},\hat{g};x')}}{e^{-\Delta_1'(\hat{\theta},  \hat{g}; x)}},\]
		again for all $x,x'\in\X$. Taking an expected value with respect to $X'\sim p_{\theta_0}(\cdot;\hat{\theta},\hat{g})$, then,
		\begin{multline*}		
			\frac{ \E_{p_{\theta_{0}}(\cdot\mid\hat\theta, \hat{g})}\left[\frac{f(X'; \hat{\theta})}{f(X'; \theta_{0})} \right]}{\frac{f(x; \hat{\theta})}{f(x; \theta_{0})}} 
			=\E_{p_{\theta_{0}}(\cdot\mid\hat\theta, \hat{g})}\left[\frac{ \frac{f(X'; \hat{\theta})}{f(X'; \theta_{0})} }{\frac{f(x; \hat{\theta})}{f(x; \theta_{0})}} \right]\\
			\leq \E_{p_{\theta_{0}}(\cdot\mid\hat\theta, \hat{g})}\left[ \frac{e^{\Delta_1(\hat{\theta},\hat{g};X')}}{e^{-\Delta_1'(\hat{\theta},  \hat{g}; x)}}\right]
			= \frac{ \E_{p_{\theta_{0}}(\cdot\mid\hat\theta, \hat{g})}\left[e^{\Delta_1(\hat{\theta},\hat{g};X')}\right]}{e^{-\Delta_1'(\hat{\theta},  \hat{g}; x)}}.
		\end{multline*}
		Therefore, on the event that $\hat\theta\in\Theta_0$, we have shown that 
		\[\left(1 - \frac{\frac{f(x; \hat{\theta})}{f(x; \theta_{0})}}{ \E_{p_{\theta_{0}}(\cdot\mid\hat\theta, \hat{g})}\left[\frac{f(X'; \hat{\theta})}{f(X'; \theta_{0})} \right]} \right)_+
		\leq 1 -  \frac{e^{-\Delta_1'(\hat{\theta},  \hat{g}; x)}}{ \E_{p_{\theta_{0}}(\cdot\mid\hat\theta, \hat{g})}\left[e^{\Delta_1(\hat{\theta},\hat{g};X')}\right]}.\]
		(Note that the right-hand side is always nonnegative, since the functions $\Delta_1,\Delta_1'$ both return only nonnegative values.)
		In particular, on the event that $\hat\theta\in\Theta_0$, 
		\begin{multline*}d_{\textnormal{TV}}(p_{\theta_{0}}(\cdot\mid \hat\theta, \hat{g}), p_{\hat\theta}(\cdot\mid \hat\theta, \hat{g}))
			= \E_{p_{\theta_{0}}(\cdot\mid \hat\theta, \hat{g})}\left[\left(1  - \frac{\frac{f(x; \hat{\theta})}{f(x; \theta_{0})}}{ \E_{p_{\theta_{0}}(\cdot\mid\hat\theta, \hat{g})}\left[\frac{f(X'; \hat{\theta})}{f(X'; \theta_{0})} \right]} \right)_+\right] \\\leq \E_{p_{\theta_{0}}(\cdot\mid \hat\theta, \hat{g})}\left[ 1 -  \frac{e^{-\Delta_1'(\hat{\theta},  \hat{g}; X)}}{ \E_{p_{\theta_{0}}(\cdot\mid\hat\theta, \hat{g})}\left[e^{\Delta_1(\hat{\theta},\hat{g};X')}\right]}\right]. \end{multline*}
		Combining both cases (i.e., $\hat\theta\in\Theta_0$ and $\hat\theta\not\in\Theta_0$), we see that
		\[d_{\textnormal{TV}}(p_{\theta_{0}}(\cdot\mid \hat\theta, \hat{g}), p_{\hat\theta}(\cdot\mid \hat\theta, \hat{g}))
		\leq \id_{\hat\theta\not\in\Theta_0} + \id_{\hat\theta\in\Theta_0} \E_{p_{\theta_{0}}(\cdot\mid \hat\theta, \hat{g})}\left[ 1 -  \frac{e^{-\Delta_1'(\hat{\theta},  \hat{g}; X)}}{ \E_{p_{\theta_{0}}(\cdot\mid\hat\theta, \hat{g})}\left[e^{\Delta_1(\hat{\theta},\hat{g};X')}\right]}\right].\]
		Therefore,
		\begin{align*}
			&\E_{Q_{\theta_{0}}^{*}}\left[d_{\textnormal{TV}}(p_{\theta_{0}}(\cdot\mid \hat\theta, \hat{g}), p_{\hat\theta}(\cdot\mid \hat\theta, \hat{g})) \right]\\
			&\leq  \P_{Q_{\theta_{0}}^{*}}\{\hat\theta\not\in\Theta_0\} +  \E_{Q_{\theta_{0}}^{*}}\left[\E_{p_{\theta_{0}}(\cdot\mid \hat\theta, \hat{g})}\left[ 1 -  \frac{e^{-\Delta_1'(\hat{\theta},  \hat{g}; x)}}{ \E_{p_{\theta_{0}}(\cdot\mid\hat\theta, \hat{g})}\left[e^{\Delta_1(\hat{\theta},\hat{g};X')}\right]}\right]\right]\\
			& \leq  \P_{Q_{\theta_{0}}^{*}}\{\hat\theta\not\in\Theta_0\} +  \E_{Q_{\theta_{0}}^{*}}\left[  \E_{p_{\theta_{0}}(\cdot\mid \hat\theta, \hat{g})}\left[\Delta_1'(\hat{\theta},  \hat{g}; X)\right]\right] + 1 -  \frac{1}{ \E_{Q_{\theta_{0}}^{*}}\left[\E_{p_{\theta_{0}}(\cdot\mid\hat\theta, \hat{g})}\left[e^{\Delta_1(\hat{\theta},\hat{g};X)}\right]\right]},
		\end{align*}
		where the last step follows the same calculation as in the analogous part of the proof of \cite[Theorem 1]{barber2020testing}.
		Next, by definition,  $(\hat\theta, \hat{g})\sim Q^{*}_{\theta_{0}}$ and $X\mid \hat\theta, \hat{g}\sim p_{\theta_{0}}(\cdot\mid  \hat\theta, \hat{g})$  is equivalent to the joint distribution of $(X,  \hat\theta(X, W), \hat{g}(X, W))$ when $(X, W)\sim P_{\theta_{0}}^{*}$. Therefore
		\begin{multline*}
			\E_{Q_{\theta_{0}}^{*}}\left[d_{\textnormal{TV}}(p_{\theta_{0}}(\cdot\mid\hat\theta, \hat{g}), p_{\hat\theta}(\cdot\mid \hat\theta, \hat{g})) \right]\\
			\le  \P_{P_{\theta_{0}}^{*}}\{\hat\theta\not\in\Theta_0\} + \E_{P_{\theta_{0}}^{*}}\left[ \Delta_1'(\hat\theta(X, W), \hat{g}(X, W); X) \right]  + 1 -\frac{1}{\E_{P_{\theta_{0}}^{*}}\left[ e^{\Delta_1(\hat\theta(X, W), \hat{g}(X, W); X)}\right] }.
		\end{multline*}
		Now define
		$$
		\Delta_{2}( x, w) = r(\theta_{0}) \sigma\max_{S: |S|\le k(\theta_{0})}\|\mathcal{P}_{v_{S}}w\|+ \frac{r(\theta_{0})^{2}}{2}\sup_{\theta'\in \Theta_{0}}\left(\lambda_{\max}\left(H(\theta'; x) - H(\theta')\right)\right)_{+},
		$$
		and
		$$
		\Delta'_{2}(x, w) = r(\theta_{0})\sigma\max_{S: |S|\le k(\theta_{0})}\|\mathcal{P}_{v_{S}}w\|  + \frac{r(\theta_{0})^{2}}{2}\sup_{\theta'\in \Theta_{0}}\left(\lambda_{\max}\left(H(\theta') - H(\theta'; x) \right)\right)_{+}.
		$$		
		Observe that $\hat{g}(X, W) = \nabla\L(\hat\theta; X, W) = \nabla\L(\hat\theta; X) + \sigma W$ by definition, and so we must have
		$$\Delta_1(\hat\theta(X, W), \hat{g}(X, W); X) = \Delta_{2}(X, W), \
		\Delta'_1(\hat\theta(X, W), \hat{g}(X, W); X) = \Delta'_{2}(X, W).$$
		Consequently, 
		\begin{multline*}	
			\E_{Q_{\theta_{0}}^{*}}\left[d_{\textnormal{TV}}(p_{\theta_{0}}(\cdot\mid\hat\theta, \hat{g}), p_{\hat\theta}(\cdot\mid \hat\theta, \hat{g})) \right]\\
			\le  \P_{P_{\theta_{0}}^{*}}\{\hat\theta\not\in\Theta_0\} + \E_{P_{\theta_{0}}^{*}}\left[ \Delta_2'(X,W) \right]  + \left(1 -\frac{1}{\E_{P_{\theta_{0}}^{*}}\left[ e^{\Delta_2(X,W)}\right] }\right).
		\end{multline*}	
		
		Next let $\calE_{\textnormal{SSOSP}}$  be the event that $(X, W)\in\Omega_{\textnormal{SSOSP}}$. Recall that  $P_{\theta_{0}}^{*}$ is the distribution of $(X, W) \sim P_{\theta_{0}}\times \mathcal{N}(0, \frac{1}{d}\I_{d})$ conditional on $\calE_{\textnormal{SSOSP}}$. Then, 
		following the exact same steps as the analogous part of the proof of \cite[Theorem 1]{barber2020testing}, it holds that
		\begin{multline*}
			\E_{Q_{\theta_{0}}^{*}}\left[d_{\textnormal{TV}}(p_{\theta_{0}}(\cdot\mid \hat\theta, \hat{g}), p_{\hat\theta}(\cdot\mid \hat\theta, \hat{g})) \right]\\
			\le  \frac{\P\{\{\hat\theta\not\in\Theta_0\}\cap\calE_{\textnormal{SSOSP}}\} + \E\left[  \Delta'_{2}(X, W) \right] + \log\E\left[ e^{ \Delta_{2}(X, W)}\right]}{1-\P\{\calE_{\textnormal{SSOSP}}^{\complement}\}} \\
			\le  \frac{\delta(\theta_{0})+\tilde\delta(\theta_{0}) - \P(\calE_{\textnormal{SSOSP}}^{\complement})+ \E\left[  \Delta'_{2}(X, W) \right] + \log\E\left[ e^{ \Delta_{2}(X, W)}\right]}{1-\P\{\calE_{\textnormal{SSOSP}}^{\complement}\}} ,
		\end{multline*}
		where now probability and expectation are taken with respect to $(X, W) \sim P_{\theta_{0}}\times \mathcal{N}(0, \frac{1}{d}\I_{d})$, and where the last step holds by Assumption \ref{asm:SSOSP_general}, together with the assumption in the theorem.
		
		Next, for a standard normal vector $Z\sim \mathcal{N}(0, \I_{d})$ and 1-Lipschitz function $f$, we have $\log \E e^{\lambda f(Z)}\le \frac{\lambda^2}{2} + \lambda \E[f(Z)]$ for all $\lambda$ \citep{boucheron2013concentration}. 
		We can verify that $f(z) = \underset{S \subseteq[p]: |S|\le k(\theta_{0})}{\max}\|\mathcal{P}_{v_{S}}z\|$ is a 1-Lipschitz function, and  by definition of $h_v$, we have  $\E[f(Z)^2] = h_{v}(k(\theta_{0}))$. Then, since $\sqrt{d} W$ is a standard normal random vector, we have
		\begin{multline*}
			\log\E\left[ e^{2r(\theta_{0})\sigma\max_{S: |S|\le k(\theta_{0})}\|\mathcal{P}_{v_{S}}W\| }\right] 
			= 	\log	\E\left[ e^{\frac{2r(\theta_{0})\sigma}{\sqrt{d}}f(\sqrt{d} W)}\right] \\
			\le  \frac{2r(\theta_{0})^2\sigma^2}{d}  + 2r(\theta_{0})\sigma \sqrt{\frac{h_{v}(k(\theta_{0}))}{d}}. 
		\end{multline*}
		Next, we can assume that $2\sigma r(\theta_{0})\le d\sqrt{\frac{h_{v}(k(\theta_{0}))}{d}}$.
		(To see why, observe that $h_v(k(\theta_0))\geq h_v(1) \geq 1$. If this inequality fails, then $3\sigma r(\theta_0)\sqrt{\frac{h_{v}(k(\theta_{0})}{d}}
		\geq \frac{3\sigma r(\theta_0)}{\sqrt{d}} \geq 1$, and so the bound in the theorem holds trivially since total variation distance can never exceed 1.) Then we have
		\[	\log\E\left[ e^{2r(\theta_{0})\sigma\max_{S: |S|\le k(\theta_{0})}\|\mathcal{P}_{v_{S}}W\| }\right]  \leq 3r(\theta_{0})\sigma \sqrt{\frac{h_{v}(k(\theta_{0}))}{d}}.\]
		Next, combining Cauchy--Schwarz and Assumption \ref{asm:L_H} we have 
		\begin{equation*}
			\begin{split}
				& \log\E\left[ e^{ \Delta_{2}(X, W)}\right] \\
				&\le  \frac{1}{2}\log\E\left[ e^{2r(\theta_{0})\sigma  \max_{S: |S|\le k(\theta_{0})}\|\mathcal{P}_{v_{S}}W\| }\right] + \frac{1}{2}\log\E\left[ e^{r(\theta_{0})^{2} \sup_{\theta'\in \Theta_0}\left(\lambda_{\max}\left(H(\theta'; x) - H(\theta')\right)\right)_{+}}\right]\\
				&\le 1.5 r(\theta_{0})\sigma \sqrt{\frac{h_{v}(k(\theta_{0}))}{d}} + \frac{\epsilon(\theta_{0})}{2}.
			\end{split}
		\end{equation*}
		Similarly, by Jensen's inequality, we have
		\begin{equation*}
			\begin{split}
				& \E\left[ \Delta'_{2}(X, W) \right] \\
				&=\E\left[ r(\theta_{0})\sigma  \max_{S: |S|\le k(\theta_{0})}\|\mathcal{P}_{v_{S}}W\|  \right] +  \frac{1}{2}\E \left[ r(\theta_{0})^{2}\sup_{\theta'\in \Theta_0}\left(\lambda_{\max}\left(H(\theta') - H(\theta'; x) \right)\right)_{+} \right]\\
				&\le  \frac{1}{2}\log\E\left[ e^{2r(\theta_{0})\sigma  \max_{S: |S|\le k(\theta_{0})}\|\mathcal{P}_{v_{S}}W\|  }\right] +  \frac{1}{2}\E \left[ r(\theta_{0})^{2}\sup_{\theta'\in \Theta_0}\left(\lambda_{\max}\left(H(\theta') - H(\theta'; x) \right)\right)_{+} \right]\\
				&\le 1.5 r(\theta_{0})\sigma \sqrt{\frac{h_{v}(k(\theta_{0}))}{d}}+ \frac{\epsilon(\theta_{0})}{2}.
			\end{split}
		\end{equation*}
		Therefore,			\begin{multline*}  
			\E_{Q_{\theta_{0}}^{*}}\left[d_{\textnormal{TV}}(p_{\theta_{0}}(\cdot\mid \hat\theta, \hat{g}), p_{\hat\theta}(\cdot\mid \hat\theta, \hat{g})) \right]\\
			\le \frac{\delta(\theta_{0})+\tilde\delta(\theta_{0}) - \P(\calE_{\textnormal{SSOSP}}^{\complement})+3\sigma r(\theta_{0})\sqrt{\frac{h_{v}(k(\theta_{0}))}{d}} +\epsilon(\theta_{0})}{1-\P\{\calE_{\textnormal{SSOSP}}^{\complement}\}} \\
			\le 3\sigma r(\theta_{0})\sqrt{\frac{h_{v}(k(\theta_{0}))}{d}} + \epsilon(\theta_{0}) + \delta(\theta_{0})+\tilde\delta(\theta_{0}),
		\end{multline*}
		where to verify the last step, we can apply the fact that $\frac{a-b}{1-b}\leq a$ for any $a\in[0,1]$ and $b\in[0,1)$ (note that we can assume that $3\sigma r(\theta_{0})\sqrt{\frac{h_{v}(k(\theta_{0}))}{d}} + \epsilon(\theta_{0}) + \delta(\theta_{0})+\tilde\delta(\theta_{0})\leq 1$, as otherwise the bound holds trivially since total variation distance can never exceed 1).
		This completes the proof.		\end{proof}

	\subsection{Proof of Theorem \ref{theorem_gaussian}: constrained aCSS for the Gaussian linear model}  \label{sec:proof_gaussian}
	Following the same reasoning as in the proof of Theorem~\ref{theorem_sparse}, we only need to bound 		
	\[ \E_{Q_{\theta_{0}}^{*}}\left[d_{\textnormal{TV}}(p_{\theta_{0}}(\cdot\mid  \hat\theta, \hat{g}), p_{\hat\theta}(\cdot\mid \hat\theta, \hat{g}))\right],
	\]
	where, as in that proof, $Q_{\theta_{0}}^{*}$ is the joint distribution of $(\hat\theta(X, W), \hat{g}(X, W))$ under $(X, W) \sim P_{\theta_{0}}^{*}$,
	where $P_{\theta_{0}}^{*}$ is the distribution of $(X, W) \sim P_{\theta_{0}}\times \mathcal{N}(0, \frac{1}{d}\I_{d})$ conditional on the event $(X, W)\in \Omega_{\textnormal{SSOSP}}$.
	For the Gaussian case, by our assumption~\eqref{eqn:assume_R} on $R(\theta)$,  the event $(X, W)\in \Omega_{\textnormal{SSOSP}}$ holds almost surely,
	and so $Q_{\theta_0}^*$ is in fact the joint distribution of $(\hat\theta(X, W), \hat{g}(X, W))$ under $(X, W) \sim \mathcal{N}(Z\theta_0,\nu^2\I_n)\times\mathcal{N}(0,\frac{1}{d}\I_{d})$.
	
	Next, applying Lemma~\ref{lemma:conditional_density}, we calculate
	\[p_{\theta_{0}}(x\mid  \hat\theta, \hat{g}) \propto \exp\left\{-\frac{1}{2\nu^2}\|x - Z\theta_0\|^2 - \frac{1}{2\sigma^2/d}\left\|\hat{g} - \left(\frac{1}{\nu^2}Z^\top(Z\hat\theta - x) + \nabla_\theta R(\hat\theta)\right)\right\|^2\right\}\]
	and
	\[p_{\hat\theta}(x\mid  \hat\theta, \hat{g}) \propto \exp\left\{-\frac{1}{2\nu^2}\|x - Z\hat\theta\|^2 - \frac{1}{2\sigma^2/d}\left\|\hat{g} - \left(\frac{1}{\nu^2}Z^\top(Z\hat\theta - x) + \nabla_\theta R(\hat\theta)\right)\right\|^2\right\},\]
	which simplifies to the normal distributions
	\[
	\mathcal{N}\left(Z\hat\theta +\left(\I_{n}+\frac{d}{\sigma^2\nu^2}ZZ^\top \right)^{-1}\left[ \frac{d}{\sigma^2} Z( \nabla_\theta R(\hat\theta) - \hat{g}) + Z(\theta_0 - \hat\theta)\right], \nu^2\left(\I_{n} + \frac{d}{\sigma^2\nu^2}ZZ^\top \right)^{-1}\right)\]
	and 
	\[		\mathcal{N}\left(Z\hat\theta +\left(\I_{n}+\frac{d}{\sigma^2\nu^2}ZZ^\top \right)^{-1}\left[ \frac{d}{\sigma^2} Z( \nabla_\theta R(\hat\theta) - \hat{g})\right], \nu^2\left(\I_{n} + \frac{d}{\sigma^2\nu^2}ZZ^\top \right)^{-1}\right),\]
	respectively.
	For any $\mu,\mu'\in\R^n$ and any positive definite $\Sigma\in\R^{n\times n}$,
	\begin{multline*}
		d_{\textnormal{TV}}\big(\mathcal{N}(\mu,\Sigma),\mathcal{N}(\mu',\Sigma)\big) \leq \sqrt{\frac{1}{2}d_{\textnormal{KL}}\big(\mathcal{N}(\mu,\Sigma)\big\|\mathcal{N}(\mu',\Sigma)\big)}\\
		= \sqrt{\frac{1}{2}\cdot \frac{1}{2}(\mu-\mu')^\top\Sigma^{-1}(\mu-\mu')} = \frac{1}{2}\|\Sigma^{-1/2}(\mu-\mu')\|,
	\end{multline*}
	where $d_{\textnormal{KL}}$ is the Kullback--Leibler divergence, and the first step holds by Pinsker's inequality.
	Applying this calculation to the distributions $p_{\theta_{0}}(\cdot\mid  \hat\theta, \hat{g})$ and $p_{\hat\theta}(\cdot\mid  \hat\theta, \hat{g})$ computed above, we have
	\begin{align*}
		d_{\textnormal{TV}}(p_{\theta_{0}}(\cdot\mid  \hat\theta, \hat{g}), p_{\hat\theta}(\cdot\mid \hat\theta, \hat{g}))
		& \leq \frac{1}{2\nu}\left\|\left(\I_{n} + \frac{d}{\sigma^2\nu^2}ZZ^\top \right)^{1/2} \cdot \left(\I_{n}+\frac{d}{\sigma^2\nu^2}ZZ^\top \right)^{-1}Z(\hat\theta - \theta_0)\right\|\\
		& \leq \frac{1}{2\nu}\left\|\left(\I_{n} + \frac{d}{\sigma^2\nu^2}ZZ^\top \right)^{-1/2} Z\right\|\cdot\|\hat\theta - \theta_0\|\\
		& = \frac{\sigma}{2\sqrt{d}}\left\|\left(\frac{\sigma^2\nu^2}{d}\I_{n} + ZZ^\top \right)^{-1/2} Z\right\|\cdot\|\hat\theta - \theta_0\|
		\leq \frac{\sigma}{2\sqrt{d}}\cdot\|\hat\theta - \theta_0\|.
	\end{align*}
	
	On the event that $\|\hat\theta - \theta_0\|\leq r(\theta_0)$ we therefore have $d_{\textnormal{TV}}(p_{\theta_{0}}(\cdot\mid  \hat\theta, \hat{g}), p_{\hat\theta}(\cdot\mid \hat\theta, \hat{g}))\leq \frac{\sigma}{2\sqrt{d}}r(\theta_0)$. Since total variation distance is always bounded by 1, and we therefore have
	\begin{multline*}\E_{Q_{\theta_{0}}^{*}}\left[d_{\textnormal{TV}}(p_{\theta_{0}}(\cdot\mid  \hat\theta, \hat{g}), p_{\hat\theta}(\cdot\mid \hat\theta, \hat{g}))\right]\\
		\leq \frac{\sigma}{2\sqrt{d}}r(\theta_0) \cdot \P_{Q_{\theta_{0}}^{*}}\{\|\hat\theta - \theta_0\|\leq r(\theta_0)\}+\P_{Q_{\theta_{0}}^{*}}\{\|\hat\theta - \theta_0\|> r(\theta_0)\}  \\\leq \frac{\sigma}{2\sqrt{d}}r(\theta_0) + \delta(\theta_0),\end{multline*}
	since $\|\hat\theta - \theta_0\|\leq r(\theta_0)$ holds with probability at least $1-\delta(\theta_0)$ by assumption.

	\subsection{Proof of Lemma  \ref{lemma:conditional_density}: conditional density}\label{sec:proof_conditional_density}
	We begin by introducing some  notation for remaining proofs. 
	For $A\in\R^{r\times d}, b\in\R^{r}$, define a subset of $\Theta$ with active set $\A\subseteq[r]$  as follows:
	$$\Theta_{A, b, \A} = \{\theta\in\Theta: A_i^\top\theta = b_{i}, \forall i\in\A;  A_i^\top\theta< b_{i}, \forall i\in[r]\backslash \A\},$$
	where $A_i$ is the $i$th row of $A$.
	We will write $\Theta_{\A}  = \Theta_{A, b, \A}$ when $A, b$ are fixed.
	As before, we define $\A(\theta) = \{i\in[r]:A_i^\top\theta = b_i\}$, the active set for a given $\theta\in\Theta$, so that
	we have $\theta\in\Theta_{A,b,\A(\theta)}$ by definition.
	
	Before proving Lemma~\ref{lemma:conditional_density}, we need a preliminary result, which we will prove below.
	
	\begin{lemma}\label{lemma:bij}
		For  index set $\A\in[r]$, define
		\[
		\Omega_{\textnormal{SSOSP},  \A} = \left\{(x, w) \in \X\times\R^{d}  : \hat\theta(x, w )  \textnormal{ is a SSOSP of~\eqref{eqn:def_thetahat}, and $\A(\hat\theta(x,w))=\A$}\right\},\]
		and
		\begin{multline*}\Psi_{\textnormal{SSOSP},   \A} =  \bigg\{
			(x, \theta, g)\in\X \times\Theta_{\A}\times\R^d : \exists w\in\R^d\textnormal{ such that }\\ \theta = \hat\theta(x,w)\textnormal{ is a SSOSP of~\eqref{eqn:def_thetahat}, and $g=\hat g(x,w)$} \bigg\}.\end{multline*}
		Define a map $\psi_{\A}$ from $\Omega_{\textnormal{SSOSP},  \A} $  as
		\begin{equation*}
			\psi_{\A}: (x, w )\rightarrow  \left(x,  \hat\theta(x, w),  \hat{g}(x, w) \right).
		\end{equation*}
		Then $\psi_{\A}$ is a bijection between $\Omega_{\textnormal{SSOSP}, \A}$ and $\Psi_{\textnormal{SSOSP}, \A}$ with inverse 
		\begin{equation*}
			\psi_{\A}^{-1}:  (x,  \theta, g)\rightarrow \left(x,  \frac{g - \nabla_{\theta}\L(\theta; x)}{\sigma} \right).
		\end{equation*} 
	\end{lemma}
	To give intuition for this result, 
	the bijection between $\Omega_{\textnormal{SSOSP}, \A}$ and $\Psi_{\textnormal{SSOSP}, \A}$ helps us see why we need to condition on both $\hat\theta$ and $\hat g$, rather than on $\hat\theta$ alone as for the (unconditional) aCSS of \cite{barber2020testing}.  Intuitively, the estimator $\hat\theta$ itself cannot reflect enough information for data $(x,w)$ when constraints appear in the optimization step, because $\hat\theta$ may have lower effective dimension (e.g., if one constraint is active, then the value of $\hat\theta$ has $d-1$ degrees of freedom; this means that $(x,\hat\theta)$ cannot contain 
	sufficient information to recover $(x,w)$, since $w$ is $d$-dimensional). 
	In the unconstrained case, $\hat{g} \equiv 0$ due to the first-order optimality conditions, so conditioning on $(\hat\theta,\hat{g})$ is equivalent to simply conditioning on $\hat\theta$, in that case.
	
	With this result in place, we are now ready to prove Lemma \ref{lemma:conditional_density}, which calculates the conditional density.
	
	\begin{proof}[Proof of Lemma \ref{lemma:conditional_density}]
		Consider the joint distribution $(X, W)\sim P_{\theta_{0}}\times \mathcal{N}(0, \frac{1}{d}\I_{d})$. By assumption in the lemma, the event $(X, W )\in\Omega_{\textnormal{SSOSP}, \A}$ has positive probability. Then the joint density of $(X, W )$, conditioning on the event that $\hat\theta(X, W)$ is a SSOSP of \eqref{eqn:def_thetahat} with active set $\A$, i.e., $(X, W )\in\Omega_{\textnormal{SSOSP}, \A}$, is proportional to the function
		$$
		h_{\theta_0}(x, w) = f(x;\theta_{0})\exp\left\{-\frac{d}{2}\|w \|^{2} \right\}\id_{(x, w )\in \Omega_{\textnormal{SSOSP}, \A}}.
		$$		 
		By Lemma \ref{lemma:bij}, $\psi_{\A}$ is a bijection between $\Omega_{\textnormal{SSOSP}, \A}$ and $\Psi_{\textnormal{SSOSP}, \A}$. For any measurable set $I_{\A}\subseteq \Psi_{\textnormal{SSOSP}, \A}$, define  $$\psi_{\A}^{-1}(I_{\A}) = \{(x, w )\in \Omega_{\textnormal{SSOSP}, \A}: \psi_{\A}(x, w )\in I_{\A})\}.$$
		Then,  we calculate
		\begin{align*}
			&\P\left\{(X, \hat\theta(X, W), \hat{g}(X, W))\in I_{\A}\mid (X,\hat\theta(X,W),\hat{g}(X,W))\in\Psi_{\textnormal{SSOSP},\A} \right\} \\
			&= \P\left\{(X, W )\in \psi_{\A}^{-1}(I_{\A})\mid (X,W)\in\Omega_{\textnormal{SSOSP},\A} \right\}\textnormal{ by Lemma~\ref{lemma:bij}}\\
			& = \frac{\int_{\psi_{\A}^{-1}(I_{\A})}h_{\theta_0}(x, w)\;\mathsf{d}\nu_{\X}(x)\;\mathsf{d}w }{\int_{\X\times\R^{d}}h_{\theta_0}(x', w')\;\mathsf{d}\nu_{\X}(x')\;\mathsf{d}w' }\\
			& = \frac{\int_{\psi_{\A}^{-1}(I_{\A})} f(x;\theta_{0})\exp\left\{-\frac{d}{2}\|w \|^{2} \right\}\id_{(x, w )\in \Omega_{\textnormal{SSOSP}, \A}}\;\mathsf{d}\nu_{\X}(x)\;\mathsf{d}w }{\int_{\X\times\R^{d}}h_{\theta_0}(x', w')\;\mathsf{d}\nu_{\X}(x')\;\mathsf{d}w' }\textnormal{ by definition of $h_{\theta_0}(x,w)$}\\
			& = \frac{\int_{\psi_{\A}^{-1}(I_{\A})} f(x;\theta_{0})e^{-\frac{d}{2\sigma^2}\|\hat{g}(x,w)-\nabla_\theta\L(\hat\theta(x,w);x) \|^{2} }\id_{(x, w )\in \Omega_{\textnormal{SSOSP}, \A}}\;\mathsf{d}\nu_{\X}(x)\;\mathsf{d}w }{\int_{\X\times\R^{d}}h_{\theta_0}(x', w')\;\mathsf{d}\nu_{\X}(x')\;\mathsf{d}w' }\\
			& = \frac{\int_{\X} f(x;\theta_{0})\int_{\R^d}e^{-\frac{d}{2\sigma^2}\|\hat{g}(x,w)-\nabla_\theta\L(\hat\theta(x,w);x) \|^{2} }\id_{(x, w )\in  \psi_{\A}^{-1}(I_{\A})}\;\mathsf{d}w \;\mathsf{d}\nu_{\X}(x)}{\int_{\X\times\R^{d}}h_{\theta_0}(x', w')\;\mathsf{d}\nu_{\X}(x')\;\mathsf{d}w' } ,
		\end{align*}
		where the last step holds  since $ \psi_{\A}^{-1}(I_{\A})\subseteq \Omega_{\textnormal{SSOSP},\A}$.
		
		Next, we need to reparameterize $\theta$ and $g$, since given the active set $\A$, these variables must lie in lower-dimensional subspaces of $\Theta$ and of $\R^d$, respectively.
		Let $k=\textnormal{rank}(\textnormal{span}(A_\A)^\perp)$, let $U_\A\in\R^{d\times k}$ be an orthonormal basis for $\textnormal{span}(A_\A)^\perp$ as before,
		and let $V_\A\in\R^{d\times (d-k)}$ be an orthonormal basis for $\textnormal{span}(A_\A)$, so that $(U_\A \ V_\A)\in\R^{d\times d}$ is an orthogonal matrix. 
		Define $\Theta' = \{U_\A^\top \theta : \theta\in\Theta_\A\}\subseteq\R^{k}$. Then $\theta' = U_\A^\top \theta$ and $g' = V_\A^\top g$ are a reparametrization of $(\theta,g)$,
		which now take values in $\Theta'$ and $\R^{d-k}$, respectively. To see why, let $\theta_*\in\R^{d-k}$ be the unique value such that $  \theta_* = V_\A^\top \theta$ for all $\theta\in\Theta_\A$, i.e.,
		$\theta_*$ is determined by the active constraints (specifically, if $A_\A = MDV_\A^\top$ is a singular value decomposition, then $\theta_* = D^{-1}M^\top b_\A$). Then $\theta = U_\A\theta' + V_\A\theta_*$, and $g = V_\A g'$, whenever $(\theta,g)$ corresponds to a SSOSP with active set $\A$
		(i.e., for any $\theta\in\Theta_\A$ and  $g\in\textnormal{span}(A_\A)$).
		
		Next, for $\theta\in\Theta_\A$ and $g\in\textnormal{span}(A_\A)$, if $(x,\theta,g)\in \Psi_{\textnormal{SSOSP},\A}$ then by the SSOSP conditions we must have some $w$ such that
		$ \theta = \hat\theta(x,w)$  is a SSOSP of~\eqref{eqn:def_thetahat}, and $g=\hat g(x,w) =  \nabla_\theta\L(\theta;x,w) =  \nabla_\theta\L(\theta;x) + \sigma w$. 
		Combining with the work above, we can write
		\[w = \phi_x(\theta',g') \textnormal{ where }\phi_x(\theta',g') = \frac{V_\A g' -   \nabla_\theta\L(U_\A\theta'+V_\A\theta_*;x)}{\sigma},\]
		and so 
		\[\theta = \hat\theta(x,w) = \hat\theta\left(x, \phi_x(\theta',g')\right), \ g = \hat{g}(x,w) = \hat{g}\left(x,  \phi_x(\theta',g')\right).\]
		Therefore,
		\[\theta' = U_\A^\top \hat\theta\left(x, \phi_x(\theta',g')\right), \ g' = V_\A^\top  \hat{g}\left(x,  \phi_x(\theta',g')\right).\]
		We can also calculate
		\[\nabla_{\theta'} \phi_x(\theta',g') = -\sigma^{-1}U_\A^\top  \nabla^2_\theta \L(U_\A\theta'+V_\A\theta_*;x)
		\]
		and\[\nabla_{g'} \phi_x(\theta',g') = \sigma^{-1} V_\A^\top .\]
		Therefore, 
		\begin{align*}
			\det\!\left(\nabla\phi_x(\theta',g')\right)
			&= \det\!\left(
			\begin{array}{c}
				\nabla_{\theta'}\phi_x(\theta',g')\\
				\nabla_{g'}\phi_x(\theta',g')
			\end{array}\right)\\
			&= \det\!\left(
			\left(\begin{array}{c}
				\nabla_{\theta'}\phi_x(\theta',g')\\
				\nabla_{g'}\phi_x(\theta',g')
			\end{array}\right)(U_\A\ V_\A)\right)\\
			&= \det\!\left(
			\begin{array}{c@{\ }c}
				\nabla_{\theta'}\phi_x(\theta',g')U_\A &
				\nabla_{\theta'}\phi_x(\theta',g')V_\A \\
				\nabla_{g'}\phi_x(\theta',g')U_\A &
				\nabla_{g'}\phi_x(\theta',g')V_\A
			\end{array}\right)\\
			&= \det\!\left(
			\begin{array}{c@{\ }c}
				-\sigma^{-1}U_\A^\top \, \nabla^2_\theta \L(\,\cdot\,;x)\, U_\A &
				-\sigma^{-1}U_\A^\top \, \nabla^2_\theta \L(\,\cdot\,;x)\, V_\A \\
				\sigma^{-1}V_\A^\top U_\A &
				\sigma^{-1}V_\A^\top V_\A
			\end{array}\right)\\
			&\quad(\text{where }\;\nabla^2_\theta \L(\,\cdot\,;x)
			= \nabla^2_\theta \L(U_\A\theta' + V_\A\theta_*;x))\\
			&= \det\!\left(
			\begin{array}{c@{\ }c}
				-\sigma^{-1}U_\A^\top \, \nabla^2_\theta \L(\,\cdot\,;x)\, U_\A &
				-\sigma^{-1}U_\A^\top \, \nabla^2_\theta \L(\,\cdot\,;x)\, V_\A \\
				0 & \sigma^{-1}\I_{d-k}
			\end{array}\right)\\
			&= (-1)^k \sigma^{-d}\,
			\det\!\left(U_\A^\top \nabla^2_\theta \L(U_\A\theta'+V_\A\theta_*;x)U_\A\right).
		\end{align*}
		
		From this point on, following similar arguments as \cite[Section B.4]{barber2020testing} to verify the validity of applying the change-of-variables formula for integration, we calculate
		\begin{multline*}\int_{\R^d}e^{-\frac{d}{2\sigma^2}\|\hat{g}(x,w)-\nabla_\theta\L(\hat\theta(x,w);x) \|^{2} }\id_{(x, w )\in  \psi_{\A}^{-1}(I_{\A})}\;\mathsf{d}w\\
			= \sigma^{-d}\int_{\Theta'\times \R^{d-k}} e^{-\frac{d}{2\sigma^2}\|V_\A g'-\nabla_\theta\L(U_\A\theta'+V_\A\theta_*;x) \|^{2} }\cdot  \textnormal{det}_{\A,\theta',x}\cdot \id_{(x, \phi_x(\theta',g') )\in  \psi_{\A}^{-1}(I_{\A})}\;\mathsf{d}g'\;\mathsf{d}\theta',\end{multline*}
		where we write $\textnormal{det}_{\A,\theta',x} = \textnormal{det}\left(U_\A^\top  \nabla^2_\theta \L(U_\A\theta'+V_\A\theta_*;x) U_\A\right)$
		(note that this determinant must be positive, by the SSOSP conditions). We can also verify from our definitions that $ \id_{(x, \phi_x(\theta',g') )\in  \psi_{\A}^{-1}(I_{\A})} = \id_{(x,U_\A\theta'+V_\A\theta_*,V_\A g')\in I_\A}$. 
		With this calculation in place we then have
		\begin{multline*}
			\P\left\{(X, \hat\theta(X, W), \hat{g}(X, W))\in I_{\A}\mid (X,\hat\theta(X,W),\hat{g}(X,W))\in\Psi_{\textnormal{SSOSP},\A} \right\} \\
			= \sigma^{-d}\int_{\X} f(x;\theta_{0})\int_{\Theta'\times \R^{d-k}} \frac{\exp\{-\frac{d}{2\sigma^2}\|V_\A g'-\nabla_\theta\L(U_\A\theta'+V_\A\theta_*;x) \|^{2} \}}{\int_{\X\times\R^{d}}h_{\theta_0}(x', w')\;\mathsf{d}\nu_{\X}(x')\;\mathsf{d}w' }\\
			\cdot \textnormal{det}_{\A,\theta',x} \cdot \id_{(x,U_\A\theta'+V_\A\theta_*,V_\A g')\in I_\A}\;\mathsf{d}g'\;\mathsf{d}\theta' \;\mathsf{d}\nu_{\X}(x),\end{multline*}
		In particular, this verifies that 			
		\[\frac{\sigma^{-d} f(x;\theta_{0}) \cdot e^{ - \frac{d}{2\sigma^2}\|V_\A g'-\nabla_\theta\L(U_\A\theta'+V_\A\theta_*;x) \|^2}\cdot\textnormal{det}_{\A,\theta',x} \cdot\id_{(x,U_\A\theta'+V_\A\theta_*,V_\A g')\in \Psi_{\textnormal{SSOSP},\A}}}{\int_{\X\times\R^{d}}h_{\theta_0}(x', w')\;\mathsf{d}\nu_{\X}(x')\;\mathsf{d}w' }\]
		is the joint density of $(X,U_\A^\top \hat{\theta},V_\A^\top \hat{g}) = (X, U_\A^\top \hat\theta(X, W),V_\A^\top  \hat{g}(X, W))$, conditional on the event $(X, \hat\theta(X, W), \hat{g}(X, W))\in\Psi_{\textnormal{SSOSP},\A}$. 
		Therefore, the conditional density of $X\mid (U_\A^\top \hat\theta,V_\A^\top \hat{g})$ (again conditioning on this same event) can be written as
		\[\propto f(x;\theta_{0}) \cdot e^{ - \frac{d}{2\sigma^2}\|V_\A g'-\nabla_\theta\L(U_\A\theta'+V_\A\theta_*;x) \|^2}\cdot\textnormal{det}_{\A,\theta',x} \cdot\id_{(x,U_\A\theta'+V_\A\theta_*,V_\A g')\in \Psi_{\textnormal{SSOSP},\A}}.\]
		Moreover, $U_\A^\top \hat\theta$ and $V_\A^\top\hat g$ uniquely determine $\hat\theta$ and $\hat g$ on the event that $\A$ is the active set, as described earlier,
		so we can equivalently condition on $(\hat\theta,\hat g)$ and  can rewrite this density as
		\begin{equation} 
			p_{\theta_{0}}(\cdot\mid \hat\theta, \hat{g}) \propto f(x;\theta_{0}) \cdot e^{ - \frac{d}{2\sigma^2}\|\hat{g}- \nabla_\theta \L(\hat\theta;x)\|^2}\cdot\textnormal{det}\left(U_\A^\top \nabla^2_\theta\L(\hat\theta;x) U_\A\right)\cdot\id_{(x,\hat\theta,\hat{g})\in \Psi_{\textnormal{SSOSP},\A}}.
		\end{equation}
		Finally, by definition, $(x,\hat\theta,\hat{g})\in \Psi_{\textnormal{SSOSP},\A}$ if and only if $\hat\theta\in\Theta_\A$ and $x\in\X_{\hat\theta,\hat{g}}$, so $\id_{(x,\hat\theta,\hat{g})\in \Psi_{\textnormal{SSOSP},\A}}$ $= \id_{x\in\X_{\hat\theta,\hat{g}}}$ for $\hat\theta\in\Theta_\A$. 
	\end{proof}

	\subsection{Proof of Theorem \ref{theorem_penalized}: error control for aCSS with an $\ell_{1}$ penalty} \label{sec:proof_l1pen}
	At a high level, the strategies underlying the proofs of Theorems~\ref{theorem_general},  \ref{theorem_sparse}, and~\ref{theorem_gaussian} are fundamentally the same.
	In the constrained case, first Lemma~\ref{lemma:conditional_density} is applied to calculate the conditional density of $X$ given $(\hat\theta,\hat g)$
	as the expression $p_{\theta_0}(\cdot\mid \hat\theta,\hat g)$ given in the lemma. This then justifies
	the sampling distribution used for the copies $\tilde X^{(m)}$, i.e., $p_{\hat\theta}(\cdot\mid \hat\theta,\hat g)$, and the distance to exchangeability is then bounded by
	bounding $d_{\textnormal{TV}}(p_{\theta_0}(\cdot\mid \hat\theta,\hat g), p_{\hat\theta}(\cdot\mid \hat\theta,\hat g))$. 
	
	In examining the $\ell_{1}$-penalized case, the arguments are exactly identical. First, by applying Lemma~\ref{lemma:conditional_density_l1pen}
	in place of Lemma~\ref{lemma:conditional_density}, the reasoning of Section~\ref{sec:step1_proof_mainthm} 
	verifies that it suffices to bound $\E_{Q_{\theta_{0}}^{*}}\left[d_{\textnormal{TV}}(p_{\theta_{0}}(\cdot\mid  \hat\theta, \hat{g}), p_{\hat\theta}(\cdot\mid \hat\theta, \hat{g}))\right]$, where $Q_{\theta_0}^{*}$ is now defined as the distribution of $(\hat\theta(X,W), \hat{g}(X,W))$ conditioning on the event that $(X,W)\in \Omega_{\textnormal{SSOSP}, S}^{\textnormal{pen}}$ where
	\[
	\Omega_{\textnormal{SSOSP}}^{\textnormal{pen}} = \left\{(x, w) \in \X\times\R^{d}  : \hat\theta(x, w)  \textnormal{ is a SSOSP of~\eqref{eqn:def_thetahat_l1pen}} \right\},\]
	i.e., we are conditioning on the event 
	of finding a SSOSP for the $\ell_1$-penalized (rather than constrained) optimization problem. The calculation of the bound on this expected
	total variation distance is then identical to the constrained case.
	
	\subsection{Proof of Lemma \ref{lemma:conditional_density_l1pen}: conditional density for aCSS with an  $\ell_{1}$ penalty}
	Now we revisit the proof of Lemma~\ref{lemma:conditional_density} and  revise it for the $\ell_{1}$-penalized case.  
	Define a subset of $\Theta$ with support $S$ as
	$$\Theta_S = \{\theta\in\Theta: S(\theta) = S\}.$$ 
	Further define
	\[
	\Omega_{\textnormal{SSOSP}, S}^{\textnormal{pen}} = \left\{(x, w) \in \X\times\R^{d}  : \hat\theta(x, w)  \textnormal{ is a SSOSP of~\eqref{eqn:def_thetahat_l1pen}, and $S(\hat\theta(x,w))=S$} \right\}.\]
	By a result analogous to Lemma \ref{lemma:bij}, we have a bijection between $\Omega_{\textnormal{SSOSP}, S}^{\textnormal{pen}}$ and $\Psi_{\textnormal{SSOSP}, S}^{\textnormal{pen}}$, where
	\begin{multline*}\Psi_{\textnormal{SSOSP}, S}^{\textnormal{pen}}=  \bigg\{
		(x, \theta, g)\in\X \times\Theta_S \times\R^{d} : \exists w\in\R^d\textnormal{ such that }\\ \theta = \hat\theta(x,w)\textnormal{ is a SSOSP of~\eqref{eqn:def_thetahat_l1pen}, and $g=\hat g(x,w)$} \bigg\},\end{multline*}  
	which is defined by the map $
	\psi_S: (x, w )\rightarrow  \left(x,  \hat\theta(x, w),  \hat{g}(x, w) \right)$,  with inverse 
	$
	\psi_S^{-1}:  (x,  \theta, g)\rightarrow \left(x,  \frac{g - \nabla_{\theta}\L(\theta; x)}{\sigma} \right).
	$
	
	Consider the joint distribution $(X, W)\sim P_{\theta_{0}}\times \mathcal{N}(0, \frac{1}{d}\I_{d})$. 
	By assumption, the event $(X, W )\in\Omega_{\textnormal{SSOSP},S}^{\textnormal{pen}}$ has positive probability. Then the joint density of $(X, W )$, conditioning on the event that $\hat\theta(X, W)$ is a SSOSP of \eqref{eqn:def_thetahat_l1pen} with support $S$, i.e., $(X, W )\in\Omega_{\textnormal{SSOSP},S}^{\textnormal{pen}}$, is proportional to the function
	$$
	h_{\theta_0}(x, w) = f(x;\theta_{0})\exp\left\{-\frac{d}{2}\|w \|^{2} \right\}\id_{(x, w )\in \Omega_{\textnormal{SSOSP}, S}^{\textnormal{pen}}}.
	$$		 
	For any measurable set $I_S\subseteq \Psi_{\textnormal{SSOSP}, S}^{\textnormal{pen}}$, define  $$\psi_{S}^{-1}(I_S) = \{(x, w )\in \Omega_{\textnormal{SSOSP}, S}^{\textnormal{pen}}: \psi_{S}(x, w )\in I_S)\}.$$
	Then, following the same calculation for 
	\[\P\left\{(X, \hat\theta(X, W), \hat{g}(X, W))\in I_\A\mid (X,\hat\theta(X,W),\hat{g}(X,W))\in\Psi_{\textnormal{SSOSP},\A} \right\}\]
	as in the proof of Lemma \ref{lemma:conditional_density} (with $\Omega_{\textnormal{SSOSP}, \A}$ replaced by $\Omega_{\textnormal{SSOSP}, S}^{\textnormal{pen}}$), we  have
	\begin{multline*}
		\P\left\{(X, \hat\theta(X, W), \hat{g}(X, W))\in I_S\mid (X,\hat\theta(X,W),\hat{g}(X,W))\in\Psi_{\textnormal{SSOSP}, S}^{\textnormal{pen}} \right\}  \\
		= \frac{\int_{\X} f(x;\theta_{0})\int_{\R^d}e^{-\frac{d}{2\sigma^2}\|\hat{g}(x,w)-\nabla_\theta\L(\hat\theta(x,w);x) \|^{2} }\id_{(x, w )\in  \psi_S^{-1}(I_S)}\;\mathsf{d}w \;\mathsf{d}\nu_{\X}(x)}{\int_{\X\times\R^{d}}h_{\theta_0}(x', w')\;\mathsf{d}\nu_{\X}(x')\;\mathsf{d}w' }.
	\end{multline*}
	
	Next we need to reparametrize $(\hat\theta,\hat{g})$, since, as in the constrained case, these parameters, which each have dimension $d$, actually 
	contain only $d$ degrees of freedom in total (i.e., since there is a bijection between $(x,w)$ and $(x,\hat\theta,\hat{g})$, and $w\in\R^d$). 
	In fact, in the $\ell_1$-penalized setting, this is simple: once we condition on the event that $S(\hat\theta)=S$, this implies
	that $\hat\theta_{S^\complement} = \mathbf{0}_{d-|S|}$, and that $\hat{g}_S =\lambda \textnormal{sign}(\hat\theta_S)$. In other words,
	$(\hat\theta_S,\hat{g}_{S^\complement})$ captures the full information contained in $(\hat\theta,\hat{g})$---which agrees with our calculation of degrees of freedom since $|S|+|S^\complement|=d$.
	For convenience, we now define $\mathbf{I}_S$ as the $d$-by-$|S|$ matrix obtained by taking the $d$-by-$d$ identity and extracting
	columns corresponding to $S$, and $\mathbf{I}_{S^\complement}$ similarly for $S^\complement$. Then, for $(x,\theta,g)\in \Psi_{\textnormal{SSOSP},S}$,  we have calculated
	\[\theta = \mathbf{I}_S \theta_S, \ g = \mathbf{I}_S \cdot\lambda\textnormal{sign}(\theta_S) + \mathbf{I}_{S^\complement}\cdot g_{S^\complement}.\]

	Next,  if $(x,\theta,g)\in \Psi_{\textnormal{SSOSP},S}$ then by the SSOSP conditions we must have some $w$ such that
	$ \theta = \hat\theta(x,w)$  is a SSOSP of~\eqref{eqn:def_thetahat_l1pen}, and $g=\hat g(x,w) =  \nabla_\theta\L(\theta;x,w) =  \nabla_\theta\L(\theta;x) + \sigma w$. 
	Combining with the work above, we can write
	\[w = \phi_x(\theta_S,g_{S^\complement}) \textnormal{ where }\phi_x(\theta_S,g_{S^\complement}) = \frac{ \mathbf{I}_S \cdot\lambda\textnormal{sign}(\theta_S) + \mathbf{I}_{S^\complement}\cdot g_{S^\complement} -   \nabla_\theta\L( \mathbf{I}_S \theta_S;x)}{\sigma},\]
	and so 
	\[\theta = \hat\theta(x,w) = \hat\theta\left(x, \phi_x(\theta_S,g_{S^\complement})\right), \ g = \hat{g}(x,w) = \hat{g}\left(x,  \phi_x(\theta_S,g_{S^\complement})\right).\]
	Therefore,
	\[\theta_S =\mathbf{I}_S^\top \hat\theta\left(x, \phi_x(\theta_S,g_{S^\complement})\right), \ g_{S^\complement} = \mathbf{I}_{S^\complement}^\top  \hat{g}\left(x,  \phi_x(\theta_S,g_{S^\complement})\right).\]
	We can also calculate
	\[\nabla_{\theta_S} \phi_x(\theta_S,g_{S^\complement}) = -\sigma^{-1}\mathbf{I}_S^\top  \nabla^2_\theta \L(\mathbf{I}_S \theta_S;x)
	\]
	and\[\nabla_{g_{S^\complement}} \phi_x(\theta_S,g_{S^\complement}) = \sigma^{-1}  \mathbf{I}_{S^\complement}^\top .\]
	Therefore,
	\begin{align*} \textnormal{det}\left( \nabla\phi_x(\theta_S,g_{S^\complement})\right)
		&= \textnormal{det}\left( \left(\begin{array}{c} \nabla_{\theta_S}\phi_x(\theta_S,g_{S^\complement})\\  \nabla_{g_{S^\complement}}\phi_x(\theta_S,g_{S^\complement})\end{array}\right)\right)\\
		&= \textnormal{det}\left( \left(\begin{array}{c} \nabla_{\theta_S}\phi_x(\theta_S,g_{S^\complement})\\  \nabla_{g_{S^\complement}}\phi_x(\theta_S,g_{S^\complement})\end{array}\right) \cdot(\mathbf{I}_S \ \mathbf{I}_{S^\complement}) \, \right) \\
		&= \textnormal{det}\left( \left(\begin{array}{c@{\ \ }c} \nabla_{\theta_S}\phi_x(\theta_S,g_{S^\complement})\mathbf{I}_S  &  \nabla_{\theta_S}\phi_x(\theta_S,g_{S^\complement})\mathbf{I}_{S^\complement} \\  \nabla_{g_{S^\complement}}\phi_x(\theta_S,g_{S^\complement})\mathbf{I}_S & \nabla_{g_{S^\complement}}\phi_x(\theta_S,g_{S^\complement})\mathbf{I}_{S^\complement}\end{array}\right)\right) \\
		&=
		\textnormal{det}\left( \left(\begin{array}{c@{\ \ }c}  -\sigma^{-1}\mathbf{I}_S^\top  \nabla^2_\theta \L(\mathbf{I}_S\theta_S;x) \mathbf{I}_S  &  -\sigma^{-1}\mathbf{I}_S^\top  \nabla^2_\theta \L(\mathbf{I}_S\theta_S;x)\mathbf{I}_{S^\complement} \\  \sigma^{-1} \mathbf{I}_{S^\complement}^\top \mathbf{I}_S &\sigma^{-1} \mathbf{I}_{S^\complement}^\top \mathbf{I}_{S^\complement}\end{array}\right)\right) \\
		&=
		\textnormal{det}\left( \left(\begin{array}{c@{\ \ }c}  -\sigma^{-1}\mathbf{I}_S^\top  \nabla^2_\theta \L(\mathbf{I}_S\theta_S;x) \mathbf{I}_S  &  -\sigma^{-1}\mathbf{I}_S^\top  \nabla^2_\theta \L(\mathbf{I}_S\theta_S;x)\mathbf{I}_{S^\complement} \\  0&\sigma^{-1} \I_{d-|S|}\end{array}\right)\right) \\
		&= (-1)^{|S|}\sigma^{-d} \cdot \textnormal{det}\left(\mathbf{I}_S^\top  \nabla^2_\theta \L(\mathbf{I}_S\theta_S;x) \mathbf{I}_S\right)\\
		&= (-1)^{|S|}\sigma^{-d} \cdot \textnormal{det}\left(\nabla^2_\theta \L(\mathbf{I}_S\theta_S;x)_S\right).
	\end{align*}
	From this point on, following similar arguments as \cite[Section B.4]{barber2020testing} to verify the validity of applying the change-of-variables formula for integration, we calculate
\begin{multline*}
	\int_{\R^d}
	\exp\left(-\frac{d}{2\sigma^2}	\left\|\hat{g}(x, w)- \nabla_\theta \L(\hat\theta(x, w); x)	\right\|^{2}	\right)	\cdot\id_{(x, w) \in \psi_{S}^{-1}(I_{S})}	\;	\mathsf{d}w	\\[6pt]
	=	\sigma^{-d}	\int_{\R^{|S|}}	\int_{\R^{d - |S|}}\exp\left(-\frac{d}{2\sigma^2}	\left\|	\mathbf{I}_S \cdot \lambda \, \textnormal{sign}(\theta_S)	+ \mathbf{I}_{S^\complement} g_{S^\complement}	- \nabla_\theta \L(\mathbf{I}_S \theta_S; x)	\right\|^{2}	\right)\\[6pt]
	\times	\det\left(	\nabla^2_\theta \L(\mathbf{I}_S \theta_S; x)_S	\right)	\cdot	\id_{\left(	x,\,	\phi_x(\theta_S, g_{S^\complement})		\right)		\in			\psi_{S}^{-1}(I_{S})}	\;	\mathsf{d}g_{S^\complement}	\;	\mathsf{d}\theta_S.
\end{multline*}
where we note that $\textnormal{det}\left(  \nabla^2_\theta \L(\mathbf{I}_S\theta_S;x)_S\right)$  must be positive, by the SSOSP conditions. We can also verify from our definitions that $ \id_{(x, \phi_x(\theta_S,g_{S^\complement}) )\in  \psi_{S}^{-1}(I_{S})} = \id_{(x,\mathbf{I}_S\theta_S,\mathbf{I}_S\cdot\lambda\textnormal{sign}(\theta_S) + \mathbf{I}_{S^\complement} g_{S^\complement})\in I_S}$. 
With this calculation in place we then have
\begin{multline*}
	\P\Big\{	(X, \hat\theta(X, W), \hat{g}(X, W)) \in I_{S}\;\Big|\;(X, \hat\theta(X, W), \hat{g}(X, W)) \in \Psi_{\textnormal{SSOSP},S}\Big\} \\[4pt]
	= \int_{\X} f(x;\theta_{0}) \int_{\R^{|S|}} \int_{\R^{d - |S|}} \frac{\exp\left(	-\frac{d}{2\sigma^2}	\left\|\mathbf{I}_S \cdot \lambda \, \textnormal{sign}(\theta_S)+ \mathbf{I}_{S^\complement}g_{S^\complement}- \nabla_\theta \L(\mathbf{I}_S \theta_S; x)\right\|^{2}\right)}{\sigma^{d}\int_{\X \times \R^{d}}h_{\theta_0}(x', w')\;\mathsf{d}\nu_{\X}(x')\;\mathsf{d}w'} \\[6pt]
	\times\det\left(\nabla^2_\theta \L(\mathbf{I}_S \theta_S; x)_S\right)\cdot\id_{\left(x,\,\mathbf{I}_S \theta_S,\,\mathbf{I}_S \cdot \lambda \, \textnormal{sign}(\theta_S)+ \mathbf{I}_{S^\complement} g_{S^\complement}\right) \in I_S}\;\mathsf{d}g_{S^\complement}	\;\mathsf{d}\theta_S\;\mathsf{d}\nu_{\X}(x).
\end{multline*}
In particular, this verifies that 			
\begin{multline*}
	\frac{ f(x; \theta_{0}) \cdot \exp\left( -\frac{d}{2\sigma^2} \left\| \mathbf{I}_S \cdot \lambda \, \textnormal{sign}(\theta_S) + \mathbf{I}_{S^\complement} g_{S^\complement}- \nabla_\theta \L(\mathbf{I}_S \theta_S; x) \right\|^2 \right)}{\sigma^{d}\cdot\displaystyle\int_{\X \times \R^d}h_{\theta_0}(x', w')\;\mathsf{d}\nu_{\X}(x')\;\mathsf{d}w'}\\
	\times\, \det\left(\nabla^2_\theta \L(\mathbf{I}_S \theta_S; x)_S\right)\cdot\,\id_{\left(x,\,\mathbf{I}_S \theta_S,\,\mathbf{I}_S \cdot \lambda \, \textnormal{sign}(\theta_S)+ \mathbf{I}_{S^\complement} g_{S^\complement}\right)\in\Psi_{\textnormal{SSOSP}, S}}
\end{multline*}
is the joint density of $(X,\hat{\theta}_S,\hat{g}_{S^\complement}) = (X,  \hat\theta(X, W)_S,  \hat{g}(X, W)_{S^\complement})$, conditional on the event $(X, \hat\theta(X, W), \hat{g}(X, W))\in\Psi_{\textnormal{SSOSP},S}$. 
Therefore, the conditional density of $X\mid (\hat\theta_S, \hat{g}_{S^\complement})$ (again conditioning on this same event) can be written as
\begin{multline*}
	\propto f(x; \theta_{0}) \cdot \exp\left( -\frac{d}{2\sigma^2} \left\| \mathbf{I}_S \cdot \lambda \, \textnormal{sign}(\theta_S) + \mathbf{I}_{S^\complement} g_{S^\complement}- \nabla_\theta \L(\mathbf{I}_S \theta_S; x) \right\|^2 \right) \\
	\times\, \det\left(\nabla^2_\theta \L(\mathbf{I}_S \theta_S; x)_S\right)\cdot\,\id_{\left(x,\,\mathbf{I}_S \theta_S,\,\mathbf{I}_S \cdot \lambda \, \textnormal{sign}(\theta_S)+ \mathbf{I}_{S^\complement} g_{S^\complement}\right)\in\Psi_{\textnormal{SSOSP}, S}}
\end{multline*}
Moreover, $\hat\theta_S$ and $\hat{g}_{S^\complement}$ uniquely determine $\hat\theta$ and $\hat g$ on the event that $S$ is the support, as described earlier,
so we can equivalently condition on $(\hat\theta,\hat g)$ and  can rewrite this density as
\begin{equation} 
	p_{\theta_{0}}(\cdot\mid \hat\theta, \hat{g}) \propto f(x;\theta_{0}) \cdot e^{ - \frac{d}{2\sigma^2}\|\hat{g}- \nabla_\theta \L(\hat\theta;x)\|^2}\cdot\textnormal{det}\left( \nabla^2_\theta\L(\hat\theta;x) _S\right)\cdot\id_{(x,\hat\theta,\hat{g})\in \Psi_{\textnormal{SSOSP},S}}.
\end{equation}
Finally, by definition, $(x,\hat\theta,\hat{g})\in \Psi_{\textnormal{SSOSP},S}$ if and only if $\hat\theta\in\Theta_S$ and $x\in\X_{\hat\theta,\hat{g}}$, so $$\id_{(x,\hat\theta,\hat{g})\in \Psi_{\textnormal{SSOSP},S}}= \id_{x\in\X_{\hat\theta,\hat{g}}}$$ for $\hat\theta\in\Theta_S$.

\section{Additional proofs}
\label{sec:appdex_otherproof}

\subsection{Verifying that the plug-in version of $p_{\theta_0}(\cdot\mid \hat\theta, \hat{g})$ defines a density}\label{sec:app_verify_density}
To ensure  that our procedure is well-defined in both constrained and $\ell_{1}$-penalized cases, we need to  verify that the plug-in version of the conditional density  
\[p_{\hat\theta}(\cdot\mid \hat\theta, \hat{g})\propto p_{\hat\theta,\hat{g}}^{\textnormal{un}}(x) \]
defines a valid density with respect to $\nu_{\X}$, where $p_{\theta, g}^{\textnormal{un}}(x)$ represents the unnormalized density, namely,
\[p_{\theta, g}^{\textnormal{un}}(x) = f(x;\theta) \cdot e^{ - \frac{d}{2\sigma^2}\|g- \nabla_\theta \L(\theta;x)\|^2}\cdot\textnormal{det}\left(U_{\A(\theta)}^\top \nabla^2_\theta\L(\theta;x) U_{\A(\theta)}\right)\cdot\id_{(x,\theta,g)\in \Psi_{\textnormal{SSOSP},\A(\theta)}}\]
in the constrained case as in~\eqref{eqn:conditional_density_thetahat}; and 
\[p_{\theta, g}^{\textnormal{un}}(x) = f(x;\theta) \cdot e^{ - \frac{d}{2\sigma^2}\|g- \nabla_\theta \L(\theta;x)\|^2}\cdot\textnormal{det}\left(\nabla^2_\theta\L(\theta;x)_{S(\theta)}\right)\cdot\id_{(x,\theta,g)\in \Psi_{\textnormal{SSOSP},S(\theta)}^{\textnormal{pen}}}\]
in the $\ell_{1}$-penalized case as in~\eqref{eqn:conditional_density_thetahat_l1pen}.
To verify this we only need to check that this unnormalized density integrates to a finite and positive value (the analogous result
for aCSS appears in \cite[Section B.3]{barber2020testing}).

\begin{lemma}
	If Assumption~\ref{asm:reg} and ~\ref{asm:L_H} hold, then for $\theta\in\Theta$ and $g\in\R^{d}$,
	the unnormalized density $p_{\theta, g}^{\textnormal{un}}(x)$ is nonnegative and integrable with respect to $\nu_{\mathcal{\X}}$. Furthermore, if the event $\hat\theta = \hat\theta(X, W)$ is a SSOSP has positive probability, then conditional on this event, $\int_{\mathcal{\X}}p_{\hat\theta, \hat{g}}^{\textnormal{un}}(x)\mathsf{d}\nu_{\mathcal{\X}}(x) >0$ holds almost surely.
\end{lemma}
\begin{proof}  
	\noindent{\bf Constrained case:} We first check nonnegativity. 
	For any $\theta\in\Theta$ and any $x$, we have $f(x;\theta)>0$ by Assumption~\ref{asm:reg}. Furthermore, if $x\in\X_{\theta,g}$ then $\textnormal{det}\left(U_{\A(\theta)}^\top \nabla^2_\theta\L(\theta;x) U_{\A(\theta)}\right)>0$ by definition of $\X_{\theta,g}$ and the SSOSP conditions. This verifies the  nonnegativity for $p_{\theta, g}^{\textnormal{un}}(x)$ for any $(\theta, g,x)$. 
	Next we check integrability.
	\begin{align*}
		&\int_{\mathcal{\X}}p_{\theta,g}^{\textnormal{un}}(x)\mathsf{d}\nu_{\mathcal{\X}}(x)\\
		&\le\int_{\mathcal{\X}}f(x;\theta)\cdot\textnormal{det}\left(U_{\A(\theta)}^\top \nabla^2_\theta\L(\theta;x) U_{\A(\theta)}\right)\cdot\id_{U_{\A(\theta)}^\top \nabla^2_\theta\L(\theta;x) U_{\A(\theta)} \succ 0}\mathsf{d}\nu_{\mathcal{\X}}(x)\\
		&\le \int_{\mathcal{\X}}f(x;\theta)\cdot\left(\lambda_{\max}\left(U_{\A(\theta)}^\top \nabla^2_\theta\L(\theta;x) U_{\A(\theta)}\right)\right)^d  \cdot\id_{U_{\A(\theta)}^\top \nabla^2_\theta\L(\theta;x) U_{\A(\theta)} \succ 0}\mathsf{d}\nu_{\mathcal{\X}}(x)\\
		&\le \int_{\mathcal{\X}}f(x;\theta)\cdot\left(\lambda_{\max}\left( \nabla^2_\theta\L(\theta;x)\right)\right)_{+}^d \mathsf{d}\nu_{\mathcal{\X}}(x)\\
		&\le\frac{d!}{r(\theta)^{2d}}\int_{\mathcal{\X}}f(x;\theta)\cdot\exp\Big\{r(\theta)^2(\lambda_{\max}\left(H(\theta, x) - H(\theta)\right))_{+} \\
		&\qquad\qquad\qquad\qquad+r(\theta)^2\lambda_{\max}\left(H(\theta) - \nabla_{\theta}^{2}\mathcal{R}(\theta)\right)_{+} \Big\} \mathsf{d}\nu_\X(x)\\
		&=\frac{d!}{r(\theta)^{2d}}\exp\left\{r(\theta)^2\lambda_{\max}\left(H(\theta) - \nabla_{\theta}^{2}\mathcal{R}(\theta)\right)_{+} \right\} \cdot \E_{P_\theta}\left[\exp\left\{r(\theta)^2\lambda_{\max}\left(H(\theta, x) - H(\theta)\right)_{+}\right\}\right]\\
		&\le \frac{d!}{r(\theta)^{2d}}e^{\epsilon(\theta)}\exp\left\{r(\theta)^2\lambda_{\max}\left(H(\theta) - \nabla_{\theta}^{2}\mathcal{R}(\theta)\right)_{+} \right\},
	\end{align*}
	where the third-to-last step holds since $t^d \leq d!e^d$ for any $t\geq 0$, and the last step holds by applying Assumption~\ref{asm:L_H}.
	This verifies that $\int_{\mathcal{\X}}p_{\theta,g}^{\textnormal{un}}(x)\mathsf{d}\nu_{\mathcal{\X}}(x)$ is finite. 
	Finally, we  check $\int_{\mathcal{\X}}p_{\hat\theta, \hat{g}}^{\textnormal{un}}(x)\mathsf{d}\nu_{\mathcal{\X}}(x) >0$ holds almost surely. For any $x$, we have $\frac{f(x, \theta_0)}{f(x, \hat\theta)}>0$ by Assumption \ref{asm:reg}. Combined with the fact that $p_{\hat\theta, \hat{g}}^{\textnormal{un}}(x)$ is nonnegative as proved above, it is therefore equivalent  to verify that $\int_{\mathcal{\X}}\frac{f(x, \theta_0)}{f(x, \hat\theta)}p_{\hat\theta, \hat{g}}^{\textnormal{un}}(x)\mathsf{d}\nu_{\mathcal{\X}}(x) >0$. This last claim must hold since $p_{\theta_0}(x\mid\hat\theta, \hat{g})\propto\frac{f(x, \theta_0)}{f(x, \hat\theta)}p_{\hat\theta, \hat{g}}^{\textnormal{un}}(x)$ is the conditional density of $X\mid\hat\theta, \hat{g}$.
	\bigskip
	
	\noindent{\bf $\ell_{1}$-penalized case:} The proof for this case mirrors that for the constrained case.
	For any $\theta\in\Theta$ and $x$, we have $f(x;\theta)>0$ by Assumption~\ref{asm:reg}. Furthermore, if $(x,\theta,g)\in \Psi_{\textnormal{SSOSP},S(\theta)}^{\textnormal{pen}}$ then $\textnormal{det}\left( \nabla^2_\theta\L(\theta;x)_{S(\theta)}\right)>0$ by definition of $\Psi_{\textnormal{SSOSP},S(\theta)}^{\textnormal{pen}}$ and the SSOSP conditions. This verifies the  nonnegativity of $p_{\theta, g}^{\textnormal{un}}(x)$ for any $(\theta, g,x)$. 
	To check integrability, we have 
	\begin{align*}
		&\int_{\mathcal{\X}}p_{\theta,g}^{\textnormal{un}}(x)\mathsf{d}\nu_{\mathcal{\X}}(x) 
		\le\int_{\mathcal{\X}}f(x;\theta)\cdot\textnormal{det}\left(\nabla^2_\theta\L(\theta;x)_{S(\theta)}\right)\cdot\id_{\nabla^2_\theta\L(\theta;x)_{S(\theta)} \succ 0}\mathsf{d}\nu_{\mathcal{\X}}(x)\\
		&\le \int_{\mathcal{\X}}f(x;\theta)\cdot\left(\lambda_{\max}\left(\nabla^2_\theta\L(\theta;x)_{S(\theta)}\right)\right)^d  \cdot\id_{\nabla^2_\theta\L(\theta;x)_{S(\theta)} \succ 0}\mathsf{d}\nu_{\mathcal{\X}}(x)\\
		&\le \int_{\mathcal{\X}}f(x;\theta)\cdot\left(\lambda_{\max}\left( \nabla^2_\theta\L(\theta;x)\right)\right)_{+}^d \mathsf{d}\nu_{\mathcal{\X}}(x)\\
		&\le\frac{d!}{r(\theta)^{2d}}\int_{\mathcal{\X}}f(x;\theta)\cdot\exp\Big\{r(\theta)^2\left(\lambda_{\max}\left(H(\theta, x) - H(\theta)\right)\right)_{+} \\
		&\qquad\qquad\qquad\qquad+r(\theta)^2\left(\lambda_{\max}\left(H(\theta) - \nabla_{\theta}^{2}\mathcal{R}(\theta)\right)\right)_{+} \Big\} \\
		&\le \frac{d!}{r(\theta)^{2d}}e^{\epsilon(\theta)}\exp\left\{r(\theta)^2\left(\lambda_{\max}\left(H(\theta) - \nabla_{\theta}^{2}\mathcal{R}(\theta)\right)\right)_{+} \right\}.
	\end{align*}
	Finally,  $\int_{\mathcal{\X}}p_{\hat\theta, \hat{g}}^{\textnormal{un}}(x)\mathsf{d}\nu_{\mathcal{\X}}(x) >0$ holds almost surely for the same reason as in the constrained case.  
\end{proof}

\subsection{Proof of Lemma \ref{lemma:bij}}
\begin{proof}
	First we check that $\psi_{\A}$ is injective on $\Omega_{\textnormal{SSOSP}, \A}$. For any $(x_{1}, w_{2}), (x_{2}, w_{2}) \in\Omega_{\textnormal{SSOSP}, \A}$, if  $\psi_{\A}(x_{1}, w_{1}) = \psi_{\A}(x_{2}, w_{2}) = (x,  \theta, g)$, then by definition of $\psi_{\A}$, we have $x_{1} = x_{2} = x$ trivially. By definition of $\psi_{\A}$ and $\hat g$, 
	$$\nabla_{\theta}\L(\theta;x)+\sigma w_{1} = \hat g(x_{1}, w_{1}) =  g = \hat g(x_{2}, w_{2}) =  \nabla_{\theta}\L(\theta;x)+\sigma w_{2},$$
	therefore $w_{1} = w_{2} = \frac{g - \nabla_{\theta}\L(\theta;x)}{\sigma}$. This establishes that $\Psi_{\A}$ is injective and that the inverse function (on the image of $\psi_{\A}$) is given as claimed above.
	
	Then we verify that $\Psi_{\textnormal{SSOSP}, \A}$ is the image of $\psi_{\A}$.
	Suppose $(x, \theta, g)\in \psi_{\A}(\Omega_{\textnormal{SSOSP}, \A})$, i.e, for some $w$ such that $(x, w )\in \Omega_{\textnormal{SSOSP}, \A}$, we have $\theta =  \hat{\theta}(x, w )$, which is a SSOSP with active set $\A$, and $g = \nabla_{\theta}\L(\hat{\theta}(x, w ); x, w) = \hat{g}(x, w)$. Then for this $w$,  $\theta = \hat\theta(x, w )\in\Theta_{\A}$, and $g = \hat g(x,w)$.  Therefore, $(x, \theta, g)\in \Psi_{\textnormal{SSOSP}, \A}$, and so we have shown that $\psi_{\A}(\Omega_{\textnormal{SSOSP}, \A})\subseteq \Psi_{\textnormal{SSOSP}, \A}$.
	
	Conversely suppose that $(x, \theta, g)\in \Psi_{\textnormal{SSOSP}, \A}$. By definition of $\Psi_{\textnormal{SSOSP}, \A}$, there exists $ w $ such that $\theta = \hat\theta(x, w )$ is a SSOSP of \eqref{eqn:def_thetahat} with active set $\A$, and $g = \hat g(x, w )$. Therefore, for this $w$ we have $(x, w ) \in \Omega_{\textnormal{SSOSP}, \A}$. Then $(x,  \theta, g) = (x, \hat\theta(x,w),\hat g(x,w)) = \psi_{\A}(x, w )\in \psi_{\A}(\Omega_{\textnormal{SSOSP}, \A})$. This verifies that $\Psi_{\textnormal{SSOSP}, \A} \subseteq \psi_{\A}(\Omega_{\textnormal{SSOSP}, \A})$, and thus completes the proof.
\end{proof} 

\subsection{Proof of Lemma \ref{lemma:proj_chisq}}
\begin{proof}
	Fix any $\lambda\in(0,1/2)$. We calculate
	\begin{align*}
		e^{\lambda h_v(k)} 
		&= \exp\left\{\lambda \E_{Z\sim \mathcal{N}(0,\I_d)}\left[\max_{S\subseteq[p],|S|\leq k} \|\mathcal{P}_{v_S}(Z)\|^2\right]\right\}\\
		&\leq \E_{Z\sim \mathcal{N}(0,\I_d)}\left[\exp\left\{\lambda \max_{S\subseteq[p],|S|\leq k} \|\mathcal{P}_{v_S}(Z)\|^2\right\}\right]\textnormal{ by Jensen's inequality}\\
		&=\E_{Z\sim \mathcal{N}(0,\I_d)}\left[\max_{S\subseteq[p],|S|\leq k}\exp\left\{\lambda  \|\mathcal{P}_{v_S}(Z)\|^2\right\}\right]\\
		&\leq \E_{Z\sim \mathcal{N}(0,\I_d)}\left[\sum_{S\subseteq[p],|S|= k}\exp\left\{\lambda  \|\mathcal{P}_{v_S}(Z)\|^2\right\}\right]  \\
		&= \sum_{S\subseteq[p],|S|= k} \E_{Z\sim \mathcal{N}(0,\I_d)}\left[\exp\left\{\lambda  \|\mathcal{P}_{v_S}(Z)\|^2\right\}\right].
	\end{align*}
	Since  $\|\mathcal{P}_{v_S}(Z)\|^2\sim \chi^2_{\textnormal{dim}(\textnormal{span}(\{v_i\}_{i\in S}))}$, we have 
	\begin{align*}
		e^{\lambda h_v(k)} &\le \sum_{S\subseteq[p],|S|=k} \left(1-2\lambda\right)^{-\frac{1}{2}\textnormal{dim}(\textnormal{span}(\{v_i\}_{i\in S}))} \\
		&\leq  \sum_{S\subseteq[p],|S|= k} \left(1-2\lambda\right)^{-k/2}
		={p\choose k}\left(1-2\lambda\right)^{-k/2} \leq \left(\frac{ep}{k}\right)^k \left(1-2\lambda\right)^{-k/2}.
	\end{align*}
	Therefore,
	\[h_v(k) \leq \inf_{\lambda\in(0,1/2)}\left\{\lambda^{-1}\log\left[\left(\frac{ep}{k}\right)^k\left(1-2\lambda\right)^{-k/2}\right]\right\}
	= \frac{k}{2} \inf_{\lambda\in(0,1/2)} \left\{\frac{2\log(ep/k) - \log(1-2\lambda)}{\lambda}\right\}.\]
	Taking $\lambda = 1/4$, 
	\[h_v(k) \leq 2k \left(2\log(ep/k) - \log(1/2)\right)\leq 4k\log(4p/k).\]
	
	Finally, we have $\max_{S\subseteq[p],|S|\leq k} \|\mathcal{P}_{v_S}(Z)\|^2\leq \|Z\|^2$, and therefore,
	\[h_v(k) = \E_{Z\sim \mathcal{N}(0,\I_d)}\left[\max_{S\subseteq[p],|S|\leq k} \|\mathcal{P}_{v_S}(Z)\|^2\right]
	\leq \E_{Z\sim \mathcal{N}(0,\I_d)}\left[\|Z\|^2\right] = d,\]
	since $\|Z\|^2\sim\chi^2_d$.

\end{proof}

\section{Checking assumptions for examples}  \label{sec:check_asm}
In this section, we verify that Assumptions~\ref{asm:reg},~\ref{asm:SSOSP_general}, and~\ref{asm:L_H} hold for the three examples considered
in Section~\ref{sec:numerical}: the Gaussian mixture model (Example~\ref{example:mixgaussian}), isotonic Gaussian linear regression (Example~\ref{example:isotonic}),
and sparse high-dimensional Gaussian linear regression (Example~\ref{example:sparse}).
\subsection{Verifying assumptions for Examples \ref{example:isotonic} (isotonic regression) and \ref{example:sparse} (sparse regression)}
We first verify the assumptions for the two examples in the Gaussian linear model setting, since these are more straightforwards.
First, Assumption~\ref{asm:reg} holds trivially by construction---we have $\Theta = \R^d$, and twice-differentiability of $\L(\theta;x)$
holds both with and without the ridge penalty.

Next we check Assumption~\ref{asm:SSOSP_general}.
In both examples, the optimization problem that defines $\hat\theta(X,W)$ is strongly convex,
meaning that we can define $\hat\theta(X,W)$ as the unique minimizer, and the SSOSP conditions then hold surely. 
Next we need to verify a high probability bound on $\|\hat\theta(X,W) - \theta_0\|$.
First, for isotonic regression,  
we see that $\hat\theta(X,W)$ can equivalently be written as
\[\hat\theta(X,W)=\arg\min_{\theta\in\R^d}\left\{\frac{1}{2}\|\theta - (X - \sigma W)\|^2_2 : \theta_1 \leq \dots \leq \theta_n\right\},\]
i.e., the isotonic projection of $X - \sigma W$. Since $X - \sigma W \sim \mathcal{N}(\theta_0, (\nu^2+\sigma^2/n)\mathbf{I}_n)$, 
applying the result of \cite[Theorem 5 and Appendix A.1]{yang2019contraction} we have a high-probability bound on the error,
\[\|\hat\theta(X,W) - \theta_0\| \leq O\left(n^{1/6}(\log n)^{1/3}(1+\sigma^2)^{2/3}\right) \textnormal{ with probability $\geq 1-1/n$.}\]
If we choose $\sigma = O(1)$, we can therefore take $r(\theta_0) = O\left(n^{1/6}(\log n)^{1/3}\right)$ and $\delta(\theta_0) = 1/n$.

Next, for sparse regression, the calculation is a bit more complex.
Our argument closely follows the framework developed in \cite[Theorem 1]{negahban2012unified}.
Let $\Delta = \hat\theta(X,W) - \theta_0$. Then by optimality of $\hat\theta(X,W)$ we have
\begin{multline*}
	\frac{1}{2\nu^2}\|X - Z(\theta_0+\Delta)\|^2_2 + \sigma(\theta_0+\Delta)^\top W + \frac{\lambda_{\textnormal{ridge}}}{2}\|\theta_0+\Delta\|^2_2+\lambda\|\theta_0+\Delta\|_1 \\
	\leq \frac{1}{2\nu^2}\|X - Z\theta_0\|^2_2 + \sigma\theta_0^\top W + \frac{\lambda_{\textnormal{ridge}}}{2}\|\theta_0\|^2_2+\lambda\|\theta_0\|_1.\end{multline*}
Rearranging terms, and writing $v = X - Z\theta_0 \sim\mathcal{N}(0,\nu^2\mathbf{I}_n)$,
\begin{multline*}\frac{1}{2}\Delta^\top\left(\frac{Z^\top Z}{\nu^2} + \lambda_{\textnormal{ridge}}\mathbf{I}_d\right)\Delta - \Delta^\top\left(\frac{Z^\top v}{\nu^2} - \sigma W - \lambda_{\textnormal{ridge}}\theta_0 \right) \leq \lambda\left(\|\theta_0\|_1 - \|\theta_0 + \Delta\|_1\right) \\\leq\lambda \|\Delta_{S(\theta_0)}\|_1 - \lambda\|\Delta_{S(\theta_0)^\complement}\|_1.\end{multline*}
Then, if the penalty parameter satisfies $\lambda\geq 2\left\| \frac{Z^\top v}{\nu^2} - \sigma W - \lambda_{\textnormal{ridge}}\theta_0\right\|_\infty$, it holds that
\[\frac{1}{2}\Delta^\top\left(\frac{Z^\top Z}{\nu^2} + \lambda_{\textnormal{ridge}}\mathbf{I}_d\right)\Delta  \leq 1.5\lambda \|\Delta_{S(\theta_0)}\|_1 - 0.5\lambda\|\Delta_{S(\theta_0)^\complement}\|_1.\]
Standard assumptions on $Z$ (namely, a restricted eigenvalue type property \citep{negahban2012unified}) will then ensure
\[\|\Delta\| \leq O\left(\sqrt{\frac{|S(\theta_0)|\log d}{n}}\right)\] 
with probability $\geq 1 - 1/n$, when we take $\nu = O(1)$, $\|\theta_0\|_\infty = O(1)$, $ \lambda_{\textnormal{ridge}} \lesssim \sqrt{n\log d}$, and $\sigma \lesssim\sqrt{nd}$. Therefore, we can take $r(\theta_0) = O\left(\sqrt{\frac{|S(\theta_0)|\log d}{n}}\right)$ and $\delta(\theta_0) = 1/n$.

Finally, we check Assumption~\ref{asm:L_H}. For isotonic regression, we have $H(\theta;x) = \nu^{-2}\mathbf{I}_d$, and for
sparse regression, $H(\theta;x) = \nu^{-2}Z^\top Z + \lambda_{\textnormal{ridge}}\mathbf{I}_d$. In both cases,
$H(\theta;x)$ does not depend on $x$, and therefore, Assumption~\ref{asm:L_H} holds trivially with $\epsilon(\theta_0) = 0$.

\subsection{Verifying assumptions for Example \ref{example:mixgaussian} (Gaussian mixture model)}
In this section, we verify that the assumptions of Theorem \ref{theorem_general} 
hold for the Gaussian mixture model setting, specifically in the case of $J=2$ components as implemented in our simulation.
Assumption~\ref{asm:reg} holds trivially by construction.
For Assumption~\ref{asm:SSOSP_general}, 
the accuracy of $\hat\theta(X,W)$ can be established 
with $r(\theta_0)\asymp \sqrt{\frac{\log n}{n}}$ and $\delta(\theta_0)\asymp n^{-1}$ via known results in the literature.
For instance, \cite[Corollary 1.4]{hardt2015tight} show this accuracy level obtained via the EM algorithm, and we can then use the EM
solution as an initialization for gradient descent within a $O(r(\theta_0))$-radius neighborhood, to find an FOSP; 
since the expected Hessian is positive definite, with high probability this FOSP is also a SSOSP.
We omit the details.

Finally, we 
check Assumption \ref{asm:L_H}, which will require some substantial calculations. To verify Assumption \ref{asm:L_H}, we will check the following stronger condition 
\begin{equation*}
	\E_{\theta_{0}}\left[\exp\left\{\sup_{\theta\in\B(\theta_{0},r(\theta_{0}) )\cap\Theta}r(\theta_{0})^2\cdot \left\|H(\theta; X)-H(\theta)\right\|\right\}\right]\le c'e^{\epsilon(\theta_0)},
\end{equation*} for any $r(\theta_0) = o(n^{-1/4})$ and $\epsilon(\theta_0)\gtrsim r(\theta_0)^2n^{1/2} + r(\theta_0)^3 n$.
We first calculate, for parameter $\theta = (\pi_1,\mu_1,\sigma_1,\mu_2,\sigma_2)$, 
\[\L(\theta;x) = -\sum_{i=1}^n \log\left(\pi_1\phi(x_i;\mu_1,\sigma^2_1)+ (1-\pi_1)\phi(x_i;\mu_2,\sigma^2_2)\right),\]
where $\phi(t;\mu,\sigma^2) = \frac{1}{\sqrt{2\pi\sigma^2}}e^{-(t-\mu)^2/2\sigma^2}$ is the density of the normal distribution.
After some calculations, we can verify that the Hessian takes the form
\begin{multline*}H(\theta;x) = \sum_{i=1}^n\Bigg[ \sum_{m=0}^2 x_i^m \cdot \bigg(a_{1,m}(\theta) f_1(x_i;\theta) + a_{2,m}(\theta) f_2(x_i;\theta)\\+b_{1,m}(\theta) f_1(x_i;\theta)^2+b_{2,m}(\theta) f_2(x_i;\theta)^2 \bigg) +  \sum_{m=0}^4 x_i^m \cdot c_m(\theta)  f_1(x_i;\theta)f_2(x_i;\theta) \Bigg], \end{multline*}
where we define
\[f_1(t;\theta) = \frac{\pi_1\phi(t;\mu_1,\sigma^2_1)}{\pi_1 \phi(t;\mu_1,\sigma^2_1)+ (1-\pi_1)\phi(t;\mu_2,\sigma^2_2)}\]
and
\[f_2(t;\theta) = \frac{(1-\pi_1)\phi(t;\mu_2,\sigma^2_2)}{\pi_1 \phi(t;\mu_1,\sigma^2_1)+ (1-\pi_1)\phi(t;\mu_2,\sigma^2_2)},\]
and where $a_{1,m},a_{2,m},b_{1,m},b_{2,m},c_m:\Theta\rightarrow \R^{5\times 5}$ are continuously differentiable functions (whose details we omit for brevity). We can rewrite this as
\[H(\theta;x) = \sum_{i=1}^n g_0(x_i;\theta) +  x_i g_1(x_i;\theta) + x_i^2 g_2(x_i;\theta)\]
where
\begin{multline*}
	g_0(t;\theta) = a_{1,0}(\theta)f_1(t;\theta) + a_{2,0}(\theta)f_2(t;\theta) \\
	+ b_{1,0}(\theta)f_1(t;\theta)^2+ b_{2,0}(\theta)f_2(t;\theta)^2 + \sum_{m=0}^4c_m(\theta)  t^m f_1(t;\theta)f_2(t;\theta)
\end{multline*}
and where
\[g_m(t;\theta) = a_{1,m}(\theta)f_1(t;\theta) + a_{2,m}(\theta)f_2(t;\theta) + b_{1,m}(\theta)f_1(t;\theta)^2+ b_{2,m}(\theta)f_2(t;\theta)^2\]
for $m=1,2$. Some additional calculations prove that we can find finite $C_m(\theta_0),C_m'(\theta_0)$ such that, as long as $r(\theta_0)$ is bounded
by some appropriately chosen constant,
\[\sup_{t\in\R}\sup_{\theta\in\B(\theta_{0},r(\theta_{0}) )\cap\Theta} \|g_m(t;\theta)\|\leq C_m(\theta_0)\]
and
\[\sup_{t\in\R}\sup_{\theta\in\B(\theta_{0},r(\theta_{0}) )\cap\Theta} \|\nabla_\theta g_m(t;\theta)\|\leq C_m'(\theta_0).\]
(To give some intuition for this---for example, for the zeroth-order term, i.e., finding $C_m(\theta_0)$, it is trivial
to see that $\sup_{t\in\R}f_\ell(t;\theta)\leq 1$ for each $\ell=1,2$; what is more subtle is the observation
that $\sup_{t\in\R} t^m f_1(t;\theta)f_2(t;\theta)$ is also finite, as long as $\mu_1\neq \mu_2$---and this condition is ensured
as long as we enforce $(\mu_1)_0\neq (\mu_2)_0$, i.e., the means are unequal in the true parameter $\theta_0$, and 
$r(\theta_0)$ is taken to be sufficiently small.)

We then calculate
\[
\left\|H(\theta; x)-H(\theta)\right\|
\leq \left\|H(\theta; x)-H(\theta_0;x)\right\| + \left\|H(\theta_0; x)-H(\theta_0)\right\| + \left\|H(\theta)-H(\theta_0)\right\|.\]
For the first term, for all $\theta\in\B(\theta_{0},r(\theta_{0}) )\cap\Theta$,
\begin{align*}
	&\left\|H(\theta; x)-H(\theta_0;x)\right\| \\
	&=\left\|\sum_{i=1}^n \left(g_0(x_i;\theta) - g_0(x_i;\theta_0)\right) +  x_i \left(g_1(x_i;\theta) - g_1(x_i;\theta_0)\right) + x_i^2\left(g_2(x_i;\theta) - g_2(x_i;\theta_0)\right)\right\|\\
	&\leq\sum_{i=1}^n \left\|g_0(x_i;\theta) - g_0(x_i;\theta_0)\right\| +  |x_i| \left\|g_1(x_i;\theta) - g_1(x_i;\theta_0)\right\| + x_i^2\left\|g_2(x_i;\theta) - g_2(x_i;\theta_0)\right\|\\
	&\leq\sum_{i=1}^n C_0'(\theta_0)r(\theta_0) +  |x_i| C_1'(\theta_0)r(\theta_0) + x_i^2C_2'(\theta_0)r(\theta_0)\\
	&\leq r(\theta_0) \left[ n\left(C_0'(\theta_0) + 0.5C_1'(\theta_0)\right)  + \sum_{i=1}^n x_i^2  \left(C_2'(\theta_0) + 0.5C_1'(\theta_0)\right) \right].
\end{align*}
Similarly, for the third term,
\[\left\|H(\theta)-H(\theta_0)\right\| \leq  r(\theta_0) \left[ n\left(C_0'(\theta_0) + 0.5C_1'(\theta_0)\right)  + \sum_{i=1}^n \E_{\theta_0}[X_i^2]  \left(C_2'(\theta_0) + 0.5C_1'(\theta_0)\right) \right].\]
By Cauchy--Schwarz, then,
\begin{multline*}
	\log 	\E_{\theta_{0}}\left[\exp\left\{\sup_{\theta\in\B(\theta_{0},r(\theta_{0}) )\cap\Theta}r(\theta_{0})^2\cdot \left\|H(\theta; X)-H(\theta)\right\|\right\}\right]\\
	\leq \frac{1}{2}\log 	\E_{\theta_{0}}\left[\exp\left\{2r(\theta_{0})^2\cdot \left\|H(\theta_0; X)-H(\theta_0)\right\|\right\}\right]\\
	+  \frac{1}{2}\log 	\E_{\theta_{0}}\left[\exp\left\{2\sup_{\theta\in\B(\theta_{0},r(\theta_{0}) )\cap\Theta}r(\theta_{0})^2\cdot \left( \left\|H(\theta; X)-H(\theta_0;X)\right\|  + \left\|H(\theta)-H(\theta_0)\right\| \right)\right\}\right]\\
	\leq \frac{1}{2}\log 	\E_{\theta_{0}}\left[\exp\left\{2r(\theta_{0})^2\cdot \left\|H(\theta_0; X)-H(\theta_0)\right\|\right\}\right] + c(\theta_0) \cdot nr(\theta)^3,
\end{multline*}
for an appropriate function $c(\theta_0)$, since the $X_i^2$'s are subexponential under $P_{\theta_0}$.

Next we bound the remaining term. Since the Hessian is a $5\times 5$ matrix, for any $c>0$ we have
\begin{align*}& \E_{\theta_{0}}\left[\exp\left\{c\cdot \left\|H(\theta_0; X)-H(\theta_0)\right\|\right\}\right]\\
	&\leq \E_{\theta_{0}}\left[\exp\left\{5c\cdot \left\|H(\theta_0; X)-H(\theta_0)\right\|_{\infty}\right\}\right]\\
	&= \E_{\theta_{0}}\left[\exp\left\{5c\cdot \max_{j=1,\dots,5}\max_{k=1,\dots,5}\max\left\{H(\theta_0; X)_{jk}-H(\theta_0)_{jk}, 
	H(\theta_0)_{jk}-H(\theta_0;X)_{jk}\right\}\right\}\right]\\
	&\leq  \sum_{j=1}^5\sum_{k=1}^5\E_{\theta_{0}}\left[\exp\left\{5c \left|H(\theta_0; X)_{jk}-H(\theta_0)_{jk}\right|\right\}\right]\\
	&\leq \sum_{j=1}^5\sum_{k=1}^5\E_{\theta_{0}}\left[\exp\left\{5c (H(\theta_0; X)_{jk}-H(\theta_0)_{jk})\right\}\right]\\
	& \qquad\qquad\qquad\qquad +\sum_{j=1}^5\sum_{k=1}^5\E_{\theta_{0}}\left[\exp\left\{5c (H(\theta_0)_{jk}-H(\theta_0;X)_{jk})\right\}\right].
\end{align*}
Now we handle each term individually. We have
\begin{align*}
	&\E_{\theta_{0}}\left[\exp\left\{5c (H(\theta_0; X)_{jk}-H(\theta_0)_{jk})\right\}\right]\\
	&=\E_{\theta_{0}}\left[\exp\left\{5c \sum_{i=1}^n\sum_{m=0}^2 \left[X_i^mg_m(X_i;\theta_0)_{jk} - \E_{\theta_0}[X_i^mg_m(X_i;\theta_0)_{jk}]\right] \right\}\right]\\
	&\leq \prod_{m=0}^2 \E_{\theta_{0}}\left[\exp\left\{15c \sum_{i=1}^n \left[X_i^mg_m(X_i;\theta_0)_{jk} - \E_{\theta_0}[X_i^mg_m(X_i;\theta_0)_{jk}]\right] \right\}\right]^{1/3}.
\end{align*}
Since $X_i^m$ is subexponential for each $m=0,1,2$ while $g_m(X_i;\theta_0)_{jk}$ is bounded, and the product of a bounded random variable
and a subexponential random variable is subexponential, we have
\[ \E_{\theta_{0}}\left[\exp\left\{15c \sum_{i=1}^n \left[X_i^mg_m(X_i;\theta_0)_{jk} - \E_{\theta_0}[X_i^mg_m(X_i;\theta_0)_{jk}]\right]\right\}\right] \leq e^{c^2 n c'_{m,jk}(\theta_0)}\]
assuming $c\leq c''_{m,jk}(\theta_0)$, for some positive-valued functions $c'_{m,jk},c''_{m,jk}$.
The same type of calculation holds for the terms of the form $\E_{\theta_{0}}\left[\exp\left\{5c (H(\theta_0)_{jk}-H(\theta_0;X)_{jk})\right\}\right]$,
for some positive-valued functions $\tilde{c}'_{m,jk},\tilde{c}''_{m,jk}$.
Combining everything,
\[
\E_{\theta_{0}}\left[\exp\left\{c\cdot \left\|H(\theta_0; X)-H(\theta_0)\right\|\right\}\right]\\
\leq  \sum_{j=1}^5\sum_{k=1}^5\prod_{m=0}^2 e^{\frac{1}{3}c^2 nc'_{m,jk}(\theta_0)} + \sum_{j=1}^5\sum_{k=1}^5\prod_{m=0}^2 e^{\frac{1}{3}c^2n \tilde{c}'_{m,jk}(\theta_0)},\]
for $0 < c <c''(\theta_0) =  \min_{m,j,k}\min\{c''_{m,jk}(\theta_0),\tilde{c}_{m,jk}''(\theta_0)\}$.
Letting $$c'(\theta_0) =  \max_{m,j,k}\max\{c''_{m,jk}(\theta_0),\tilde{c}_{m,jk}''(\theta_0)\},$$ then,
\[
\E_{\theta_{0}}\left[\exp\left\{c\cdot \left\|H(\theta_0; X)-H(\theta_0)\right\|\right\}\right]\\
\leq  50 e^{c^2n c'(\theta_0)}.\]
Choosing $c>r(\theta_0)^2$, then, by Jensen's inequality,
\begin{multline*} \E_{\theta_{0}}\left[\exp\left\{r(\theta_0)^2\cdot \left\|H(\theta_0; X)-H(\theta_0)\right\|\right\}\right]
	\leq  \E_{\theta_{0}}\left[\exp\left\{c\cdot \left\|H(\theta_0; X)-H(\theta_0)\right\|\right\}\right]^{r(\theta_0)^2/c}\\
	\leq ( 50 e^{c^2n c'(\theta_0)}) ^{r(\theta_0)^2/c}
	=  \exp\left\{ \frac{r(\theta_0)^2}{c}\log 50 + r(\theta_0)^2 c n c'(\theta_0)\right\}.\end{multline*}
Choosing $c = \sqrt{\frac{\log 50}{nc'(\theta_0)}}$, then, which (for sufficiently large $n$) satisfies $c>r(\theta_0)^2$ and $c<c''(\theta_0)$,
\[\E_{\theta_{0}}\left[\exp\left\{r(\theta_0)^2\cdot \left\|H(\theta_0; X)-H(\theta_0)\right\|\right\}\right]
\leq \exp\left\{ r(\theta_0)^2 \cdot 2\sqrt{n c'(\theta_0)\log 50}\right\}.\]
Combining everything, the 
assumption holds with any $r(\theta_0) = o(n^{-1/4})$ and $\epsilon(\theta_0)\gtrsim r(\theta_0)^2n^{1/2} + r(\theta_0)^3 n$.

\section{Experiment details}\label{sec:exp_detail}
For Example \ref{example:mixgaussian}, we use MCMC to generate the copies $\tilde{X}^{(m)}$; see details in Section \ref{sec:mcmc}. For Example \ref{example:isotonic} and \ref{example:sparse}, the conditional distribution is tractable, and we sample directly from the conditional distribution; see details in Section \ref{sec:con_distribution_detail}.

\subsection{Implementation details for Example \ref{example:mixgaussian} (Gaussian mixture model)} 
\label{sec:mcmc}
For the Gaussian mixture model, the copies $\tilde{X}^{(m)}$ are sampled via MCMC. Here we give the details for this process.

When sampling directly from $p_{\hat\theta}(\cdot\mid \hat\theta, \hat{g})$ is  infeasible, \cite{barber2020testing} discusses two schemes for constructing copies with MCMC sampling: the Hub-and-spoke sampler and the Permuted serial sampler. In our simulation for Example \ref{example:mixgaussian}, aCSS (with and without constraints) is run with the hub-and-spoke sampler.
Given $X$ and $\hat\theta, \hat{g}$, we sample the copies as follows:
\begin{itemize}
	\item Initialize at $X$, and run the Markov chain (specified below) for $L$ steps to define the ``hub" $\tilde{X}^{*}$.
	\item Independently for $m = 1, \dots, M$, initialize at  $\tilde{X}^{*}$ and run the Markov chain  (specified below) for $L$ steps to define the ``spoke'' $\tilde{X}^{m}$.
\end{itemize}		
Similar to 	\cite{barber2020testing}, we can use use the Metropolis–Hastings (MH) to construct an efficient sampling scheme. Given $\hat\theta$, the reversible MCMC is   given by the following:
\begin{itemize}
	\item  Starting at state $x'$, generate a proposal $x$ according to a properly chosen proposal distribution $q_{\hat\theta}(x\mid x')$.
	\item With probability $A_{\hat\theta}(x \mid x') = \min\left\{1, \frac{q_{\hat\theta}(x'\mid x) }{q_{\hat\theta}(x \mid x')}\frac{p_{\hat\theta}(x\mid \hat\theta, \hat{g})}{p_{\hat\theta}(x'\mid \hat\theta, \hat{g})}\right\}$, set the next state to equal $x$. Otherwise, the next state is set to equal $x'$.
\end{itemize}		
Next, we will describe the proposal distribution and MH acceptance probability; we also refer to \cite[Appendix D.2]{barber2020testing} for more details.

\subsubsection{Proposal distribution $q_{\hat\theta}(x\mid x')$}
In Example \ref{example:mixgaussian}, the model $P_{\theta}$ is a product distribution with density 
$$
f_{\theta}(x) = \prod_{i=1}^{n}f_{\theta}^{i}(x_{i}).
$$
We then use the same proposal distribution as \cite[Examples 1,2,4]{barber2020testing}.
For  $s\in[n]$, define $q_{\hat\theta}(x |x' )$ as follows: 
\begin{itemize}
	\item Draw a subset $\mathcal{S}\subseteq \{1, \dots, n\}$ of size $s$, uniformly at random.
	\item For each $i=1, \dots, n,$
	\begin{itemize}
		\item Set $x_{i} = x_{i}'$, if $i\notin S$,
		\item Draw $x_{i} \sim f_{\hat\theta}^{(i)}$, if $i\in S$.
	\end{itemize} 
\end{itemize}
Here  $s$  controls the tradeoff between two goals: (1) the acceptance probability $A_{\hat\theta}(x | x' )$ should not be too close to zero; (2) the proposed state should not be too similar to the previous state. 
Note that we can tune this MCMC hyperparameter after looking at $\hat\theta$ without violating any of our theoretical assumptions. We can then choose $s$ based on the following simulation:
\begin{itemize}
	\item Let $\theta_{0}^{\textnormal{sim}} = \hat\theta$.
	\item Draw $X^{\textnormal{sim}}\sim P_{\theta_{0}^{\textnormal{sim}}}$, $W\sim \mathcal{N}(0, \frac{1}{d}\I_{d})$; \\calculate $\hat\theta^{\textnormal{sim}} = \hat{\theta}\left(X^{\textnormal{sim}}, W\right)$, and $\hat{g}^{\textnormal{sim}} = \nabla\L(\hat\theta^{\textnormal{sim}}; X^{\textnormal{sim}}, W)$.
	\item For each candidate   of $s$  , run one step of Metropolis-Hasting initialized at $X^{\textnormal{sim}} $ to generate $X^{\textnormal{new}} $.
	\item Repeat for 100 draws of $X^{\textnormal{sim}} $, discarding any draws for which $\hat\theta^{\textnormal{sim}} $ is not a SSOSP, to get an average acceptance probability $\bar{A}_{s}$  . Among all values of $s$ where $\bar{A}_{s}\ge 0.05$, choose $s$ that maximizes $s\bar{A}_{s}$.
\end{itemize}
Note that  this choice of $s$ only depends on $\hat\theta$, and completing our $\theta$-dependent definition of the proposal distribution $q_{\hat\theta}(x\mid x')$. Then we choose $L = \min\{2000, \frac{2n}{s\hat{A}_{s}}\}$ to   ensure that  most entries will be resampled within $L$ steps. 

\subsubsection{MH acceptance probability}
Given $\hat\theta, \hat{g}$,  and a properly chosen proposal distribution $q_{\hat\theta}(x\mid x')$, the MH acceptance probability
$A_{\hat\theta}(x \mid  x' )$ can be written as
$$
A_{\hat\theta}(x \mid  x ')  = \min\left\{1, \frac{q_{\hat\theta}(x' \mid  x) }{q_{\hat\theta}(x \mid  x')}\frac{p_{\hat\theta}(x\mid  \hat\theta, \hat{g})}{p_{\hat\theta}(x'\mid   \hat\theta, \hat{g})}\right\},
$$
where 
\begin{equation*} 
	p_{\hat\theta}(x| \hat\theta, \hat{g})\propto f(x; \hat\theta ) \exp\left\{-\frac{\|\hat{g} - \nabla\L(\hat\theta; x)\|^{2} }{2\sigma^{2}/d}\right\}\det\left(U_{\A(\hat\theta)}^\top \nabla^{2}_{\theta}\L(\hat\theta; x)U_{\A(\hat\theta)}\right) 
	\id_{x\in\X_{\hat\theta,\hat{g}}}
\end{equation*}
The ratio  in the MH acceptance probability without the indicator variables are straightforward to calculate. The ratio with indicator variables
$	\id_{x\in\X_{\hat\theta,\hat{g}}}/	\id_{x'\in\X_{\hat\theta,\hat{g}}}$ requires more careful consideration. First, we will always have $	\id_{x'\in\X_{\hat\theta,\hat{g}}}= 1$ since $x'$ is sampled from \eqref{eqn:conditional_density} with $	x'\in\X_{\hat\theta,\hat{g}}$. To check $\id_{x\in\X_{\hat\theta,\hat{g}}}$, we have 
\begin{equation*}
	\begin{split}
		\id_{x\in\X_{\hat\theta,\hat{g}}} &= \id\left\{\exists w\in\R^{d}\ \textnormal{s.t.}\  \hat\theta = \hat{\theta}(x, w) \ \textnormal{is a SSOSP of \eqref{eqn:def_thetahat}}, \ \textnormal{and}\ \hat{g} = \nabla\L(\hat\theta; x, w)\right\}\\
		& = \id\left\{ \hat\theta\left(x, \frac{\hat{g} - \nabla_{\theta}\L(\hat\theta; x)}{\sigma}\right) = \hat\theta, \textnormal{and}\  U_{\A(\hat\theta)}^\top \nabla^{2}_{\theta}\L(\hat\theta; x)U_{\A(\hat\theta)}\succ 0\right\}.
	\end{split}
\end{equation*}
This means given proposed $x$, we only need to verify (1) $U_{\A(\hat\theta)}^\top \nabla^{2}_{\theta}\L(\hat\theta; x)U_{\A(\hat\theta)}\succ 0$ and (2) the algorithm $\hat\theta\left(x, \frac{\hat{g} - \nabla_{\theta}\L(\hat\theta; x)}{\sigma}\right)$  returns  value $\hat{\theta}$.

\subsection{Implementation details for Examples \ref{example:isotonic} (isotonic regression) and \ref{example:sparse} (sparse regression)} 
\label{sec:con_distribution_detail}

In this section, we derive the sampling distribution for the copies $\tilde{X}^{(m)}$ for the two Gaussian linear model examples.

Recall that the objective function $\L(\theta; x, w)$ is defined as 
$$
\L(\theta; x, w) = \frac{1}{2\nu^2}\|x - Z\theta\|^2+ \mathcal{R}(\theta) +\sigma w^\top \theta,
$$ 
and
$$
\left\{ \begin{array}{l} 
	\hat\theta = \hat\theta(X, W),\\ 
	\hat{g} =   \frac{1}{\nu^2}Z^\top (Z\hat\theta - X)+ \nabla_{\theta}\mathcal{R}(\hat\theta)+ \sigma  W,
\end{array}\right.
$$
where $ \hat\theta(X, W)$ is the minimizer of $\L(\theta; X, W)$ subject to arbitrary linear constraints or $\ell_{1}$ penalty. Note that the original aCSS is a special case of the constrained aCSS with no constraints and $\hat{g} = 0$.
When $\L(\theta;x,w)$ is strictly convex (like if we add ridge penalty), a unique SSOSP exists (and is computationally efficient
to find), and we can then define $\hat\theta(x,w)$ to be equal
to this unique SSOSP. Based on the conditional density derived in  \eqref{eqn:conditional_density},
we can efficiently compute the conditional distribution $p_{\theta_{0}}(\cdot\mid \hat\theta, \hat{g})$ as follows 
\begin{equation*}  
	X\mid \hat\theta,\hat{g} \sim \mathcal{N}\left(Z\hat\theta + \left(\I_{n} + \frac{d}{\sigma^2\nu^2}ZZ^\top \right)^{-1}Z(\theta_{0} - \hat\theta + \frac{d}{\sigma^2}(\nabla_{\theta}\mathcal{R}(\hat\theta) - \hat{g})), \nu^2\left(\I_{n} + \frac{d}{\sigma^2\nu^2}ZZ^\top\right)^{-1}\right).
\end{equation*} 
The plug-in conditional distribution  $ \tilde{X}$, i.e., $p_{\hat\theta}(\cdot\mid \hat\theta, \hat{g})$, is  
\begin{equation*} 
	\tilde{X} \sim \mathcal{N}\left(Z\hat\theta + \left(\I_{n} + \frac{d}{\sigma^2\nu^2}ZZ^\top \right)^{-1}Z \frac{d}{\sigma^2}(\nabla_{\theta}\mathcal{R}(\hat\theta) - \hat{g}), \nu^2\left(\I_{n} + \frac{d}{\sigma^2\nu^2}ZZ^\top\right)^{-1}\right). 
\end{equation*}

\begin{itemize}
	\item	In  Example \ref{example:isotonic},   $\mathcal{R}(\theta) = 0$, $Z = \I_{n}$ and $\nu^2 = 1$. Details of sampling using the aCSS method, with and without constraints, are as follows: 
	
	For \cite{barber2020testing}'s aCSS method, $\hat\theta$ is computed via perturbed and unconstrained maximum likelihood estimation,
	\[\hat\theta = \hat\theta_{\textnormal{OLS}} = \argmin_{\theta\in\R^n} \left\{\frac{1}{2}\|X - \theta\|^2  + \sigma W^\top \theta\right\} = X - \sigma W, \]
	and then the copies $\tilde{X}^{(m)}$ are sampled directly from $p_{\hat\theta}(\cdot\mid \hat\theta)$ via the distribution
	\[\tilde{X}^{(m)}\stackrel{\textnormal{i.i.d.}}{\sim} \mathcal{N}\left(\hat\theta, \left(1+\frac{n}{\sigma^2}\right)^{-1}\I_n\right).\]
	
	For our proposed constrained aCSS method, $\hat\theta$ is computed with the isotonic constraint,
	\[\hat\theta = \hat\theta_{\textnormal{iso}}  = \argmin_{\substack{\theta\in\R^n\\\theta_1\leq \dots \leq \theta_n}} \left\{\frac{1}{2}\|X - \theta\|^2  + \sigma W^\top \theta\right\}, \]
	the gradient is given by
	\[\hat g = \hat\theta - X + \sigma W,\]
	and then the copies $\tilde{X}^{(m)}$ are sampled directly from $p_{\hat\theta}(\cdot\mid \hat\theta,\hat g)$ via the distribution
	\[\tilde{X}^{(m)}\stackrel{\textnormal{i.i.d.}}{\sim} \mathcal{N}\left(\hat\theta- \frac{n/\sigma^2}{1+n/\sigma^2}\hat{g}, \left(1+\frac{n}{\sigma^2}\right)^{-1}\I_n\right).\]

	\item In Example \ref{example:sparse}, we choose $\mathcal{R}(\theta) = \frac{\lambda_{\textnormal{ridge}}}{2}\|\theta\|^2$ as a ridge penalization with $\lambda_{\textnormal{ridge}} = 0.01$, $\nu^2 = 1$. Details of sampling using the aCSS method, with and without an $\ell_{1}$ penalty, are as follows:
	
	For \cite{barber2020testing}'s aCSS method, we will use a ridge regularizer.
	The method is then defined by setting
	\begin{equation*}
		\begin{split}
			\hat\theta = \hat\theta_{\textnormal{ridge}} &= \argmin_{\theta\in\R^d} \left\{\frac{1}{2}\|X - Z\theta\|^2 + \frac{\lambda_{\textnormal{ridge}}}{2}\|\theta\|^{2} + \sigma W^\top \theta\right\}\\
			&=\left(\lambda_{\textnormal{ridge}}\I_{d} + Z^{T}Z \right)^{-1}(Z^{T}X - \sigma W),
		\end{split}
	\end{equation*} 
	and then sampling the copies $\tilde{X}^{(m)}$ directly from $p_{\hat\theta}(\cdot\mid \hat\theta)$ via the distribution
	\[\tilde{X}^{(m)}\stackrel{\textnormal{i.i.d.}}{\sim} \mathcal{N}\left(Z\hat\theta + \frac{\lambda_{\textnormal{ridge}} d}{\sigma^2}\left(\I_{n}+\frac{d}{\sigma^2}ZZ^\top \right)^{-1}Z\hat\theta, \left(\I_{n} + \frac{d}{\sigma^2}ZZ^\top \right)^{-1}\right).\]
	
	For our proposed penalized aCSS method, in order to be more comparable to aCSS, we also add the regularizer $R(\theta)$. This means that our estimator
	is given by the elastic net, incorporating both $\ell_1$ and $\ell_2$ penalization:
	\[\hat\theta=\hat\theta_{\textnormal{elastic-net}} =  \argmin_{\theta\in\R^d} \left\{\frac{1}{2}\|X - Z\theta\|^2 + \frac{\lambda_{\textnormal{ridge}}}{2}\|\theta\|^{2} + \lambda\|\theta\|_1+ \sigma W^\top \theta\right\},\]
	with $\lambda = 2$, and the gradient is then computed as 
	\[\hat{g} =Z^{T}(Z\hat\theta-X)+ \sigma W +\lambda_{\textnormal{ridge}} \hat\theta.\]
	We then sample the copies $\tilde{X}^{(m)}$ directly from $p_{\hat\theta}(\cdot\mid \hat\theta,\hat g)$ via the distribution
	\[\tilde{X}^{(m)}\stackrel{\textnormal{i.i.d.}}{\sim} \mathcal{N}\left(Z\hat\theta + \frac{d}{\sigma^2}\left(\I_{n}+\frac{d}{\sigma^2}ZZ^\top \right)^{-1}Z(\lambda_{\textnormal{ridge}} \hat\theta - \hat g), \left(\I_{n} + \frac{d}{\sigma^2}ZZ^\top \right)^{-1}\right).\]
\end{itemize}

Besides results in the main paper, to better understand the difference in performance in terms of Type I error rate, in Figure~\ref{fig:regression2} we show 
the Type I error as a function of the parameter $\sigma$. For both settings,  we see that aCSS suffers a rapid increase in Type I error rate, thus necessitating a very small value of $\sigma$
to maintain validity, while constrained or penalized aCSS maintains Type I error control across a broad range of values of $\sigma$. 
Finally, Figure~\ref{fig:regression3} illustrates the issue of Type I error in more detail for the specific choice $\sigma=7$ for both examples
(chosen to be large enough so that the methods can achieve substantial power). This figure shows a highly nonuniform distribution of
the p-values for aCSS, in contrast to the approximately uniform distribution for constrained or penalized aCSS.
\begin{figure} 	
	\centering   
	\includegraphics[width=0.45\textwidth]{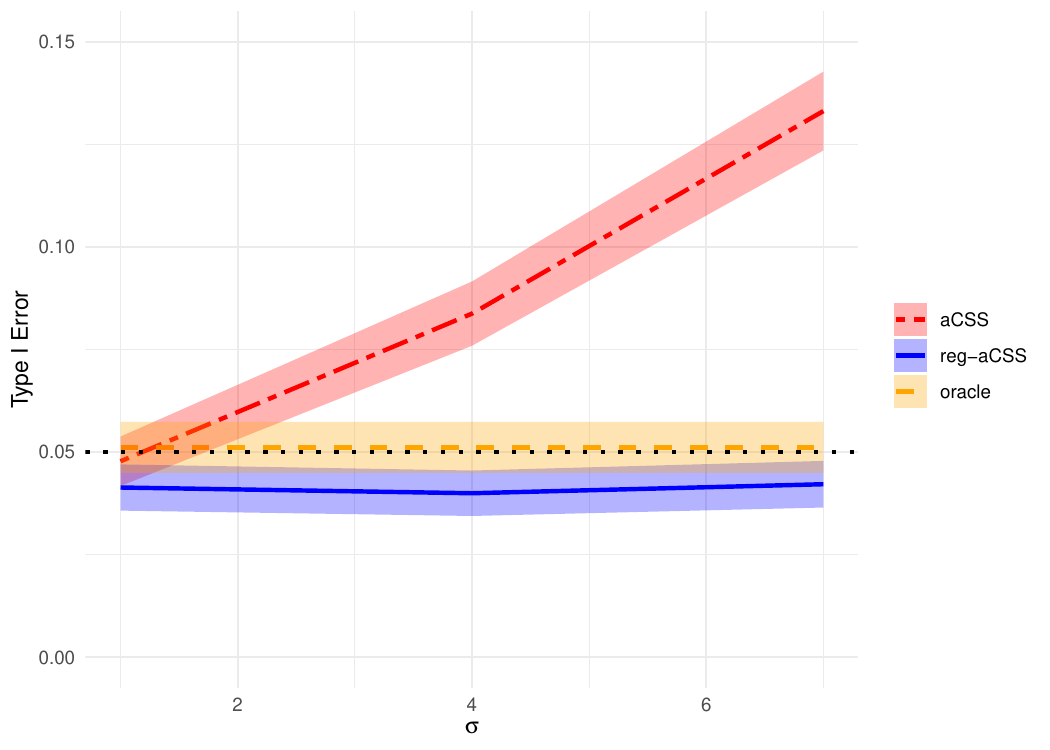} 
	\includegraphics[width=0.45\textwidth]{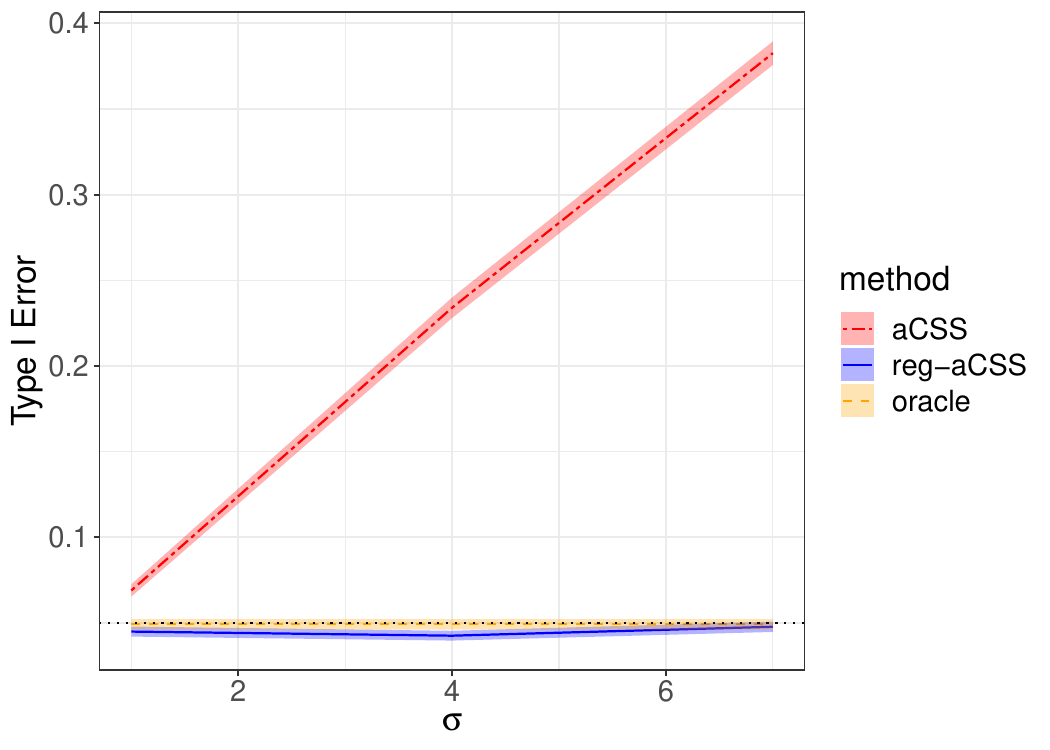}  
	\caption{ Type I error rate of  aCSS, regularized (i.e., constrained or penalized) aCSS (denoted as reg-aCSS in the plot), and the oracle method, for isotonic regression (left) and sparse regression (right),  with different values of the parameter $\sigma$, over 5000 independent trials. The dotted red line denotes the nominal 5\% level (i.e., $\alpha=0.05$). The shaded bands denote standard error  for each method.}
	\label{fig:regression2}
\end{figure}
\begin{figure} 	
	\centering   
	\includegraphics[width=0.45\textwidth]{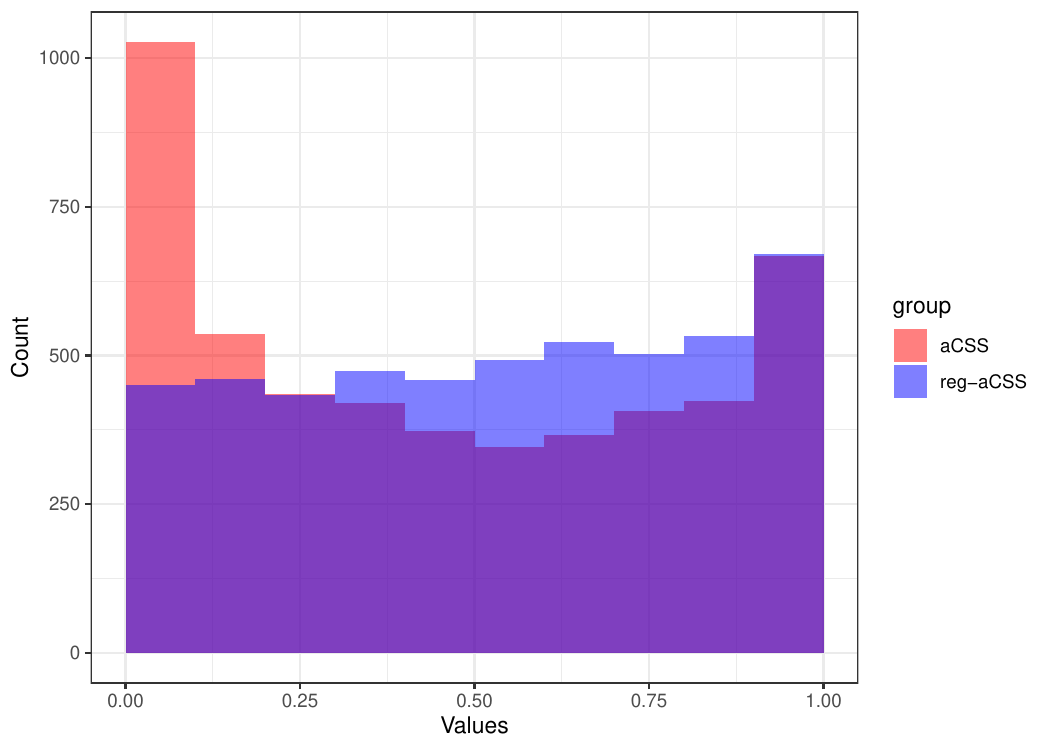}
	\includegraphics[width=0.45\textwidth]{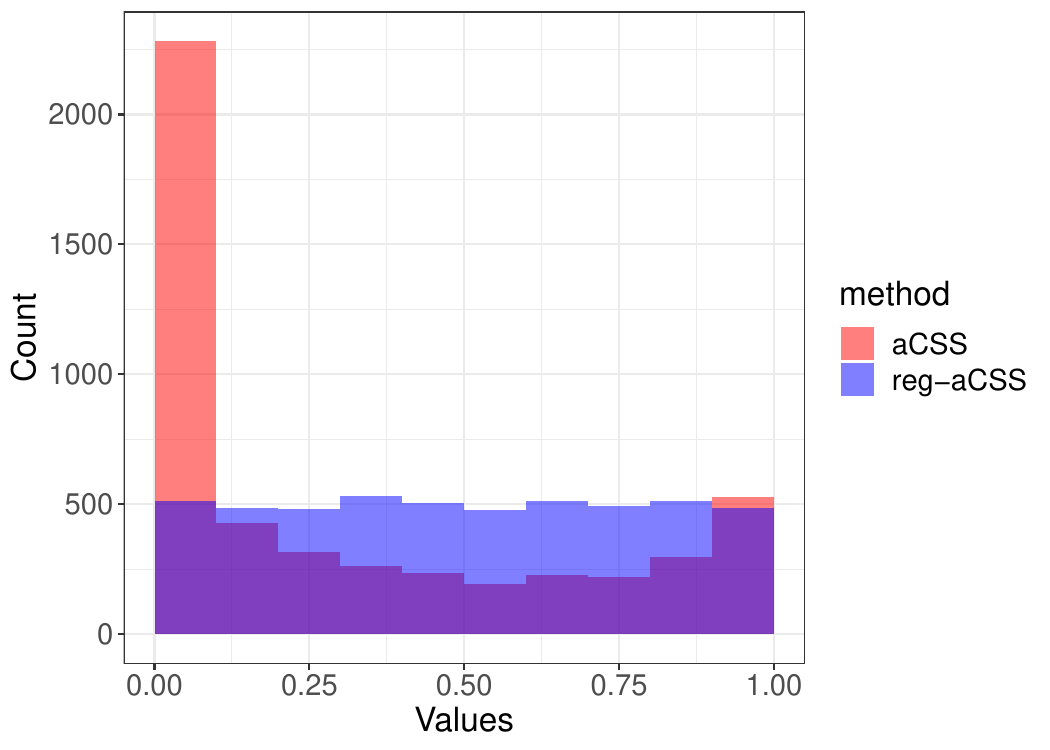}
	\caption{Histogram of p-values under the null, for   aCSS and for regularized (i.e.,  constrained or penalized) aCSS, for isotonic regression (left) and sparse regression (right), over 5000 independent trials.
		The parameter $\sigma$ is chosen as $\sigma=7$ for both examples.}
	\label{fig:regression3}
\end{figure}

\subsection{Extension for Gaussian linear model with unknown  $\nu$}\label{sec:sample_unknown_var}
Recall the 	Gaussian linear model
\[X\sim  \mathcal{N}(Z\theta, \nu^2\I_{n}).\]
We now consider the setting where the noise variance $\nu^2 $ is unknown.  To complement our earlier analysis under known noise variance, we revisit Example \ref{example:isotonic} and  \ref{example:sparse}. For clarity, we restate both examples with updated labels.

\begin{example}[Isotonic regression with unknown variance]\label{example:iso_unknown_nu}
	Assume that the true parameter $\theta_0$ satisfies the isotonic (monotonic non-decreasing) constraint:
	\[(\theta_0)_{1}\le\cdots\le (\theta_0)_{n}.\]
	We are given a noisy observation $X\in\R^{n}$, with $X\sim  \mathcal{N}(\theta_0, \nu^2\I_{n})$ for some unkown $\nu$. This model is a special case of the Gaussian linear model with $d = n$ and $Z = \I_n$. 
\end{example} 

\begin{example}[Sparse regression with unknown variance]\label{example:sparse_unknown_nu}
	Let $d>n$, and 	let $Z\in\R^{d\times n}$ be a fixed covariate matrix. We assume the model	
	\[	X\sim  \mathcal{N}(Z\theta_0, \nu^2\I_{n}),\]
	for an unknown noise level $\nu^2$. We further assume that the true parameter $\theta_0$ is sparse,  and estimate it using the Lasso—that is, by solving a penalized optimization problem with an $\ell_1$ regularization term.
\end{example}  

\subsubsection{Simulation: setting}
All components of the simulation setup remain the same as in Section \ref{sec:simu_setting_gaussian}, except that the sampling step in the testing procedure is modified to account for the additional uncertainty arising from the unknown noise variance $\nu^2$. We now derive the modified sampling distribution for the copies $\tilde{X}^{(m)}$ in the two  examples.

As discussed in Remark \ref{rmk:unknown_var}, for general gaussian linear model, when $\nu$ is unkown we solve for $\hat\theta$ by  
\[ \argmin_{\theta\in\R^d}\left\{ \frac{1}{2}\|X-Z\theta\|^2 +  R(\theta)+ \sigma W^\top  \theta\right\}\] 
subject to arbitrary linear constraints or $\ell_{1}$ penalty, and compute the gradient as
\[\hat g =  Z^\top (Z\hat\theta - X)+ \nabla_\theta R(\hat\theta)+ \sigma  W.\]
We further compute $\hat\nu$ as
\[\hat\nu = \sqrt{\frac{1}{n}\|X - Z\hat\theta\|^2}.\] 
The  conditional density $p_{\hat\theta, \hat\nu}(\cdot\mid \hat\theta, \hat g, \hat\nu)$ is  proportional to 
\[\exp\{-\frac{1}{2}(x - \mu)^{\top}\Sigma^{-1}(x-\mu)\}\cdot \id_{\{\|x - Z\hat\theta\|^2 = n\hat{\nu}^2\}},\]
where
\[\mu = Z\hat\theta + \frac{d}{\sigma^2}\left(\frac{1}{\hat\nu^2}\I_{n}+\frac{d}{\sigma^2}ZZ^\top \right)^{-1}Z(\nabla_\theta R(\hat\theta) - \hat{g}), \ \Sigma =\left(\frac{1}{\hat\nu^2}\I_{n} + \frac{d}{\sigma^2}ZZ^\top \right)^{-1}. \]

\begin{itemize}  
	\item	In  Example \ref{example:iso_unknown_nu}, $\mathcal{R}(\theta) = 0$, $Z = \I_{n}$. 	Details of sampling using the aCSS method, with and without constraints, are as follows: 
	
	For \cite{barber2020testing}'s aCSS method, $\hat\theta$ is computed via perturbed and unconstrained maximum likelihood estimation,
	\[\hat\theta = \hat\theta_{\textnormal{OLS}} = \argmin_{\theta\in\R^n} \left\{\frac{1}{2}\|X - \theta\|^2  + \sigma W^\top \theta\right\} = X - \sigma W, \]
	and 	compute $\hat\nu$ as
	\[\hat\nu = \sqrt{\frac{1}{n}\|X - Z\hat\theta\|^2}.\]
	Then the copies $\tilde{X}^{(m)}$ are sampled  via  
	\[\exp\{-\frac{1}{2}(x - \mu)^{\top}\Sigma^{-1}(x-\mu)\}\cdot \id_{\{\|x - \hat\theta\|^2 = n\hat{\nu}^2\}},\]
	with $$\mu = \hat\theta, \ \Sigma = \left(\frac{1}{\hat\nu^2}+\frac{n}{\sigma^2}\right)^{-1}\I_n.$$
	
	For our proposed constrained aCSS method, $\hat\theta$ is computed with the isotonic constraint,
	\[\hat\theta = \hat\theta_{\textnormal{iso}}  = \argmin_{\substack{\theta\in\R^n\\\theta_1\leq \dots \leq \theta_n}} \left\{\frac{1}{2}\|X - \theta\|^2  + \sigma W^\top \theta\right\}, \]
	the gradient is given by
	\[\hat g = \hat\theta - X + \sigma W,\]
	and compute $\hat\nu$ as
	\[\hat\nu = \sqrt{\frac{1}{n}\|X - Z\hat\theta\|^2}.\]
	Then the copies $\tilde{X}^{(m)}$ are sampled via  
	\[\exp\{-\frac{1}{2}(x - \mu)^{\top}\Sigma^{-1}(x-\mu)\}\cdot \id_{\{\|x - \hat\theta\|^2 = n\hat{\nu}^2\}},\]
	with $$\mu = \hat\theta- \frac{n/\sigma^2}{1/\hat\nu^2+n/\sigma^2}\hat{g}, \ \Sigma = \left(\frac{1}{\hat\nu^2}+\frac{n}{\sigma^2}\right)^{-1}\I_n.$$

	\item 	In Example \ref{example:sparse_unknown_nu}, we still choose $\mathcal{R}(\theta) = \frac{\lambda_{\textnormal{ridge}}}{2}\|\theta\|^2$ as a ridge penalization with $\lambda_{\textnormal{ridge}} = 0.01$. 
	Details of sampling using the aCSS method, with and without an $\ell_{1}$ penalty, are as follows:
	
	For \cite{barber2020testing}'s aCSS method, we will use a ridge regularizer.
	The method is then defined by setting
	\begin{equation*}
		\begin{split}
			\hat\theta = \hat\theta_{\textnormal{ridge}} &= \argmin_{\theta\in\R^d} \left\{\frac{1}{2}\|X - Z\theta\|^2 + \frac{\lambda_{\textnormal{ridge}}}{2}\|\theta\|^{2} + \sigma W^\top \theta\right\}\\
			&=\left(\lambda_{\textnormal{ridge}}\I_{d} + Z^{T}Z \right)^{-1}(Z^{T}X - \sigma W),
		\end{split}
	\end{equation*} 
	and 	compute $\hat\nu$ as
	\[\hat\nu = \sqrt{\frac{1}{n}\|X - Z\hat\theta\|^2}.\]
	Then sampling the copies $\tilde{X}^{(m)}$ via 
	\[\exp\{-\frac{1}{2}(x - \mu)^{\top}\Sigma^{-1}(x-\mu)\}\cdot \id_{\{\|x - Z\hat\theta\|^2 = n\hat{\nu}^2\}},\]
	with $$\mu =Z\hat\theta + \frac{\lambda_{\textnormal{ridge}} d}{\sigma^2}\left(\frac{1}{\hat\nu^2}\I_{n}+\frac{d}{\sigma^2}ZZ^\top \right)^{-1}Z\hat\theta, \ \Sigma = \left(\frac{1}{\hat\nu^2}\I_{n} + \frac{d}{\sigma^2}ZZ^\top \right)^{-1}.$$
	
	For our proposed penalized aCSS method, in order to be more comparable to aCSS, we also add the regularizer $R(\theta)$. This means that our estimator
	is given by the elastic net, incorporating both $\ell_1$ and $\ell_2$ penalization:
	\[\hat\theta=\hat\theta_{\textnormal{elastic-net}} =  \argmin_{\theta\in\R^d} \left\{\frac{1}{2}\|X - Z\theta\|^2 + \frac{\lambda_{\textnormal{ridge}}}{2}\|\theta\|^{2} + \lambda\|\theta\|_1+ \sigma W^\top \theta\right\},\]
	with $\lambda = 2$, and the gradient is then computed as 
	\[\hat{g} =Z^{T}(Z\hat\theta-X)+ \sigma W +\lambda_{\textnormal{ridge}} \hat\theta,\]
	and 	compute $\hat\nu$ as
	\[\hat\nu = \sqrt{\frac{1}{n}\|X - Z\hat\theta\|^2}.\]
	Then sampling the copies $\tilde{X}^{(m)}$ via 
	\[\exp\{-\frac{1}{2}(x - \mu)^{\top}\Sigma^{-1}(x-\mu)\}\cdot \id_{\{\|x - Z\hat\theta\|^2 = n\hat{\nu}^2\}},\]
	with $$\mu =Z\hat\theta + \frac{d}{\sigma^2}\left(\frac{1}{\hat\nu^2}\I_{n}+\frac{d}{\sigma^2}ZZ^\top \right)^{-1}Z(\lambda_{\textnormal{ridge}} \hat\theta - \hat g), \ \Sigma = \left(\frac{1}{\hat\nu^2}\I_{n} + \frac{d}{\sigma^2}ZZ^\top \right)^{-1}.$$

\end{itemize}

\subsubsection{Sampling from the constrained Gaussian distribution}
Next, we discuss how to sample from the constrained Gaussian distribution
\[\exp\left\{-\frac{1}{2}(x - \mu)^{\top}\Sigma^{-1}(x-\mu)\right\}\cdot \id_{\{\|x - Z\hat\theta\|^2 = n\hat{\nu}^2\}}.\]
To enforce the constraint  $\|x - Z\hat\theta\|^2 = n\hat{\nu}^2$, we consider the following transformation:
\begin{itemize}
	\item Define the centered vector $x_{\textnormal{center}} = x-Z\hat\theta$
	\item Represent $x_{\textnormal{center}}$ using spherical (hyperspherical) coordinates $\phi = (\phi_{1}, \dots, \phi_{n-1})\in [0, \pi]^{n-2}\times[0, 2\pi)$ as follows
	\[	\left\{ \begin{array}{l} 
		\phi_{1} =\arccos\left(\frac{x_{\textnormal{center}}[1]}{\sqrt{n\hat\nu^2}}\right),\\
		\phi_{i} =\arccos\left(\frac{x_{\textnormal{center}}[i]}{\sqrt{n\hat\nu^2}\prod_{k=1}^{i-1}sin(\phi_{k})}\right), i = 2, \dots n-1.\\
	\end{array}\right.	
	\]   
	\item Adjust the last angle $\phi_{n-1} = 2\pi -  \phi_{n-1}$, if $x_{\textnormal{center}}[n]/\prod_{k=1}^{n-2}sin(\phi_{k})<0$.
\end{itemize}
The target density in angular coordinates can be computed using a change of variables:
\[f(\phi_{1}, \dots, \phi_{n-1}) \propto \phi(x_{\textnormal{center}}; \mu-Z\hat\theta, \Sigma)\cdot \prod_{i=1}^{n-2}\sin^{n-1-i}(\phi_i),\]
where $\phi$ denotes the multivariate Gaussian density.

We then sample copies  in spherical (hyperspherical) coordinates $\phi$ via the same hub-and-spoke sampler and MH described in Section \ref{sec:mcmc}. The proposal distribution $q_{\hat\theta, \hat\nu}(\phi\mid \phi')$ simulates a small move from $x'$ to $x$  on the hypersphere as follows:
\begin{itemize}
	\item For  a chosen $s\in[n-1]$, draw a subset $\mathcal{S}\subseteq \{1, \dots, n-1\}$ of size $s$, uniformly at random.
	\item For each $i=1, \dots, n-1,$
	\begin{itemize}
		\item Set $\phi_{i} = \phi_{i}'$, if $i\notin S$,
		\item Draw $\phi_{i} = \phi_{i}' + \mathcal{U}(-\delta,\delta)$, if $i\in S$,
		\item  Set $\phi_{i} \leftarrow \phi_{i}\mod (2\pi-\id_{\{i<d-1\}}\pi)$.
	\end{itemize} 
\end{itemize}

Finally, we transform the sampled polar coordinates  $\tilde\phi$ back to  Cartesian coordinates:  
\[	\left\{ \begin{array}{l} 
	\tilde{X}_{\textnormal{center}} [1] = \sqrt{n\hat\nu^2} cos(\tilde\phi_{1}),\\
	\tilde{X}_{\textnormal{center}} [i] = \sqrt{n\hat\nu^2} cos(\tilde\phi_{i})\prod_{k=1}^{i-1}sin(\tilde\phi_{k}), i = 2, \dots, n-1,\\
	\tilde{X}_{\textnormal{center}} [n] = \sqrt{n\hat\nu^2} \prod_{k=1}^{n-1}sin(\tilde\phi_{k}),
\end{array}\right.	
\]   
and recover the final sample as  $\tilde{X} = \tilde{X}_{\textnormal{center}} + Z\hat\theta$.

In our simulations, we set $\delta = \pi/20$, $s = 2$ when sampling both aCSS and regularized aCSS copies across both examples. The results are shown in Figure~\ref{fig:nu_mcmc}. The Type I error of regularized aCSS is still better controlled compared to the unconstrained version. Compared to the known variance case, the results are nearly identical in the isotonic example, whereas both the Type I error and power increase in the sparse example. One possible explanation is that the estimation error in $\hat\nu = \sqrt{|X - Z\hat\theta|^2/n}$ is substantial, leading to greater uncertainty than in the known variance setting and, consequently, a higher Type I error.
\begin{figure} 	
	\centering  
	\includegraphics[width=0.45\textwidth]{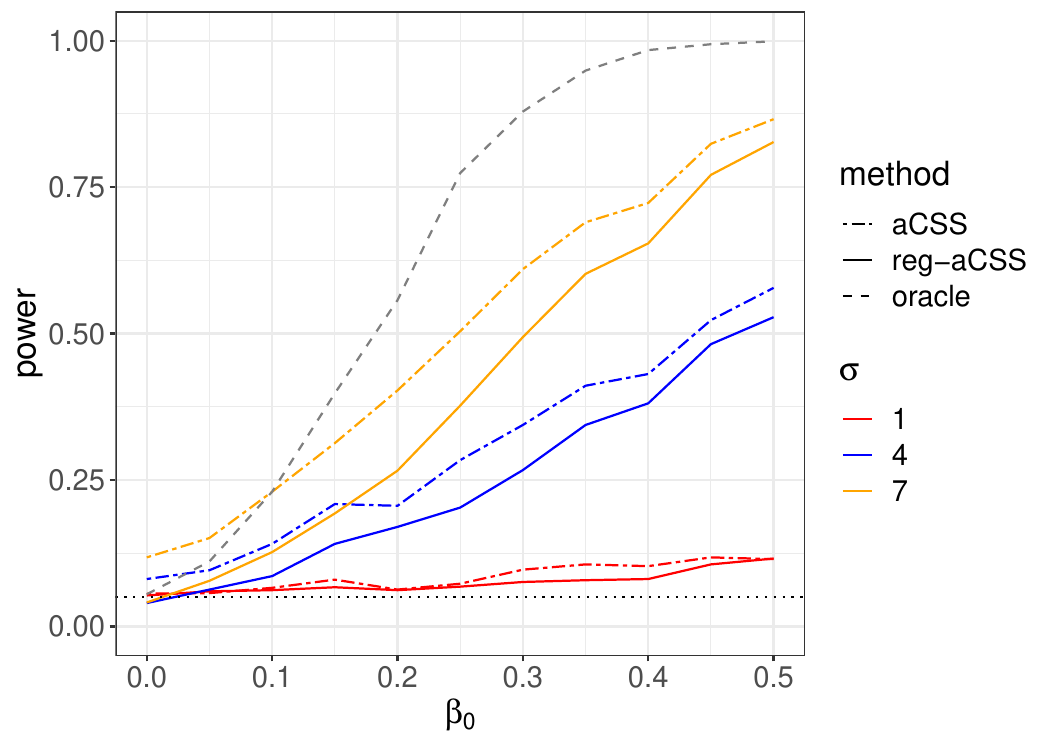}  	
	\includegraphics[width=0.45\textwidth]{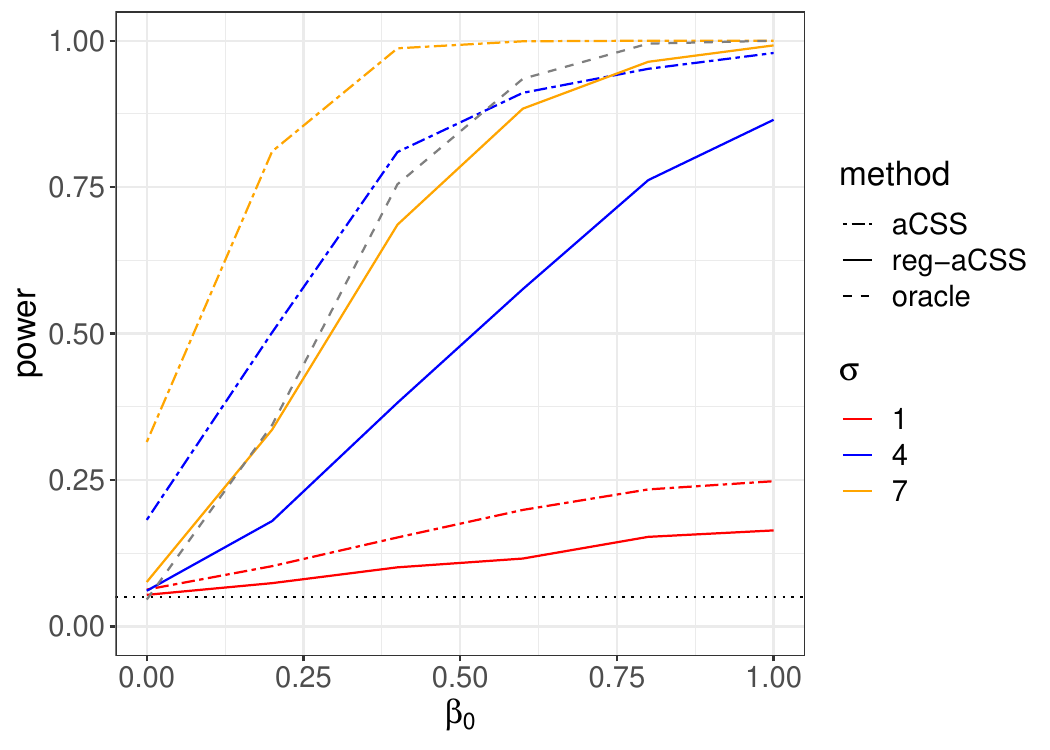}   
	\caption{unknown $\nu$. Power of  aCSS, regularized (i.e., constrained or penalized) aCSS (denoted as reg-aCSS in the plot), and the oracle method, for isotonic regression (left) and sparse regression (right),  with different values of the parameter $\sigma$, over 1000 independent trials. The dotted red line denotes the nominal 5\% level (i.e., $\alpha=0.05$). For both settings, $\beta_0=0$ corresponds to the null hypothesis being true.}
	\label{fig:nu_mcmc}
	
\end{figure}

 \bibliographystyle{plainnat} 
\bibliography{reference}    

\end{document}